\documentclass[letterpaper]{article}
\usepackage{geometry}
\usepackage{setspace}

\usepackage[doi=true]{achemso}
\setkeys{acs}{articletitle = true}
\usepackage{chemformula} 
\mciteErrorOnUnknownfalse
\usepackage{float}
\newfloat{scheme}{htbp}{los}
\floatname{scheme}{Scheme}
\floatname{chart}{Chart}
\newfloat{graph}{htbp}{loh}
\usepackage{amsmath,amsfonts}
\usepackage[noend,ruled,vlined]{algorithm2e}
\usepackage{array}
\usepackage[caption=false,font=normalsize,labelfont=sf,textfont=sf]{subfig}
\usepackage{textcomp}
\usepackage{stfloats}
\usepackage{url}
\usepackage{hyperref}
\usepackage{verbatim}
\usetikzlibrary{shapes}
\usepackage{xr}

\usepackage{array}
\usepackage{tabularx}
\usepackage{placeins}
\usepackage{mathtools}

\usepackage{bm}
\usepackage{tikz}
\usepackage{etoc}
\usepackage{orcidlink}
\usepackage{xspace}

\usepackage{booktabs} 
\usepackage{enumitem}
\usepackage{amssymb}
\usepackage{amsthm}
\theoremstyle{thmstyleone}%
\newtheorem{theorem}{Theorem}
\newtheorem{corollary}[theorem]{Corollary}%
\newtheorem{proposition}[theorem]{Proposition}%
\newtheorem{lemma}[theorem]{Lemma}%

\newcommand{\code}[1]{\texttt{#1}}
\newcommand{\ourtool}{\code{autogato}\xspace}

\usepackage[switch]{lineno}

\theoremstyle{thmstyletwo}%
\newtheorem{example}{Example}%
\theoremstyle{thmstylethree}%
\newtheorem{definition}[theorem]{Definition}%

\SetCommentSty{mycommfont}

\makeatletter
\patchcmd{\@bibitem}{\texttt{doi:}}{\href{https://doi.org/\@bibitem}{doi:}}{}{}

\def\moverlay{\mathpalette\mov@rlay}
\def\mov@rlay#1#2{\leavevmode\vtop{%
    \baselineskip\z@skip \lineskiplimit-\maxdimen
    \ialign{\hfil$\m@th#1##$\hfil\cr#2\crcr}}}
\newcommand{\charfusion}[3][\mathord]{
  #1{\ifx#1\mathop\vphantom{#2}\fi
    \mathpalette\mov@rlay{#2\cr#3}
  }
  \ifx#1\mathop\expandafter\displaylimits\fi}
\DeclareRobustCommand\bigop[1]{%
  \mathop{\vphantom{\sum}\mathpalette\bigop@{#1}}\slimits@
}
\newcommand{\bigop@}[2]{%
  \vcenter{%
    \sbox\z@{$#1\sum$}%
    \hbox{\resizebox{\ifx#1\displaystyle.9\fi\dimexpr\ht\z@+\dp\z@}{!}{$\m@th#2$
}}%
  }%
}

\newcommand{\cupdot}{\charfusion[\mathbin]{\cup}{\cdot}}

\DeclareMathOperator{\Hasse}{Hasse}

\newcommand{\ORCID}[1]{}

\newcommand{\child}[0]{\bm{\ensuremath{\kappa}}}
\newcommand{\king}[0]{\ensuremath{\mathbf{K}}}

\newcommand{\SM}[0]{\ensuremath{\mathbf{S}}}

\newcommand{\Jacob}[0]{\ensuremath{\mathbf{D}}}
\newcommand{\Metzler}[1]{\ensuremath{\mathfrak{M}(#1)}}
\newcommand{\elem}[0]{\ensuremath{\mathfrak{Q}}}

\newcommand{\tree}[0]{\ensuremath{\mathcal{T}}}

\newcommand{\kine}[0]{\ensuremath{\nu}}
\newcommand{\Reac}[0]{\ensuremath{\mathcal{N}}}
\newcommand{\Auto}[0]{\ensuremath{\mathbf{A}}}

\newenvironment{mainitem}[2]{\medskip\par\noindent%
  \textbf{#1~M\ref{#2}.}}{\par}

\usepackage{authblk}
\author[1,3]{Richard Golnik*}
\author[1]{Thomas Gatter}
\author[1,2,3,4,5,6,7,8,9]{Peter F. Stadler}
\author[1]{Nicola Vassena}
\affil[1]{Bioinformatics Group, Department of Computer Science, Leipzig University, D-04107, Leipzig, Germany}
\affil[2]{Interdisciplinary Center for Bioinformatics, Leipzig University, D-04107, Leipzig, Germany}
\affil[3]{Zuse School for Embedded and Composite Artificial Intelligence (SECAI)}
\affil[4]{Center for Scalable Data Analytics and Artificial Intelligence, Leipzig University, D-04107, Leipzig, Germany}
\affil[5]{Max Planck Institute for Mathematics in the Sciences, D-04103 Leipzig Germany}
\affil[6]{Department of Theoretical Chemistry, University of Vienna, A-1090 Wien, Austria}
\affil[7]{Center for non-coding RNA in Technology and Health, University of Copenhagen, DK-1870, Frederiksberg, Denmark}
\affil[8]{Facultad de Ciencias, Universidad Nacional de Colombia, Bogot{\'a}, Colombia}
\affil[9]{Santa Fe Institute, Santa Fe, NM 87501, USA}

\title{Enumeration of Autocatalytic Subsystems in Large Chemical Reaction Networks}

\date{*Email: richard@bioinf.uni-leipzig.de}

\begin{document}

\maketitle

\begin{abstract}
  Autocatalysis is an important feature of metabolic networks, contributing
  crucially to the self-maintenance of organisms. Autocatalytic subsystems
  of chemical reaction networks (CRNs) are characterized in terms of
  algebraic conditions on submatrices of the stoichiometric matrix
  $\SM$. Here, we derive sufficient conditions for subgraphs supporting
  irreducible autocatalytic systems in the bipartite K{\"o}nig
  representation of the CRN. On this basis, we develop an efficient
  algorithm to enumerate autocatalytic subnetworks and, as a special case,
  autocatalytic cores, i.e., minimal autocatalytic subnetworks, in
  full-size metabolic networks. The same algorithmic approach can also be
  used to determine autocatalytic cores only. As a showcase application, we
  provide a complete analysis of autocatalysis in the core metabolism of
  \textit{E.\ coli} and enumerate irreducible autocatalytic subsystems of
  limited size in full-fledged metabolic networks of \textit{E.\ coli},
  human erythrocytes, and \textit{Methanosarcina barkeri}
  (Archaea). The mathematical and algorithmic results are accompanied
  by software enabling the routine analysis of autocatalysis in large CRNs.
\end{abstract}

\section{Introduction}

An autocatalytic reaction is ``a chemical reaction in which a product (or a
reaction intermediate) also functions as a catalyst''
\cite{gold_iupac_2025}. Self-replication, i.e., the ability of multiplying
instances of the self, is a special case of autocatalysis that is inherent
to all living organisms. The emergence of self-replicating systems hence is
a key issue in theories of the origin of life, independent of whether an
RNA world, a lipid world, or a metabolism-first scenario is envisioned
\cite{eigen_selforganization_1971,eigen_principle_1977,gilbert_origin_1986,
  joyce_rna_1989, joyce_antiquity_2002,
  bissette_mechanisms_2013,howlett_autocatalysis_2023}.  In a more general
setting, autocatalysis is a property of chemical reaction networks (CRNs)
that collectively implements an autocatalytic overall reaction without any
of the constituent reactions being autocatalytic. It is important to
distinguish two fundamentally different modeling frameworks: networks
of autocatalysts, such as the hypercycles of Eigen \& Schuster
\cite{eigen_principle_1977}, and catalytic reaction systems of Hordijk \&
Steele \cite{hordijk_autocatalytic_2010}, which presuppose that \emph{all}
reactions are explicitly and specifically catalyzed by members of the
system. Such systems naturally represent interactions of complex
entities such as RNAs, proteins, or other heteropolymers, even though
  there is mounting evidence that small molecules and cofactors in metabolic
  networks, as well as metal ions, also exert catalytic activity
  \cite{xavier_autocatalytic_2020}. In contrast, catalysis is usually considered as
an emergent property in networks of abiotic chemical reactions among
small molecules. More precisely, catalysis in this setting is the net
effect of a sequence of individual reactions. Here, we will be concerned
exclusively with CRNs without explicit catalysts.

Autocatalysis in CRNs was generally considered to be scarce in
non-enzymatic chemistry
\cite{butlerow_bildung_1861,orgel_implausibility_2008,
  peng_assessment_2023}.  On the other hand, it has been argued repeatedly,
that metabolic networks are dominated by autocatalytic subsystems
\cite{hordijk_autocatalytic_2010, orgel_implausibility_2008,
  barenholz_design_2017}.  Until recently, the lack of a consistent
definition of autocatalysis made it difficult to discuss the prevalence of
autocatalytic structures in chemical networks
\cite{andersen_defining_2021}. This situation changed when Blokhuis
\emph{et al.} proposed an algebraic definition of autocalytic submatrices
based only on structural properties of a chemical network encoded by the
stoichiometric matrix \cite{blokhuis_universal_2020}. This notion of
autocatalysis has become widely accepted because it not only captures
key features of collective autocatalysis, but also turned out to
mathematically well-behave \cite{vassena_unstable_2024} and to be suitable
for constructing practical algorithms \cite{gagrani_polyhedral_2024}
identifying autocatalytic subnetworks. To the best of our knowledge,
available tools enumerating autocatalytic cycles are restricted to network
sizes of approximately 300 metabolites and reactions
\cite{gagrani_polyhedral_2024, honegger_efficient_2022, golnik_birne_2025,
  kosc_thermodynamic_2025}. Gagrani \emph{et al.}
\cite{gagrani_polyhedral_2024} at present offers the most capable method
currently accessible for networks of this scale, while the other approaches
are either not yet publicly available \cite{honegger_efficient_2022} or
limited to smaller networks \cite{golnik_birne_2025}. This falls short of
the capability to analyze much larger metabolic networks in living
organisms ranging from bacteria to animals and plants
\cite{norsigian_bigg_2020, moretti_metanetx_2016, hawkins_plant_2025,
  deoliveiradalmolin_aragem_2010, poolman_responses_2013, saha_zea_2011,
  yuan_genomescale_2016}.

\section{Road Map}

In this contribution, we describe a graph-theoretical approach to identify
irreducible autocatalytic subsystems and the software package
  \textit{autogato} that implements these methods. The concise and
  self-contained presentation of the mathematical results underlying the
  algorithmic developments requires a considerable level of technicality.
  We therefore start with a road map of this contribution that summarizes
  the main results and their implications in a less technical manner and with
  only a minimum of notation. Still, we introduce a few concepts formally
  already in this overview section. An illustration of this section is
  depicted in Fig.~\ref{fig:Roadmap}.
  
\begin{figure}[htb]
  \centering
  \begin{tikzpicture}
    \node[draw, ellipse, minimum width = 2.0cm, minimum height = 1.1cm, label={[font=\footnotesize, align = center] center:{Autocatalytic\\ cores}} ] (a) at (2,4.5) {};
    \node[draw, ellipse, minimum width = 5.0cm, minimum height = 2.1cm, label={[font=\footnotesize, align=center, xshift=73pt] left:{Autocatalytic \\ CS matrices  \\ with irreducible \\ Metzler part}}] (c) at (1,4.5) {};
    \node[draw, ellipse, minimum width = 7.5cm, minimum height = 3.6cm,
    line width=1.4pt,
    label={[font=\footnotesize, align=center, xshift=65pt] left:{CS  Matrices \\ with irreducible \\ Metzler-part}}] (d) at (0,4.5) {};
    
    \node[] (f1) at (0, 1.1) {\textbf{Fluffles} $F\in \mathfrak{F}$};
    \node[draw, rectangle, rounded corners, line width=1.4pt] (f) at (0,-0.5) {
      \begin{minipage}{11cm}
    	\begin{itemize}
      	\item {${|X(F)|=|R(F)|}$  \hspace{26pt} \textbf{(Child-selection, Cor.~\ref{cor:DegK(k)})}}
      	\item {$d_{out}(x)=1, \forall x\in X(F)$ \textbf{(Metzler property, Lem.~\ref{lem:CSgraph})}}
	\item{$d_{in}(r)=1, \forall r\in R(F)$  \hspace{4pt} \textbf{(Metzler property, Lem.~\ref{lem:CSgraph})}}
      	\item {Strong block \hspace{44pt} \textbf{(Irreducible Metzler part, Thm.~\ref{thm:IrrerMetzler})}}
    	\end{itemize}
      \end{minipage}
    };
    \node[draw, rectangle, rounded corners, align = center] (h) at (-8, -0.5) {K\"onig graph: \\ $\king(X,R)$};			
    \node[draw, rectangle, rounded corners, align = center] (i) at (-8, -3) {Elementary circuits: $\mathbf{C}$};
    \node[draw, rectangle, rounded corners, align=center] (j) at (-4, -4) {Circuitnets: \\ $\mathcal{C} \in \mathfrak{P}(\mathbf{C})$};						
    \node[align=center] (k) at (-4,2.5)  {\textbf{Translation of matrix into}\\ \textbf{graph properties (Prop.~\ref{prop:usefluffles})}};
			
    \draw[<->, double, rounded corners] (d) to (f1);
    \draw[->, dashed, rounded corners] (0.3,3.3) to [bend right=30] node[midway, below right=-3pt] {ILP/$\lambda^*$} (1,3.7);
    \draw[->, rounded corners] (h) to (-8, 2) to (0, 2);
    \draw[->, rounded corners] (h) to (i);
    \draw[->, rounded corners] (i) to (-8,-4) to (j);
    \draw[->, rounded corners] (j) to (0,-4) to node[draw, pos=0.4, right, align =center, rectangle, rounded corners] {CS-equivalence class $\bumpeq$ (Def.~\ref{def:csequivalence}): \\ $\mathcal{C}_1\bumpeq \mathcal{C}_2 \Leftrightarrow E_1\left(\bigcup \mathcal{C}_1\right)=E_1\left(\bigcup \mathcal{C}_2\right)$} node[pos=0.75, right] {\textbf{Alg. Enumeration (Lem.~\ref{lem:bumpeqEquiv})}} (f);			
    
  \end{tikzpicture}
  \caption{Schematic depiction of the main results. A CRN $\Gamma =
      (X,R)$ is represented by its bipartite K\"onig graph
      $\king(X,R)$. Child-selection (CS) matrices with irreducible Metzler
      part can be identified as subgraphs of $\king$, termed
      \emph{fluffles}. Equivalent \emph{fluffles}, w.r.t. to their CS
      matrices, are enumerated efficiently via sets of elementary circuits
      (\emph{circuitnets}) by taking advantage of CS-equivalence
      classes. CS matrices $\SM[\child]$ of \emph{fluffles} are extracted
      and their autocatalytic capacity assessed using an ILP or, in the
      Metzler case, computing the leading eigenvalue $\lambda^*$ 
      (Lemma~\ref{cor:SingularMetzlerHurwitz}). The
      matrix properties inducing fluffle properties and the relevant
      associated results are stated in brackets.}
  \label{fig:Roadmap}
  \end{figure}

\paragraph{Reaction Networks.}
A chemical reaction network (CRN) $\Gamma$ is a pair of finite sets
$\Gamma\coloneqq (X,R)$, where $X$ is a set of chemical species or
metabolites and $R$ is a set of chemical reactions. A reaction $r$ is a
directed transformation between nonnegative linear combinations of
metabolites and can be described as:
\begin{equation}\label{eq:reaction}
  \sum_{x\in X} s_{xr}^- \cdot x \quad
  \underset{r}\rightarrow \quad \sum_{x\in X} s_{xr}^+ \cdot x 
\end{equation}
with $s_{xr}^-\ge0$ and $s_{xr}^+\ge0$ denoting the nonnegative
stoichiometric coefficients of the molecular count. Typically these
coefficients are integers, but our theory does not make a distinction about
it. Metabolites $x$ appearing with nonzero $s_{xr}^->0$ on the left-hand
side of \eqref{eq:reaction} are called \emph{reactants} or
\emph{substrates} or \emph{educts} of $r$, while metabolites $x$
appearing with nonzero $s_{xr}^+>0$ on the right-hand side of
\eqref{eq:reaction} are called \emph{products} of $r$. We may collect the
stoichiometric coefficients of the reactant species in the $|X|\times |R|$
\emph{reactant matrix} $\mathbf{S}^{-}: \mathbf{S}^-_{xr} \coloneqq
s^-_{xr}$ and the coefficients of the product species in the $|X|\times
|R|$ \emph{product matrix} $\mathbf{S}^+: \mathbf{S}^+_{xr} \coloneqq
s^+_{xr}$. The difference between the two matrices gives rise to the
\emph{stoichiometric matrix} $\SM:=\mathbf{S}^+ -\mathbf{S}^-$ with
entries:
\begin{equation}\label{eq:stoichmatrix}
   \SM_{xr} \coloneqq s_{xr}^+ - s_{xr}^-.
\end{equation}
In the absence of explicit catalysis, no chemical species $x$ is both
  reactant and product of the same reaction $r$ and hence $\SM_{xr}>0$
  implies $\SM_{xr} = s_{xr}^+$ and $s_{xr}^-=0$, i.e., $x$ is a product of
  $r$, while $\SM_{xr}<0$ implies $\SM_{xr} = -s_{xr}^-$ and $s_{xr}^+=0$,
  i.e., $x$ is a reactant of $r$. The stoichiometric matrix then completely
  describes the CRN. In particular, therefore, we do not consider here
  explicitly autocatalytic reactions of the form $s^-_{xr}\cdot x + \ldots
  \;\underset{r}{\rightarrow}\; s^{+}_{xr} \cdot x + \ldots$ with $0 <
  s^-_{xr} < s^{+}_{xr}$
  \cite{eigen_principle_1977,bissette_mechanisms_2013}.  \emph{Throughout
  this contribution, we assume that there are no explicit catalysts.}  An
  extension to the more general case will be considered elsewhere
  \cite{golnik_autocatalytic_2026}.
  
The most natural representation of $\Gamma$ is a directed, weighted
hypergraph with vertex set $X$ and edge set $R$. Each reaction $r \in R$ is
represented by a hyperarrow, whose inputs correspond to the reactant
species and outputs to the product species. The stoichiometric coefficients
appear as weights for the hyperarrow.  Building on this perspective,
we consider CRNs in \emph{K\"onig representation}
\cite{zykov1974hypergraphs} throughout, i.e., we represent the hypergraph
$\Gamma$ as a directed bipartite graph
\begin{equation}
  \mathbf{K}\coloneqq \king(X,R) \coloneqq (X \cupdot R,E),
\end{equation}

with disjoint vertex sets $X$ and $R$ and edge set $E\coloneqq E_1\cup E_2$
where $E_{1}(\king)\coloneqq \{(x,r) \ \vert \ s_{xr}^- > 0\}$ and
$E_{2}(\king)\coloneqq \{(r,x) \ \vert \ s_{xr}^+ > 0\}$. The absence
  of explicit catalysts excludes so-called digons, i.e., pairs of edges
  $(x,r)$ and $(r,x)$ are never present in $\king$.
\FloatBarrier
\paragraph{Child-selections (CS)} 
\cite{vassena_good_2020,vassena_unstable_2024} are a central technical
  tool relating the structure of a CRN to certain dynamical properties.
  Mathematically, they are motivated by a structural and symbolic analysis
  of the Jacobian matrix of dynamical systems associated with a CRN. More
  precisely, via a Cauchy--Binet decomposition, only the square `CS'
  matrices $\SM[\child]$ associated with a child-selection contribute to the
  characteristic polynomial of the Jacobian matrix. We briefly introduce
  them below, but we defer a detailed discussion of this topic to
  \textit{Sec.~\nameref{sect:PRCK}}.

\begin{definition}
  A \emph{$k$-child-selection triple}, or $k$-CS for short, is a triple
  $\pmb{\kappa}=(X_\kappa,R_\kappa,\kappa)$ such that
  $|X_\kappa|=|R_\kappa|=k$, $X_\kappa\subseteq X$, $R_\kappa\subseteq R$,
  and $\kappa:X_\kappa\to R_\kappa$ is a bijection satisfying
  $s_{x\kappa(x)}^->0$ for all $x\in X_\kappa$. We call $\kappa$ a \emph{CS
  bijection}.
\end{definition}
To any given $k$-CS
$\child=(X_\kappa,R_\kappa,\kappa)$, we associate a $k \times k$ CS matrix
$\SM[\child]$ defined as follows:
\begin{equation}
  \SM[\child]_{xw}\coloneqq s^+_{x\kappa(w)}-s^-_{x\kappa(w)},
  \quad\quad\text{for all }x,w\in X_\kappa.
\end{equation} 
Note that a CS matrix may differ from a submatrix of $\SM$ by having a
different column order. In particular, for a fixed ordering of the species
$X=\{x_1,\ldots,x_{|X|}\}$, the column order of the stoichiometric matrix
$\SM$ depends on the ordering of the set $R$, whereas the ordering of
$\SM[\child]$ is independent of it, depending only on the order of $X$.

\paragraph{Autocatalysis.}
 Blokhuis et al.~\cite{blokhuis_universal_2020} derived a matrix definition
 of autocatalysis from the definition of the IUPAC (International Union of
 Pure and Applied Chemistry) \cite{gold_iupac_2025}:
\begin{definition}\label{def:autocatmatrix}
  A matrix $\Auto\in\mathbb{R}^{n,m}$ is \emph{autocatalytic} if
  \begin{itemize}
  \item[(i)]  there is $v\in \mathbb{R}_{>0}^{m}$ such that
    $\Auto v > 0$;
  \item[(ii)] for each $r\in\{1,\dots,m\}$ there is $x,y\in\{1,\dots,n\}: 
    \Auto_{xr}<0$ and $\Auto_{yr}>0$.
  \end{itemize} 
\end{definition}
We use $\SM[X_\kappa,R_\kappa]$ to refer to the submatrix of $\SM$ with
species index in $X_\kappa\subseteq X$ and reaction index in
$R_\kappa\subseteq R$, and we can then define an autocatalytic network
consequently.

\begin{definition}
A network $\Gamma$ is autocatalytic if its stoichiometric matrix $\SM$
possesses an autocatalytic submatrix.
\end{definition}
Recent literature on autocatalysis
\cite{blokhuis_universal_2020, vassena_unstable_2024,
    gagrani_polyhedral_2024} put considerable emphasis on minimal
  autocatalytic matrices, so-called autocatalytic cores:
\begin{definition}
  A matrix $\mathbf{A}$ is an autocatalytic core if $\mathbf{A}$ is
  autocatalytic and does not contain a proper autocatalytic submatrix.
\end{definition}
The main goal of this contribution is to develop an algorithmic approach
  to compute autocatalytic substructures in large CRN.

  \paragraph{From autocatalytic matrices to fluffle graphs.}
  The key idea underlying our approach is to exploit relationships
  between certain autocatalytic subnetworks $(X',R')$, including
  autocatalytic cores, and subgraphs of $\king$ with special structures.
  This is not a trivial endeavor since the conditions on submatrices of
  $\SM$ that define autocatalysis are algebraic by nature and do not have
  an obvious translation to the language of graphs.

  Our starting point is the known observation that each autocatalytic core
  induces a \emph{unique} child-selection $\child$ since each reaction $r$
  in a core has a unique reactant, which appears as a negative entry in
  the diagonal of $\mathbf{A}=\SM[\child]$, see \ \cite{vassena_unstable_2024} 
  and Prop.~\ref{prop:AutCore} below. Child-selections, in turn, yield the
  first tangible connection to graphs: each pair $(x,\kappa(x))$ appears as
  an edge in $\king$ that connects a reaction with one of its
  reactants. Each child-selection $\child=(X',R',\kappa)$, therefore, can
  be associated with a subgraph $\king(\child)$ that is spanned by the set
  of reactant-to-reaction edges $E_1(\child)\coloneqq \{(x,\kappa(x))| x\in
  X'\}$ (see Thm.~\ref{thm:CSPerfMatch}). Our task thus becomes the
  characterization of child-selections, and thus subgraphs of $\king$, which
  can give rise to interesting autocatalytic subnetworks. To this end, we
  further investigate the matrices $\SM[\child]$.

  The second key observation is that for autocatalytic cores, $\SM[\child]$
  is a so-called Metzler matrix, i.e., all off-diagonal entries are
  non-negative, see \cite{blokhuis_universal_2020,vassena_unstable_2024}
  and Prop.~\ref{prop:Metzler} below, which implies that $\king(\child)$ is
  an induced subgraph of $\king$ (Cor.~\ref{cor:AutMatrixKing}). For more
  general child-selections, we show that a necesssary condition for
  $\SM[\child]$ to be autocatalytic is that its Metzler part
  $\Metzler{\SM[\child]}$ is autocatalytic. This matrix is obtained by
  replacing all negative off-diagonal entries in $\SM[\child]$ by $0$. The
  matrix $\Metzler{\SM[\child]}$ in turn corresponds exactly to the (not
  necessarily induced) subgraph $\king(\child)$.

  Autocatalytic cores are not the only interesting autocatalytic
  subnetworks. In child-selections, every species $x$ is consumed; if it is
  not produced, it cannot be part of an autocatalytic subnetwork. The
  definition of autocatalysis, furthermore, requires that every reaction
  has at least one reactant and at least one product. Thus every species
  $x\in X'$ and every reaction $r\in R'$ of an autocatalytic CS
  $\child=(X',R',\kappa)$ lies on a directed circuit in $\king$, and
  equivalently, all $x\in X'$ lie on a circuit in the so-called
  \emph{substrate graph}, which coincides with the graph representing the
  non-negative matrix $\Metzler{\SM[\child]}$. Very naturally, we restrict
  ourselves to irreducible matrices $\Metzler{\SM[\child]}$ and thus
  strongly connected substrate graphs. Considering the corresponding
  subgraph $\king(\child)$, we show that the irreducible matrices
  $\Metzler{\SM[\child]}$ are in 1-1 correspondence to subgraphs of
  $\king(\child)$ with the following properties:
  \begin{itemize}
  \item[(i)] $\king(\child)$ is balanced, i.e., $|X'|=|R'|$,
    Cor.~\ref{cor:DegK(k)};
  \item[(ii)] every $x\in X'$ has out-degree $1$ and every $r\in R'$ has
    in-degree $1$, Lem.~\ref{lem:CSgraph}; 
  \item[(iii)] $\king(\child)$ is a strong block, i.e., it is strongly
    connected and its underlying undirected graph is bi-connected,
    Thm.~\ref{thm:IrrerMetzler}.
  \end{itemize}
  We introduce the term \emph{fluffle} for directed bipartite (sub)graphs
  satisfying (i), (ii), and (iii).

In the following, we say that an (autocatalytic) subnetwork defined
    by a CS $\child$ is irreducible if the Metzler part of its corresponding
    child-selection matrix, i.e., $\Metzler{\SM[\child]}$ ,is irreducible.

  \paragraph{From Fluffles to Enumerating Irreducible Autocatalytic
    Subnetworks.}  The discussion so far suggests an attractive avenue to
  list irreducible autocatalytic subnetworks: enumerate the fluffle
  subgraphs of $\king$ and test their corresponding CS matrices.  Our task,
  therefore, becomes to list fluffles efficiently. The key property of
  fluffles that makes such an algorithm practicable is the fact that
  fluffles are characterized by a special class of \emph{ear
  decompositions} (Thm.~\ref{thm:Woffle}). More precisely, every elementary
  circuit is a fluffle, and a directed ear can be added to a fluffle if and
  only if it has an even number of inner vertices, its initial vertex is
  a reaction vertex, and its terminal vertex is a species vertex. This yields a 
  simple condition on how a fluffle and an additional elementary circuit may 
  overlap for the superposition to again be a fluffle, and suggests the notion 
  of \emph{circuitnets} as sets of compatible elementary circuits that produce 
  fluffles (Thm.~\ref{thm:flufflecnets}). The final building block is the
  observation that for every child-selection $\child$ there is a
  representative (maximal) fluffle that can be constructed from elementary
  circuits as follows: $E_1(\child)=\bigcup_i E_1(C_i)$, i.e., the union of
  the reactant-to-reaction edges of the superimposed elementary circuits
  defines the CS edges, and thus $X'$ and $R'$. The remaining
  (reaction-to-product) edges are then induced by $\king$, i.e.,
  $E_2=E(\king)\cap(R'\times X')$, see Prop.~\ref{prop:kingrep}. Moreover,
  the $E_2$ edges $(\kappa(x),y)$ correspond to positive off-diagonal
  entries $\Metzler{\SM[\child]}_{xy}>0$. In particular, therefore, we need
  to consider only one elementary circuit for $X'\cup R'$ and restrict
  circuitnets to ``fundamental'' ones, where each circuit covers vertices
  not present in any other one.

  \paragraph{Computational Results and Performance.}
  These mathematical results form the basis for a practical strategy to
enumerate irreducible autocatalytic subnetworks, or more precisely,
autocatalytic CS matrices with an irreducible Metzler part from the
K{\"o}nig graph $\king$ of a CRN. Conceptually, we can break up this task
into four steps:
\begin{itemize}
\item[(1)] Enumerating the elementary circuits of $\king$.
\item[(2)] Construction of representatives for `CS-equivalence
  classes' (see Def. \ref{def:csequivalence} and \ref{def:bumpeq}) of fluffles by iteratively adding elementary circuits.
\item[(3)] Testing of these candidate CS matrices $\SM[\child]$ for
  autocatalysis.
\item[(4)] Identification of autocatalytic cores.
\end{itemize}
The algorithms addressing these basic tasks are described in detail in
Sec.~\textit{\nameref{ssect:basic}}. The enumeration of elementary
  circuits in digraphs is a well-studied task for which we employ
Johnson's algorithm \cite{johnson_finding_1975} in a version that
optionally limits the length of the circuit \cite{gupta_finding_2021}.
The algorithm described in this contribution is implemented in the
  \texttt{python} package \texttt{autogato}
  (https://github.com/hollyritch/autogato), which is distributed with a
  detailed description and examples. We therefore refrain from a more
  thorough description of the software here.
  
The direct application of this strategy to large CRNs, in particular to
genome-scale metabolic networks, requires prohibitive computational
resources. We observe, however, that the K\"onig graph of metabolic
networks is rather sparsely connected. Moreover, it typically contains
modules with higher internal connectivity that often can be identified with
functional biological submodules \cite{ravasz_hierarchical_2002,
  ma_connectivity_2003, ma_decomposition_2004, zhao_hierarchical_2006,
  sridharan_identification_2011}. We use this structure to decompose the
CRN into smaller parts following a divide-and-conquer approach. Interfaces
between submodules, however, may also be part of autocatalytic
subsystems. We therefore consider elementary circuits that connect modules
in the final stage. A major advantage of the decomposition into modules is
that their analysis can be trivially parallelized, making it possible in
practise to tackle large metabolic networks.

The exhaustive enumeration of \emph{all} autocatalytic subnetworks
  with irreducible Metzler part or \emph{all} autocatalytic cores is only
  feasible for moderate-sized CRNs, such as a model of the core carbon
  metabolism with 36 metabolites and 71 reactions obtained by removing
  exchange metabolites from the model proposed earlier
  \cite{orth_reconstruction_2010}.  Nevertheless, our approach makes it
  feasible to tackle much larger networks by restricting the size of
  subnetworks of interest. For genome-scale metabolic networks (again,
  after removing a small set of exchange metabolites) \texttt{autogato} can
  enumerate all irreducible autocatalytic subnetworks (including all cores)
  up to sizes of 10 metabolites and 10 reactions (e.g.\ for
  \textit{E.\ coli} DH5$\alpha$).

The computational results in Sec.~\textit{\nameref{sect:appl}} indicate
that autocatalysis is prevalent in all domains of life, however, to varying
degrees. In particular, they support earlier claims that autocatalytic
subsystems are abundant in metabolic networks
\cite{hordijk_autocatalytic_2010, orgel_implausibility_2008,
  barenholz_design_2017}.  Moreover, they demonstrate that the
graph-theoretic approach is robustly applicable to metabolic network models
of practical interest and that \emph{autogato} is a tool that can be used
  routinely by system biologists and chemists.
  
\paragraph{The remainder of this contribution} is concerned with the
  technical details required to make the road map above precise. We start
  in Sec.~\textit{\nameref{sect:CRNs}} with a summary of the main results
  on autocatalytic sets and the -- closely related -- connection between
  child-selections, the stoichiometric matrix, and the reaction kinetics of
  a CRN. In Sec.~\textit{\nameref{sect:autokoenig}}, we characterize a
class of subgraphs of the bipartite K{\"o}nig representation of a CRN that
we term fluffles and show that only fluffles can induce autocatalytic
  subsystems with irreducible Metzler part in the CRN. The
  graph-theoretical algorithm enumerating the autocatalytic subsystems in a
  CRN is described in full detail in
  Sec.~\textit{\nameref{sec:Algorithm}}. In these two sections, we
  introduce all necessary notation and report our mathematical results in
  the form of precise statements. Since this material is already quite
  extensive, we relegate all lengthy and technical proofs to the
  Supplemental Material.

Subsequently, Sec.~\textit{\nameref{sect:appl}} provides details on
  the selected computational results briefly discussed above. Although the
  technical part of this paper provides a fairly comprehensive
  understanding of the connection between autocalysis and the K{\"o}nig
  graph of a CRN, some technical questions of interest remain open; they
are briefly summarized in the \textit{\nameref{sect:conclusion}}. A
  short appendix connects the formalism derived here -- focused on
  circuits in the K{\"o}nig graph $\king$ -- with the classification of
  autocatalyic cores by Blokhuis \textit{et al.}
  \cite{blokhuis_universal_2020}. To this end, we introduce the notion of
  \emph{centralized autocatalysis}. Finally, we clarify the relationship of
  autocatalytic core and minimal autocatalytic subsystems
  (MAS) introduced by Gagrani et al. \cite{gagrani_polyhedral_2024}, 
  and show that our algorithmic approach can also enumerate MAS.

\section{Reaction Kinetics, Child-selections, and Autocatalysis}
\label{sect:CRNs}

\subsection{Parameter-rich chemical kinetics}
\label{sect:PRCK}

Let $z(t)\in\mathbb{R}^{|X|}_{\ge0}$ indicate the vector of the chemical
concentration of the species at time $t$ in a well-mixed, spatially
homogeneous, reactor.  The time-evolution for $z(t)$ is described by the
system of ordinary differential equations (ODEs):
\begin{equation}\label{eq:ODEdynamics}
  \dot{z} = f(z) := \SM\cdot \kine(z)
\end{equation}
where $\SM$ is the stoichiometric matrix defined in
\eqref{eq:stoichmatrix}, and $\kine(z)\in \mathbb{R}^{|R|}$ is the vector
of the reaction rate functions (kinetics).  A primary modeling issue in
reaction networks is the ubiquitous lack of precise knowledge of the
mathematical form of the rates $\kine(z)$. For this reason, the literature
typically resorts to kinetic models: a wide class of reaction functions
defined as follows.
\begin{definition}
  A monotone kinetic model for a reaction network $\Gamma\coloneqq (X,R)$
  is a vector-valued function $\kine:\mathbb{R}^{|X|}_{\ge0} \mapsto
  \mathbb{R}^{|R|}_{\geq 0}$, which satisfies the following conditions:
  \begin{enumerate}
  \item[i.] $\kine_r(z) \geq 0,$  for all $ z\in\mathbb{R}^{|X|}_{\ge0}$;
  \item[ii.] $\kine_r(z)>0$ implies $z_x>0$ for all species $x$ with
    $s^-_{xr}>0$;
  \item[iii.] $s_{xr}^- = 0$ implies $\partial \kine_r/\partial z_x \equiv 0$;
  \item[iv.] $z>0$ and $s_{xr}^->0$ implies $\partial \kine_r/z_{z}>0$.
\end{enumerate}
\end{definition}
Since this work always focuses on monotone kinetic models, for brevity,
throughout we simply write ``kinetic models''. Given the
aforementioned uncertainty on the precise quantities involved, it is
typical to consider parametric kinetic models. Whenever necessary, we will
refer to this dependency by writing $\kine(z,p)$. Standard examples of such
parametric kinetic models are classic \cite{MA64} and generalized
\cite{muller_generalized_2012} mass-action kinetics, both polynomials, and
more involved rational functions such as Michaelis--Menten kinetics
\cite{johnson_original_2011} and the Hill model \cite{Hill1910}.

Consider now a network $\Gamma$ endowed with a kinetic model $\kine$. Fixed
points $\bar{z}\ge0$ of $f$ in \eqref{eq:ODEdynamics}, i.e.,
\begin{equation}
    0=f(\bar{z})=\SM \cdot \kine(\bar{z}),
\end{equation}
are called \emph{steady-states} of $\Gamma$. Throughout this work, we
consider only \emph{consistent} networks \cite{angeli_petri_2007}, whose
stoichiometric matrix admits a positive right kernel vector
$\mathbf{\kine}>0: \SM \cdot \mathbf{\kine} = 0$, which is a necessary
condition for a network $\Gamma$ to admit at least one positive
steady-state.  It is well-known \cite{hsu_ordinary_2013} that the dynamical
stability of $\bar{z}$ can be addressed at first approximation by studying
the linearization of system \eqref{eq:ODEdynamics} at $\bar{z}$:
\begin{equation}\label{eq:ODEdynamicslinear}
    \dot{z}= \Jacob_f(\bar{z})z = (\SM\cdot \Reac(\tilde{z})) \; z,
\end{equation}
where $\Jacob_f(\bar{z})$ is the \emph{Jacobian matrix} evaluated at
$\bar{z}$.  The nonnegative matrix $\Reac \in \mathbb{R}^{|R| \times |X|}$
with entries
\begin{equation}
  \Reac_{rx}(\bar{z})\coloneqq
  \frac{\partial{\kine_r(z)}}{\partial z_x}\biggr\rvert_{z=\bar{z}} 
\end{equation}
is called \emph{reactivity matrix}. In particular for hyperbolic
steady-states $\bar{z}$, i.e., for which the Jacobian $\Jacob_f(\bar{z})$
has only eigenvalues with nonzero real part, the spectrum of the Jacobian
determines the dynamical stability: the steady-state $\bar{z}$ is stable if
$\Jacob_f(\bar{z})$ is Hurwitz-stable, i.e. it possesses only eigenvalues
with negative-real part, and $\bar{z}$ is unstable if $\Jacob_f(\bar{z})$
is Hurwitz-unstable, i.e., it possesses at least one eigenvalue with
positive-real part.

Given a network $\Gamma$ endowed with a parametric kinetic model
$\kine(z,p)$, the relation between the network structure and the possible
spectrum configurations of the Jacobian at the varying of parameters $p$ is
a classic problem that has turned out to be very challenging
\cite{aris_prolegomena_1965}. These relationships become much more
tractable if the parametric kinetic model has sufficient internal freedom,
at least as long existence results are the major concern
\cite{vassena_unstable_2024}.

\begin{definition}
  A monotone kinetic rate model $\kine(z,p)$ is \emph{parameter-rich} if,
  for every positive steady-state $\bar{z}>0$ and every choice of an
  $|R|\times |X|$ matrix $\Reac$ satisfying $\Reac_{rx}>0$ iff
  $s_{xr}^->0$, there exists a choice of parameters $\bar{p} = p(\bar{z},
  \Reac)$ such that $\partial\kine_{rx}(z,\bar{p})/\partial z_x
    |_{z=\bar{z}} = \Reac_{rx}.$
\end{definition}
Far from being just a theoretical construct, widely used schemes in
biochemistry -- such as Michaelis-Menten, Hill, and generalized mass action
-- are naturally parameter-rich. Classical mass-action kinetics, however,
lacks sufficient parametric freedom and is therefore not parameter-rich.

The advantage of the parameter-rich framework is that we may then consider
a \emph{symbolic reactivity matrix} $\Reac$, that is, any $|R|\times |X|$
matrix whose nonnegative symbolic entries satisfy $\Reac_{rx} > 0 \iff
s_{xr}^- > 0$, and study the spectrum of the associated \emph{symbolic
Jacobian matrix}
\begin{equation}
  \Jacob := \SM \Reac,
\end{equation}
which no longer depends explicitly on the steady-state value $\bar{z}$. In
particular, under parameter-rich kinetics, the existence of some evaluation
of $\Reac$ such that $\Jacob$ has a given spectrum directly implies the
existence of kinetic parameters for which this very $\Jacob$ is realized as
the actual Jacobian at a steady state $\bar{z}$, and thus dynamical
conclusions can be drawn. Accordingly, we say that the network $\Gamma$
\emph{admits instability} if there exists a choice $\Reac$ of the kinetic
matrix such that the symbolic Jacobian $\Jacob$ is Hurwitz-unstable.

On the other hand, since the entries of the symbolic reactivity matrix are
determined as zero or positive solely from the stoichiometric matrix, this
approach can be used to draw conclusions about the range of realizable
  dynamics from purely structural information. To see this, we consider a
  $k$-CS $\child$. Without loss of generality, assume
$X_\kappa=\{x_1,..,x_k\} \subset X$. By choosing the following
rescaling for the symbolic reactivity matrix $\Reac$,
\begin{equation}
  \Reac_{rx}(\varepsilon)=\begin{cases}
  1 & \text{if $x\in X_\kappa$ and $r=\kappa(x)$} \\
  \varepsilon &\text{otherwise, if $s^-_{xr}>0$}
  \end{cases} 
\end{equation}
a straightforward computation shows that the associated symbolic Jacobian
$\Jacob(\varepsilon)$ reads:
\begin{equation}\label{eq:reductionCSJacobian}
  \Jacob(\varepsilon)=\begin{pmatrix}
  \SM[\child] + O(\varepsilon) & O(\varepsilon)\\
  ... & O(\varepsilon)
  \end{pmatrix},
\end{equation}
where $O(\varepsilon)$ indicate an expression of order $\varepsilon$. A
detailed derivation can be found in Vassena \textit{et al.}
\cite{vassena_unstable_2024}. In particular, for $\varepsilon$ small
enough, the $k$ eigenvalues of $\SM[\child]$ approximate the $k$ largest
(in absolute value) eigenvalues of $\Jacob(\varepsilon)$. This argument
shows that any $k$-CS matrix can be used to approximate $k$ dominant
eigenvalues of the Jacobian. In particular, we obtain a straightforward
condition for a network to admit instability:
\begin{proposition}[Cor.~5.1 \cite{vassena_unstable_2024}]
  \label{prop:csunstable}
  Consider a network $\Gamma\coloneqq (X,R)$ with parameter-rich kinetics.
  If there is a $k$-CS $\pmb{\kappa}$ such that its associated
  $k\times k$ CS matrix is Hurwitz-unstable, then the network admits
  instability.
\end{proposition}

Moreover, using the same line of reasoning, it was shown 
that the presence of autocatalysis in the network always implies 
that the network admits instability \cite{vassena_unstable_2024}. 
The next section briefly reviews this connection.

\subsection{Autocatalytic matrices}\label{sec:AutocatalyticCycles}

Since we exclude explicit catalysis here, we can equivalently express
network properties as properties of the matrix $\SM$. Key features of
autocatalytic cores are collected in the following proposition, which
has been proven previously \cite{blokhuis_universal_2020,
vassena_unstable_2024}.
  
\begin{proposition}
  \label{prop:Metzler} 
  Let $\tilde{\Auto}$ be an autocatalytic core. The following all hold true:
  \begin{enumerate}
  \item $\tilde{\Auto}$ is an invertible square matrix;
  \item There exists a unique autocatalytic core $\Auto$ with strictly
    negative diagonal obtained by reordering the columns of $\tilde{\Auto}$;
  \item The off-diagonal entries of $\Auto$ obtained at point 2 are
    nonnegative.
  \end{enumerate}
\end{proposition}
Square matrices with nonnegative off-diagonal entries are called
\emph{Metzler} in the literature, and their stability properties have been
extensively studied in connection with the Frobenius-Perron Theorem
\cite{bullo_lectures_2018}. Throughout this paper, we refer to
autocatalytic cores $\Auto$ always intending the Metzler representation
with negative diagonal and nonnegative off-diagonal. In this case, further
properties -- related to dynamical stability -- were shown in
ref.\cite{vassena_unstable_2024}.
\begin{proposition}
  \label{prop:AutCore}
Let $\Auto$ be an $n\times n$ autocatalytic core in Metzler form. The
following all hold true:
\begin{enumerate}
\item $\Auto=\SM[\child]$ for a unique child-selection $\child$; 
\item $\Auto$ is irreducible;
\item $\Auto$ is Hurwitz-unstable. More precisely, $\Auto$ possesses
  exactly one eigenvalue with positive real part, and thus its determinant
  is of sign $\mathrm{sign}\det\Auto=(-1)^{n-1}$.
\end{enumerate}
\end{proposition}
\noindent As a key consequence, these properties imply:
\begin{corollary}
  If the network is autocatalytic, then it admits instability.
\end{corollary}

The emphasis on minimal autocatalytic subnetworks is mostly justified for
qualitative and classification purposes. In contrast, the Jacobian
rescaling \eqref{eq:reductionCSJacobian} suggests that larger CS matrices
may better capture the overall dynamical impact of autocatalysis on the
system, in terms of instability and growth rate, since fewer variables are
$\varepsilon$-rescaled. We therefore aim to develop a detection algorithm
that goes beyond autocatalytic cores. A natural broader class of interest
is given by CS matrices that are irreducible Metzler matrices. For this
class, the link between autocatalysis and instability is fully preserved,
as in Prop.~\ref{prop:AutCore}. It can be restated as a direct consequence
of the Perron–Frobenius theorem
\cite{vassena_unstable_2024,bullo_lectures_2018} as follows:
\begin{lemma}\label{cor:SingularMetzlerHurwitz}
  Let $\SM[\child]$ be an irreducible Metzler matrix. The following are
  equivalent:
 \begin{enumerate}
 \item $\SM[\child]$ is Hurwitz-unstable;
 \item $\SM[\child]$ has a real positive eigenvalue;
 \item $\SM[\child]$ is autocatalytic.
 \end{enumerate} 
\end{lemma}
For irreducible Metzler matrices, Hurwitz-instability (a spectral property
in general sensitive to column ordering) is therefore equivalent to
autocatalysis (a structural property independent of ordering).  On the
other hand, without loss of generality in network labeling, any
\emph{reducible} Metzler CS matrix $\SM[\child]$ can be represented in
block form as
\begin{equation}
  \SM[\child]=
  \begin{pmatrix}
    \mathbf{S}[\child'] & 0\\
    \mathbf{B} & \mathbf{S}[\child'']
  \end{pmatrix},
\end{equation}
where $\mathbf{S}[\child']$ is irreducible. In the above representation, we
say that $\SM[\child]$ is decomposed as a \emph{cascade originating from
$\mathbf{S}[\child']$}. The next result shows that autocatalysis for
$\SM[\child]$ necessarily requires autocatalysis in $\mathbf{S}[\child']$.

\begin{proposition}[Proof: SI] 
  \label{prop:autored}
  Let $\child=(X_{\kappa}, R_{\kappa}, \kappa)$ be a $k$-CS whose
  associated CS matrix $\SM[\child]$ is reducible, Metzler, and
  autocatalytic. Then there exists a $k'$-CS $\child'=(X_{\kappa'},
  R_{\kappa'}, \kappa')$ with $X_{\kappa'}\subset X_{\kappa}$,
  $R_{\kappa'}\subset R_{\kappa}$, and
  $\kappa'(X_{\kappa'})=\kappa(X_{\kappa'})$, such that its associated
  CS matrix $\SM[\child']$ is an irreducible autocatalytic Metzler matrix.
\end{proposition}

In line with our emphasis on Metzler matrices, we will see that the
Metzler part of a CS matrix, introduced below, plays a key role.

\begin{definition}
  For every 
  CS matrix $\SM[\child]$ we define the $k\times k$ matrix
  $\Metzler{\SM[\child]}$ with entries:
  \begin{equation}
    \Metzler{\SM[\child]}_{xr}:=\begin{cases}
    \SM_{xr}[\child] \quad &\text{if } x=r \text{ or }
    x\neq r \text{ and } \SM_{xr}[\child]>0\\
    0 \quad &\text{otherwise}
    \end{cases}
  \end{equation}
\end{definition}
  
By construction, we have $\Metzler{\SM[\child]}_{xr}\ne\SM[\child]_{xr}$ if
and only if $r\ne\kappa(x)$ and $\SM_{xr}<0$, i.e., if and only if $x$ is a
reactant of a reaction $r$ other than the one assigned to $x$ by the
CS-bijection $\kappa$. These entries correspond exactly to the negative
off-diagonal elements of $\SM[\child]$. Consequently,
$\Metzler{\SM[\child]}$ contains only nonnegative off-diagonal entries and
is therefore a Metzler matrix.  We therefore call $\Metzler{\SM[\child]}$
the \emph{Metzler part} of $\SM[\child]$. The Metzler part of a CS
  matrix will be of key interest because it yields a necessary condition
  for autocatalysis:
\begin{proposition}
  \label{prop:Metzlerpartaut} 
  If $\SM[\child]$ is an autocatalytic CS matrix, then its Metzler part
  $\Metzler{\SM[\child]}$ is autocatalytic. 
\end{proposition}
\begin{proof}
  We observe that $\Metzler{\SM[\child]}_{ij}\ge \SM[\child]_{ij}$ for all
  $i,j$, and hence $\Metzler{\SM[\child]}v \;\ge\; \SM[\child]v$ 
  holds for every nonnegative vector $v$. 
\end{proof}

In light of Prop.~\ref{prop:autored}, any reducible autocatalytic CS matrix
$\SM[\child]$ can always be decomposed into a cascade originating from an
irreducible autocatalytic Metzler matrix $\SM[\child']$. Moreover,
  Prop.~\ref{prop:Metzlerpartaut} thus justifies our focus on irreducible
autocatalytic Metzler CS matrices.

\section{Autocatalytic K{\"o}nig graphs}
\label{sect:autokoenig}

\subsection{Child-selective subgraphs}

We first turn our attention to identify substructures in the K{\"o}nig
graph that induce child-selections (CS). In the following, we will be
  concerned with subgraphs $\king'$ of $\king$. We write
$V(\king')=X(\king')\cupdot R(\king')$, and $E(\king')=E_1(\king')\cup
E_2(\king')$ for the vertex and edge set of $\king'$, respectively.
Moreover, we denote the set of vertices incident with the edges in any edge
set $E_i$ by $V(E_i)$. Note that $\king'$ is not necessarily an
  induced subgraph, i.e., $E(\king')$ does not necessarily include
all edges $e\in \king$ between two selected vertices $x,r\in V(\king')$.
\begin{definition}
  \label{def:CSive}
  A subgraph $\king'$ of $\king$ is \emph{child-selective} if there exists
  a map $\kappa: X(\king')\rightarrow R(\king')$ such that $\child =
  (X(\king'), R(\king'), \kappa)$ is a CS.
\end{definition}

Since the map $\kappa$ in a CS is bijective, $\king'$ can only be
child-selective if $\vert X(\king')\vert = \vert R(\king') \vert$. Recall
that a \emph{matching} in a graph is a set of vertex-disjoint edges, while a
\emph{perfect matching} is one incident with every vertex.

\begin{theorem}\label{thm:CSPerfMatch}
  A subgraph $\king'\subseteq \king$ is child-selective if and only if
  the subset $E_1(\king') \coloneqq \{(x,r) \ \vert \ s_{xr}^->0\}
  \subset E(\king')$ of the reactant-to-reaction edges contains a perfect
  matching.
\end{theorem}
\begin{proof}
  If $\king'$ is child-selective with CS $\child$, we directly obtain a
  perfect matching by $E_{\child}\coloneqq \{(x,\kappa(x)) \ \vert \ x\in
  X(\king')\}\subseteq E_1$. Conversely, let $M\subseteq E_1(\king')
  \subseteq E_1$ be a perfect matching in $\king'$, then there is
  $(x,y)\in M$ for all $x\in X(\king')$ and $y$ is uniquely defined for
  every $x$. Thus $\kappa: X(\king')\to R(\king')$ with $\kappa(x)=y$ if
  $(x,y)\in M$ is uniquely defined and injective. Moreover, there is $u$
  with $(u,v)\in M$ for all $v\in R(\king')$. Hence $\kappa$ is a
  bijection, and  $\child\coloneqq (X(\king'), R(\king'), \kappa)$ is
  a CS.
\end{proof}

\begin{figure}
\centering
  \begin{minipage}[c]{0.3\columnwidth}
  \centering
    \begin{tikzpicture}
      \node[draw, circle] (x) at (0,0) {$x_1$};
      \node[draw, circle] (y) at (1.5,1.5) {$x_2$};
      \node[draw, circle] (z) at (3,0) {$x_3$};
      
      \node[draw, rectangle] (a) at (0,1.5) {$r_1$};
      \node[draw, rectangle] (b) at (3,1.5) {$r_2$};
      \node[draw, rectangle] (c) at (1.5,0) {$r_3$};
      
      \draw[->] (x) -- (a);
      \draw[->] (a) -- (y);
      \draw[->] (y) -- (b);
      \draw[->] (b) -- (z);
      \draw[->] (z) -- (c);
      \draw[->] (c) -- (x);
    \end{tikzpicture}
  \end{minipage}%
  \hfill%
  \begin{minipage}[c]{0.3\columnwidth}
    \centering
    \begin{tikzpicture}
      \node[draw, circle] (x) at (0,0) {$x_1$};
      \node[draw, circle] (y) at (1.5,1.5) {$x_2$};
      \node[draw, circle] (z) at (3,0) {$x_3$};
      
      \node[draw, rectangle] (a) at (0,1.5) {$r_1$};
      \node[draw, rectangle] (b) at (3,1.5) {$r_2$};
      \node[draw, rectangle] (c) at (1.5,0) {$r_3$};
		
      \draw[->] (x) -- (a);
      \draw[->] (a) -- (y);
      \draw[->] (y) -- (b);
      \draw[->] (b) -- (z);
      \draw[->] (z) -- (c);
      \draw[->] (c) -- (x);
      \draw[->] (x) -- (b);
    \end{tikzpicture}
  \end{minipage}%
  \hfill%
  \begin{minipage}{0.3\columnwidth}
    \centering
    $\SM[\child] = 
    \begin{pmatrix}
      -1 &  -1  &  1   \\
      1 &  -1  &  0   \\
      0 &   1  & -1   \\
    \end{pmatrix}$
  \end{minipage}%
  \hfill%
  \begin{minipage}{0.44\columnwidth}
  \end{minipage}
  \caption{In general, $\king(\child)$ \textbf{(left)} is a proper
    subgraph of the induced subgraph $\king[\child]$
    \textbf{(middle)}. The example corresponds to the CS matrix
    $\SM[\child]$ shown on the \textbf{right} with columns ordered as
    $\mathsf{x_1}$, $\mathsf{x_2}$, $\mathsf{x_3}$.}
      \label{fig:CSGraphs}
\end{figure}

The proof of Thm.~\ref{thm:CSPerfMatch} contains an explicit recipe to
construct CS. In fact, there is a 1-1 correspondence between perfect
matchings in $E_1(\king')$ and bijections $\kappa: X(\king')\to
R(\king')$. Moreover, the spanning subgraph $\king'' \subseteq \king'$ with
$X(\king'')=X(\king')$, $R(\king'')=R(\king')$ and $E(\king'')=M\cupdot
E_2(\king')$ is child-selective for every perfect matching $M\subseteq
E_1(\king')$. Conversely, for any CS $\child=(X_\kappa, R_\kappa, \kappa)$,
let $E_1^{\kappa}$ and $E_2^{\kappa}$ be subsets of $E_1$ and $E_2$, 
respectively, where edges have both adjacent vertices in $(X_\kappa \cupdot
R_\kappa)$. Then we write
\begin{equation}
  \label{eq:kingchild}
  \king(\child) \coloneqq (X_\kappa \cup R_\kappa, M_\kappa \cup
  E_2^{(\kappa)})
\end{equation}
defined by the perfect matching $M_\kappa \subseteq E_1^{(\kappa)}$. We
note that $\king(\child)$ is a spanning subgraph of the induced subgraph
$\king[\child]=\king[X_\kappa \cup R_{\kappa}, E_1^{(\kappa)}\cup
  E_2^{(\kappa)}]$, see Fig.~\ref{fig:CSGraphs} for an example. In the
following subsections, we will exclusively investigate $\king(\child)$. We
return to the induced subgraphs $\king[\child]$ in
Sec.~\nameref{ssect:induced} only.

The subgraphs $\king(\child)$ naturally fit together with the Metzler
  matrices introduced in the previous section. Given a CS
  $\child=(X',R',\kappa)$, we have a path $(x,r,y)$ with $x,y\in X'$ and
  $r\in R'$ in $\king(\child)$ if and only if $r=\kappa(x)$ and $y$ is a
  product of reaction $r$, i.e., if and only if
  $\SM[\child]_{y\kappa(x)}>0$. Since we necessarily have $x\ne y$ in this
  case, and $(x,r)$ is an edge in the child-selective subgraph if and only if
  $r=\kappa(x)$, we observe that the non-zero off-diagonal entries of
  $\Metzler{\SM[\child]}$ uniquely determines the edge set
  $E_2(\king(\child))$, while $\kappa$, as we know, determines 
  $E_1(\king(\child))$.

\begin{corollary}
  For a child-selection $\child$, the subgraph $\king(\child)$ is the
  directed bipartite K\"onig graph determined by $\Metzler{\SM[\child]}$.
\end{corollary}

Moreover, each CS $\child'$ for the induced subgraph $\king[\child]$
gives rise to a distinct subgraph $\king(\child')$, see
Fig.~\ref{fig:ambiguous-core}(left). We note in passing that
polynomial-delay algorithms exist for enumerating perfect matchings in
bipartite graphs
\cite{fukuda_finding_1994,uno_algorithms_1997,fink_constant_2025}. In the
present work, however, we adopt a different approach to constructing the
relevant child-selective subgraphs of $\king$.  The following statement is
a direct consequence of the fact that $\kappa$ is a bijection:
\begin{corollary}\label{cor:DegK(k)}
  Let $\child$ be a CS. Then every substrate vertex $x\in X_{\kappa}$ in
  $\king(\child)$ has out-degree $1$ and every reaction vertex $r\in
  R_{\kappa}$ in $\king(\child)$ has in-degree $1$.
\end{corollary}

The subsequent results provide us with a purely graph-theoretical
characterization of the subgraphs of $\king$ that derive from child
selections.
\begin{lemma}[Proof: SI] 
  \label{lem:CSgraph}
  Let $\king'=(X'\cupdot R',E_1'\cupdot E_2')$ be a subgraph of $\king$ with
  reactant vertices $X'$, reaction vertices $R'$, and edges $E_1\subseteq
  X'\times R'$ and $E_2'\subseteq R'\times X'$ 
  such that
  \begin{enumerate}
  \item $|X'|=|R'|$;
  \item every $x\in X'$ has out-degree $1$ and every $x\in R'$ has
    in-degree $1$.
  \end{enumerate}
  Then $\king'$ is child-selective with $\kappa(x)=r$ for
  $(x,r)\in E_1'$. 
\end{lemma}

\begin{figure}
\centering
  \hspace*{-1.6em}
  \begin{minipage}[c]{0.47\columnwidth}
    \centering
    \begin{tikzpicture}[scale=0.88, transform shape]
      \node[draw, circle] (x) at (-2,0) {$x_2$};
      \node[draw, rectangle] (a) at (-1,1) {$r_1$};
      \node[draw, circle, rounded corners] (y) at (0,0) {$x_3$};
      \node[draw, rectangle] (b) at (-1,-1) {$r_2$};
      
      \node[draw, circle, rounded corners] (z) at (1,1) {$x_1$};
      \node[draw, rectangle] (c) at (1,2) {$r_3$};
      
      \draw[->, rounded corners] (a) -- (x);
      \draw[->, rounded corners] (b) -- (y);
      \draw[->, rounded corners] (c) -- (z);
      
      \draw[->, blue, rounded corners] (x) -- (b);
      \draw[->, blue, rounded corners] (y) -- (c);
      \draw[->, blue, rounded corners] (z) -- (a);
      
      \draw[->, green, rounded corners] (x) to [bend left = 45] (c);
      \draw[->, green, rounded corners] (y) --  (a);
      \draw[->, green, rounded corners] (z) to [bend left = 45] (b);
    \end{tikzpicture}
  \end{minipage}%
  \begin{minipage}[c]{0.47\columnwidth}
  \centering
    \begin{tikzpicture}[scale=0.88, transform shape]
      \node[draw, circle] (a) at (0,0) {$x_1$};
      \node[draw, rectangle] (b) at (1,1) {$r_1$};
      \node[draw, circle] (c) at (2,0) {$x_2$};
      \node[draw, rectangle] (d) at (1,-1) {$r_2$};
  
      \node[draw, circle, rounded corners] (e) at (3,0) {$x_3$};
      \node[draw, rectangle] (f) at (4,1) {$r_3$};
      \node[draw, circle, rounded corners] (g) at (5,0) {$x_4$};
      \node[draw, rectangle] (h) at (4,-1) {$r_4$};
      
      \draw[->, red, rounded corners] (a) -- (b);
      \draw[->, rounded corners] (b) -- (c);
      \draw[->, red, rounded corners] (c) -- (d);
      \draw[->, rounded corners] (d) -- (a);
      
      \draw[->, red, rounded corners] (e) -- (f);
      \draw[->, rounded corners] (f) --  (g);
      \draw[->, red, rounded corners] (g) -- (h);
      \draw[->, rounded corners] (h) -- (e);
      
      \draw[->, blue, rounded corners] (c) to [bend left = 20] (f);
      \draw[->, green, rounded corners] (c) to [bend right = 20] (h);	
      \draw[->, blue, rounded corners] (e) to [bend left = 20] (d);
    \end{tikzpicture}
  \end{minipage}%
  \caption{\textbf{Left:} \emph{Multiple CS may exist in a given induced
    subgraph of $\king$.} In the example, there are indeed two perfect
    matchings and thus two CS $\child_1=(X_\kappa,R_\kappa,\kappa_1)$ and
    $\child_2(X_\kappa,R_\kappa,\kappa_2)$: $\kappa_1(x_1)=r_1$,
    $\kappa_1(x_2)=r_2$, $\kappa_1(x_3)=r_3$ and $\kappa_2(x_1)=r_2$,
    $\kappa_2(x_2)=r_3$, $\kappa_2(x_3)=r_1$.  \textbf{Right:} \emph{A
    strongly connected child-selective subgraph may not have a CS matrix
    with irreducible Metzler part.}  Here, choosing, $\kappa(x_1)=r_1$,
    $\kappa(x_2)=r_2$, $\kappa(x_3)=r_3$, and $\kappa(x_4)=r_4$ (edges
    depicted in red) yields a block diagonal $\Metzler{\SM[\kappa]}$
    composed of two $2\times 2$ blocks.}
  \label{fig:ambiguous-core} 
\end{figure}
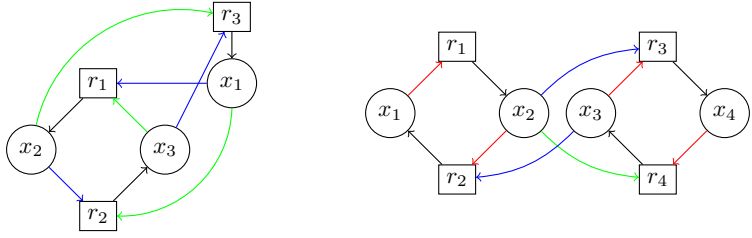 

As an immediate consequence of Thm.~\ref{thm:CSPerfMatch} we note:
\begin{corollary}
  Let $\king$ be an even elementary circuit graph. Then $\king$ is
  child-selective with a unique child-selection $\child$. Moreover,
  $\king(\child)=\king$.
\end{corollary}
\begin{corollary}
  If $\king$ is not connected, then it is child-selective if and only if
  each weakly connected component is child-selective.
\end{corollary}

\subsection{Irreducibility and Strong Connectedness}\label{subsec:IrrStr}

We start this section with two simple technical observations:
\begin{lemma}[Proof: SI] 
  \label{lem:scirreduc}
  Let $\child=(X_\kappa,E_\kappa,\kappa)$ be a CS. Then $\king(\child)$ is
  strongly connected if and only if $\Metzler{\SM[\child]}$ is irreducible.
\end{lemma}
\begin{lemma}[Proof: SI] 
  \label{lem:cutvertex}
  If $\king(\child)$ is strongly connected, then it does not contain a
  cut vertex
\end{lemma}
Thus, if $\king(\child)$ is strongly connected, then its underlying
undirected graph is also 2-connected. Such graphs are called \emph{strong
blocks} in the literature \cite{grotschel_minimal_1979}. 
Combining Lemmas~\ref{lem:scirreduc}
and \ref{lem:cutvertex} yields the main result of this section:
\begin{theorem}\label{thm:IrrerMetzler}
  Let $\child$ be a CS. Then $\king(\child)$ is a strong block if and only
  if $\Metzler{\SM[\child]}$ is irreducible.
\end{theorem}

\begin{figure}
\centering
\begin{minipage}[c]{0.47\columnwidth}
\centering
  \begin{tikzpicture}[scale=0.8, transform shape]
    \node[draw, circle] (a) at (0,0) {$x_1$};
    \node[draw, rectangle] (b) at (1,1) {$r_1$};
    \node[draw, circle, rounded corners] (c) at (2,0) {$x_2$};
    \node[draw, rectangle] (d) at (1,-1) {$r_2$};
    \node[draw, rectangle] (e) at (-1,1) {$r_3$};
    \node[draw, circle, rounded corners] (f) at (-2,0) {$x_3$};
    \node[draw, rectangle] (g) at (-1,-1) {$r_4$};
 
    \draw[->, rounded corners] (a) -- (b);
    \draw[->, rounded corners] (b) -- (c);
    \draw[->, rounded corners] (c) -- (d);
    \draw[->, rounded corners] (d) -- (a);
    \draw[->, rounded corners] (a) -- (e);
    \draw[->, rounded corners] (e) -- (f);
    \draw[->, rounded corners] (f) to (g);
    \draw[->, rounded corners] (g) -- (a);
    
    \draw[->, rounded corners] (b) to [bend right=60, looseness =2] (f);
    \draw[->, rounded corners] (g) to [bend right=60, looseness =2] (c);
  \end{tikzpicture}
\end{minipage}%
\begin{minipage}[c]{0.47\columnwidth}
\centering
  \begin{tikzpicture}[scale=0.8, transform shape]
    \node[draw, circle] (a) at (0,0) {$x_3$};
      \node[draw, rectangle] (b) at (1,1) {$r_1$};
      \node[draw, circle, rounded corners] (c) at (3,1) {$x_2$};
      \node[draw, rectangle] (d) at (3,-1) {$r_2$};
      \node[draw, rectangle] (e) at (-1,1) {$r_3$};
      \node[draw, circle, rounded corners] (f) at (-2,0) {$x_4$};
      \node[draw, rectangle] (g) at (-1,-1) {$r_4$};
      \node[draw, rectangle] (h) at (1,-1) {$x_1$};
		
      \draw[->, rounded corners] (a) -- (b);
      \draw[->, rounded corners] (b) -- (c);
      \draw[->, red, rounded corners] (c) -- (d);
      \draw[->, rounded corners] (d) to [bend left = 60, looseness = 1.5] (a);
      \draw[->, red, rounded corners] (a) -- (e);
      \draw[->, rounded corners] (e) -- (f);
      \draw[->, red, rounded corners] (f) --  (g);
      \draw[->, rounded corners] (g) -- (a);
      \draw[->, rounded corners] (d) -- (h);
      \draw[->, red, rounded corners] (h) -- (b);
  \end{tikzpicture}
\end{minipage}
\caption{\textbf{Left:} Strongly-connected bipartite graph without a
  cut-vertex but not child-selective, as the four reaction vertices are
  more than the three species vertices. \textbf{Right:} Strongly-connected,
  child-selective (red edges) bipartite digraph with a cut-vertex.}
\label{fig:NonCSStrCon}
\end{figure}
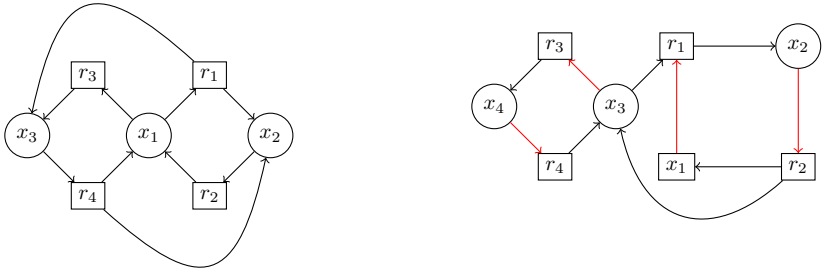

The delicate interplay between the notions of child-selectiveness and
strong blocks is exemplified in Fig.~\ref{fig:NonCSStrCon} and
Fig~\ref{fig:autoCore}. In particular, Fig.~\ref{fig:autoCore}b and
Fig.~\ref{fig:autoCore}c depict graphs that are strongly connected, albeit
not strong blocks, but not child-selective. Fig.~\ref{fig:autoCore}d and
Fig.~\ref{fig:autoCore}e exemplify strong-blocks that are not
child-selective.

We conclude this section by revisiting the notion of autocatalysis and
expressing Def.~\ref{def:autocatmatrix} in terms of the following
graph-theoretic characterization of an autocatalytic CS. We recall the
standard notions of \emph{source} and \emph{sink} vertices in a directed
graph, i.e. vertices with no incoming edges (zero in-degree)
and with no outgoing edges (zero out-degree), respectively.
\begin{lemma}[Proof: SI] 
  \label{lem:nosourcesink}
  A CS $\child=(X_\kappa,E_\kappa,\kappa)$ is autocatalytic if and only if
  the following conditions both hold:
  \begin{enumerate}
  \item there is a positive vector $v>0$ such
    that $\SM[\child]v >0$.
  \item 
    $\king(\child)$ does not possess source and sink vertices;
  \end{enumerate} 
\end{lemma}

\subsection{Fluffles and Circuitnets} 

Let us summarize the main arguments in the discussion so far. (1) By
Prop.~\ref{prop:Metzlerpartaut}, we may restrict our attention to the
Metzler parts of CS matrices. (2) By Prop.~\ref{prop:autored}, an arbitrary
autocatalytic CS matrix is either irreducible or it contains one or more
disjoint irreducible autocatalytic blocks, so we may focus on CS
matrices with irreducible Metzler part. (3) By Lemma~ \ref{lem:scirreduc}
and Thm.~\ref{thm:IrrerMetzler}, CS matrices with irreducible Metzler part
correspond to strong blocks $\king(\child)$. (4) Using
Cor.~\ref{cor:DegK(k)}, Lemma~\ref{lem:CSgraph}, and
Thm.~\ref{thm:IrrerMetzler}, we finally derive the following central
observation:
\begin{proposition}
  \label{prop:usefluffles} 
  Let $G$ be a subgraph of $\king$. Then there is a CS
  $\child$ with a CS matrix that has an irreducible Metzler part
  $\Metzler{\SM[\child]}$ such that $G=\king(\child)$ if and only if $G$
  satisfies
  \begin{itemize}
  \item[(i)] $G$ is bipartite with vertex partition $V(G)=X(G)\cup R(G)$
    such that $|X(G)|=|R(G)|$,
  \item[(ii)] if $x\in X'$ then $x$ has out-degree $1$ and if $r\in R'$, then
    $r$ has in-degree $1$
  \item[(iii)] $G$ is a strongly connected block.
  \end{itemize}
\end{proposition} 
We call a graph satisfying these three properties a
\emph{fluffle}\footnote{\emph{Fluffle} is an informal, whimsical term for a
group of rabbits, motivated by the rabbit's ears in our context.}, and we
denote such subgraphs with $G$ throughout. Prop.~\ref{prop:usefluffles} in
particular implies that the enumeration of fluffle subgraphs in the
K{\"o}nig graph $\king$ encompasses all autocatalytic CS matrices with
irreducible Metzler part. In the following, we describe how the fluffles
can be constructed recursively.

Strongly connected blocks are precisely those graphs that admit an
\emph{open directed ear decomposition} \cite{grotschel_minimal_1979}, i.e.,
they can be constructed from an elementary circuit by iteratively adding
open directed ears, which are directed paths whose endpoints are distinct
vertices already present in the graph. Throughout the remainder of this
contribution, we refer to open directed ear decompositions and open
directed ears simply as ``ear decompositions'' and ``ears'',
respectively. After attaching an ear, the interior vertices of the ear have
in-degree and out-degree $1$, while its initial vertex has out-degree $>1$
and its terminal vertex has in-degree $>1$. In our bipartite setting, we
then have the following theorem characterizing fluffle graphs.
\begin{theorem}[Proof: SI] 
  \label{thm:Woffle}
  A graph $G$ is a fluffle if and only if it is bipartite with vertex set
  $X\cupdot R$ and it has an ear decomposition such that every ear
  initiates in a reaction vertex $r\in R$ and terminates in a substrate
  vertex $x\in X$. In this case, all directed open ear decompositions have
  this property.
\end{theorem}
It will be useful to note that strong blocks in fluffles
are again fluffles themselves:
\begin{lemma}[Proof: SI] 
  \label{lem:subwoffle}
  Let $G$ be a fluffle in $\king$ and $G'$ a subgraph of $G$ that
  is a strong block. Then $G'$ is a fluffle.
\end{lemma}

Clearly, every elementary circuit in $\king$ is a fluffle. It is therefore
natural to identify fluffles as unions of elementary circuits. According to
Thm.~\ref{thm:Woffle}, however, such unions must be consistent with an ear
decomposition in which each ear originates at a reaction vertex and
terminates at a substrate vertex. This requirement is made explicit in the
following theorem. See also Fig.~\ref{fig:autoCore} for an overview of
possible ways to combine two elementary circuits, only one of which
actually constitutes a fluffle.

\begin{figure}
  \centering
  \includegraphics[width=\columnwidth]{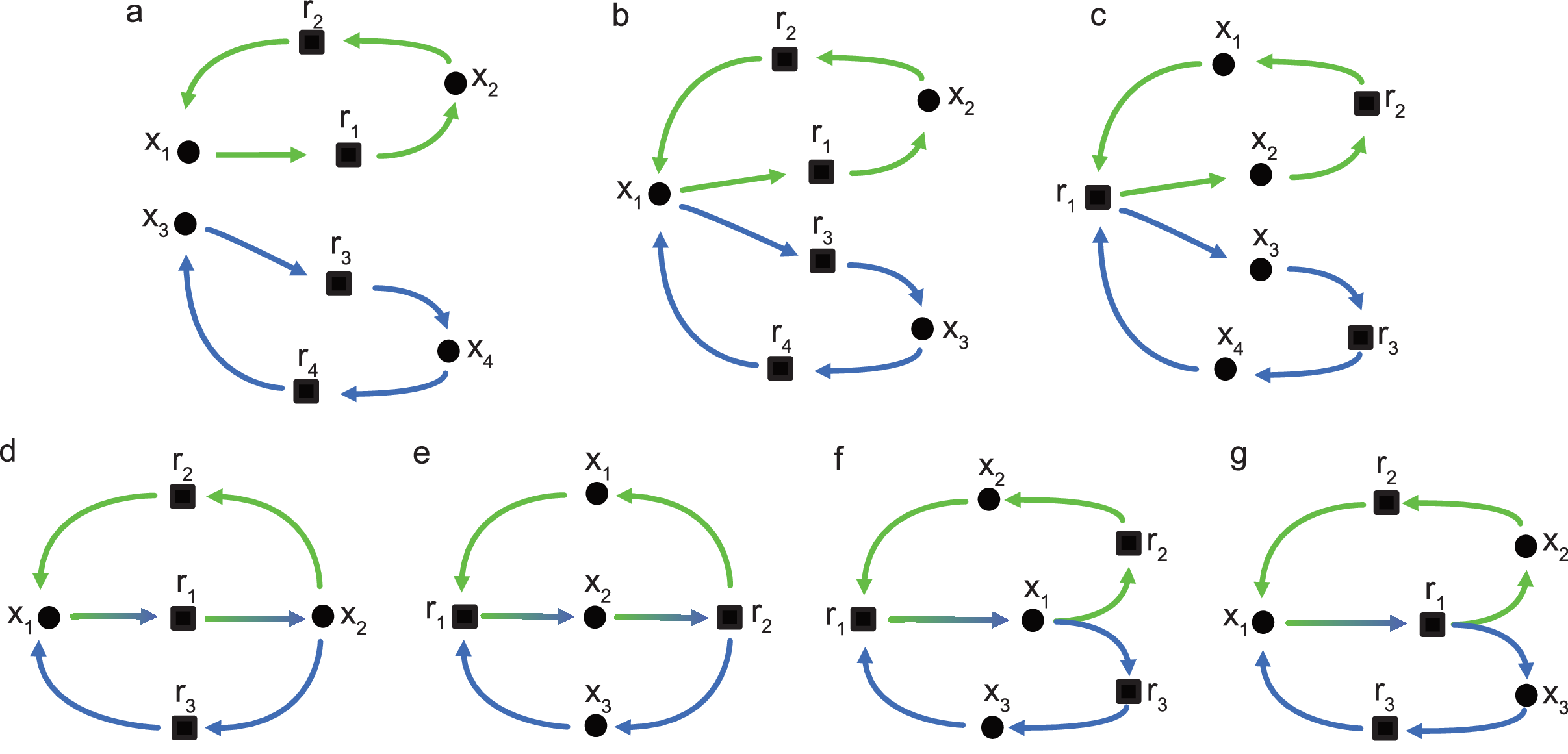}
  \caption{Unions of two elementary circuits, depicted in green and blue,
    yielding different configurations: \textbf{a} is not connected and
    thereby violates the fluffle condition $(iii)$ in
    Prop.~\ref{prop:usefluffles}, \textbf{b-e} possess more metabolites
    than reaction vertices or vice-versa, and contradict fluffle condition
    $(i)$ in Prop.~\ref{prop:usefluffles}. In addition, \textbf{b} and
    \textbf{c} are also not a strong block and thus violate $(iii)$ as
    well. The union of two elementary circuits as depicted in
    \textbf{f} has a substrate vertex $x_1$ with out-degree two and a
    reaction vertex $r_1$ with in-degree two, which contradicts fluffle
    condition $(ii)$ in Prop.~\ref{prop:usefluffles}.  The only combination
    consistent with the fluffle definition is depicted in \textbf{g}, where
    indeed the intersection of the two elementary circuits is a path from a
    substrate to a reaction vertex, as prescribed by
    Thm.~\ref{thm:woffle-union}. The two elementary circuits in \textbf{d},
    \textbf{e}, \textbf{f}, moreover, also, constitute an example of
    circuitnets that are not associated with fluffles.}
\label{fig:autoCore}
\end{figure}

\begin{theorem}[Proof: SI] 
  \label{thm:woffle-union}
  Let $G$ be a fluffle with vertex partition $X\cupdot R$ and $C$ an
  elementary circuit such that $\emptyset \subset G\cap C \subset C$. Then,
  the connected components of $G\cap C$ are directed paths $P_i$. Moreover,
  $G\cup C$ is a fluffle if, and only if, all such paths $P_i$ start from a
  substrate vertex $x_i \in X$ and terminate with a reaction vertex $r_i
  \in R$.
\end{theorem}

\begin{definition}\label{def:cnet}
  A set $\mathcal{C}=\{C_1,\dots,C_n\}$ of elementary circuits in
  $\king$ is a \emph{circuitnet} if there is an ordering $\pi$ such that
  the union $G_k\coloneqq \bigcup_{i=1}^{k}C_{\pi(i)}$ is a strong block
  for all $1\le k\le n$. 
\end{definition}
We say that $\mathcal{C}$ is a circuitnet for a graph $G$ if $G$, as a
graph, is the union of all the elementary circuits in $\mathcal{C}$, and we
write $G=\bigcup(\mathcal{C})$. The next theorem guarantees that there is a
circuitnet for any fluffle $G$.

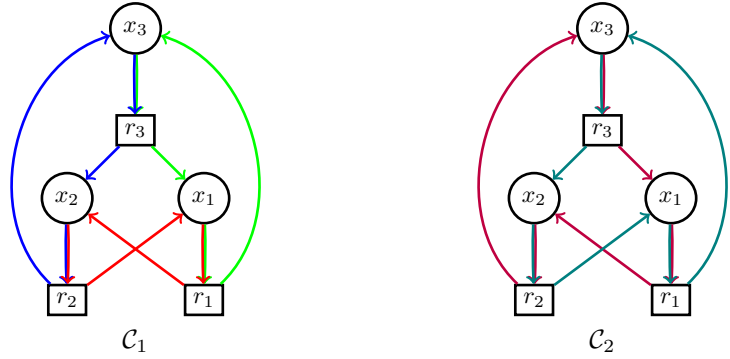
\begin{figure}
\centering
  \begin{minipage}[c]{0.47\columnwidth}
    \centering
    \begin{tikzpicture}[scale=0.9, transform shape, line width = 1pt]
      \node[draw, rectangle] (a) at (1,0) {$r_1$};     
      \node[draw, rectangle] (b) at (-1,0) {$r_2$};
      \node[draw, rectangle] (d) at (0,2.5) {$r_3$};
      \node[draw, circle] (w) at (0,4) {$x_3$};
      \node[draw, circle] (x) at (1,1.5) {$x_1$};            
      \node[draw, circle] (y) at (-1,1.5) {$x_2$};
      \draw[->, green] (a) to [bend right =60] (w);
      \draw[->, green] (w) to  [out = 272, in = 88] (d);
      \draw[->, green] (d) -- (x);
      \draw[->, green] (x) to [out = 272, in = 88] (a);
      \draw[->, blue] (w) to [out = 268, in = 92] (d);
      \draw[->, blue] (d) -- (y);
      \draw[->, blue] (y) to [out = 268, in = 92] (b);
      \draw[->, blue] (b) to [bend left = 60] (w);     
      \draw[->, red] (a) -- (y);
      \draw[->, red] (y) to [out = 272, in = 88] (b);
      \draw[->, red] (b) -- (x);
      \draw[->, red] (x) to [out = 268, in = 92] (a);
    \end{tikzpicture}
    \par\noindent
    $\mathcal{C}_1$
  \end{minipage}%
  \begin{minipage}[c]{0.47\columnwidth}
    \centering
    \begin{tikzpicture}[scale=0.9, transform shape, line width = 1pt]
      \node[draw, rectangle] (a) at (1,0) {$r_1$};     
      \node[draw, rectangle] (b) at (-1,0) {$r_2$};
      \node[draw, rectangle] (d) at (0,2.5) {$r_3$};
      \node[draw, circle] (w) at (0,4) {$x_3$};
      \node[draw, circle] (x) at (1,1.5) {$x_1$};            
      \node[draw, circle] (y) at (-1,1.5) {$x_2$};
      \draw[->, purple] (b) to [bend left = 60] (w);
      \draw[->, purple] (w) to  [out = 272, in = 88] (d);
      \draw[->, purple] (d) -- (x);
      \draw[->, purple] (x) to [out = 272, in = 88] (a);
      \draw[->, purple] (a) -- (y);        
      \draw[->, purple] (y) to [out = 272, in = 88] (b);
      \draw[->, teal] (a) to [bend right =60] (w);
      \draw[->, teal] (w) to [out = 268, in = 92] (d);
      \draw[->, teal] (d) -- (y);
      \draw[->, teal] (y) to [out = 268, in = 92] (b);
      \draw[->, teal] (b) -- (x);
      \draw[->, teal] (x) to [out = 268, in = 92] (a);	
    \end{tikzpicture}
    \par\noindent
    $\mathcal{C}_2$
  \end{minipage}%
  \vspace*{3mm}
  \caption{The fluffle corresponding to an autocatalytic core of Type V,
    Eq.~\eqref{eq:coretype5}. The fluffle $F$ is depicted as a union of the
    elementary circuits in two \emph{different} circuitnets.  The
    circuitnet $\mathcal{C}_1=\{C_1,C_2,C_3\}$ \textbf{(left)} comprises
    the three elementary circuits $C_1=(x_1,r_1,x_3,r_3,x_1)$
    (green), $C_2=(x_2,r_2,x_3,r_3,x_2)$ (blue), and
    $C_3=(x_1,r_1,x_2,r_2,x_1)$ (red); the circuitnet
    $\mathcal{C}_2=\{C_4,C_5\}$ \textbf{(right)} consists of the two
    elementary circuits $C_4=(x_1,r_1,x_3,r_3,x_2,r_2,x_1)$ (teal) and
    $C_5=(x_1,r_1,x_2,r_2,x_3,r_3,x_1)$ (purple).  We have
    $F=\bigcup(\mathcal{C}_1)=\bigcup(\mathcal{C}_2)$.}
  \label{fig:altFluf}
\end{figure}
  
\begin{theorem}\label{thm:flufflecnets}
  Let $G$ be a fluffle. Then there exists a circuitnet $\mathcal{C}$ for
  $G$, i.e.,
\begin{equation}
 G=\bigcup(\mathcal{C}).  
\end{equation}
\end{theorem}
\begin{proof}
  The statement follows directly from the well-known connection of
  elementary circuits, ear decompositions, and cycle bases: The cycle
  space of strongly-connected digraphs has a circuit basis \cite{Berge:73},
  and for strong blocks, such a basis can be constructed from an ear
  decomposition by completing each ear $P_i$ to a directed circuit $C_i$
  using any directed path in $G_{i-1}$ from the terminal to the initial
  vertex of the ear. In particular, therefore, every fluffle has a
  circuitnet.
\end{proof}

Although there is at least one circuitnet for a fluffle $G$,
in general, there may exist more circuitnets for the same fluffle. An
example is depicted as the autocatalytic core of type V,
\eqref{eq:coretype5}: see Fig.~\ref{fig:altFluf}. A simple corollary
follows from Thm.~\ref{thm:flufflecnets}.
\begin{corollary}\label{cor:superimposingfluffle}
  Let $\mathcal{C}$ be a circuitnet for the fluffle $G$. Then the following
  hold true.
  \begin{enumerate}
  \item $\mathcal{C}'\subseteq\mathcal{C}$ is a circuitnet for a fluffle
    $G'\subseteq G$ if and only if $G'$ is a strong block.
  \item There exist an ordering $\pi$ of the circuits in $\mathcal{C}$ such
    that $G_k=\bigcup_{i=1}^k C_{\pi(i)}$ is a fluffle.
\end{enumerate}
\end{corollary}
\begin{proof}
The first statement follows directly from Lemma~\ref{lem:subwoffle} and the
second statement from the first and Def.~\ref{def:cnet}.
\end{proof}
 
\subsection{The set of fluffles of a CRN}

Summarizing the discussion so far, we have shown that all CS matrices in
$\Gamma$ with an irreducible Metzler part, i.e., the viable candidates for
``interesting'' autocatalytic subnetworks, are the fluffles in the
associated bipartite König graph $\king$ (Prop.~\ref{prop:usefluffles}). 
Moreover, any fluffle can be constructed by
superimposing elementary circuits such that each intermediate step is
itself a fluffle (Cor.~\ref{cor:superimposingfluffle}). Equivalently, this
boils down to enumerating the sets of circuitnets whose union is a fluffle.
Elementary circuits can be enumerated efficiently in a lazy manner
\cite{tarjan_enumeration_1973,johnson_finding_1975,szwarcfiter_search_1976}
with linear delay, i.e., $O(|V|+|E|)$. Moreover, algorithms exist that
allow to restrict circuit length \cite{gupta_finding_2021}.

Denote by $\mathfrak{F}$ the set of circuitnets
$\mathcal{C}$ 
whose union $\bigcup(\mathcal{C})$ is a
fluffle. In SI Sec.~\emph{\nameref{sect:flufflesetsystems}} we summarize
  the properties of $\mathfrak{F}$ that enable efficient enumeration.  The
example in Fig.~\ref{fig:altFluf} suggests to investigate the equivalence
relation on $\mathfrak{F}$ defined by $\mathcal{C}_1\sim \mathcal{C}_2$ iff
$\bigcup(\mathcal{C}_1)=\bigcup(\mathcal{C}_2)$. These equivalence classes
are specified by subsets of edges in $\king$.  More precisely, we have
$\mathcal{C}_1\sim \mathcal{C}_2$ if and only if
$E(\bigcup(\mathcal{C}_1))=E(\bigcup(\mathcal{C}_2))$ because the
corresponding vertex set is given implicitly by the vertices incident with
these edge sets. Our primary aim, however, is not to enumerate fluffles,
but to enumerate CS matrices $\SM[\child]$ with irreducible Metzler
parts. To this end, we recall that the CS $\child$ is equivalent to the
substrate-to-reaction edges in a fluffle, hence these edges completely
determine $\SM[\child]$. This suggests to consider the following
equivalence relation.
\begin{definition}
  \label{def:csequivalence}
  Let $\mathcal{C}_1$ and $\mathcal{C}_2$ be two circuitnets for fluffles
  $G_1$ and $G_2$. We say that $\mathcal{C}_1$ and $\mathcal{C}_2$ are
  \emph{CS-equivalent} and we write $\mathcal{C}_1 \bumpeq \mathcal{C}_2$
  if
  \begin{equation}\label{eqn:bumpeq}
    E_1\left(\bigcup(\mathcal{C}_1)\right)=
    E_1\left(\bigcup(\mathcal{C}_2)\right).
  \end{equation}
\end{definition}
Let now $\child_1=(X_{\child_1},R_{\child_1},\kappa_1)$ and
$\child_2=(X_{\child_2},R_{\child_2},\kappa_2)$ be two CS. We recall that we
consider two CS matrices $\SM[\child_1]$ and $\SM[\child_2]$ are the 
\emph{same} if 
\begin{equation}\label{eq:samecs}
  X_{\child_1}=X_{\child_2},\quad R_{\child_1}=R_{\child_2},\quad
  \kappa_1=\kappa_2.
\end{equation}
In other words, we consider two CS matrices the same if they involve the
same species, reactions, and bijection between species and reactions. We do
not require that the ordering of the rows is the same. That is, we consider
$\SM[\child_1]$ and $\Pi^{\top}\SM[\child_2]\Pi$ to be same for any
permutation matrix $\Pi$ on $X_{\child}\times X_{\child}$.  With this notion
of same-ness, we obtain the following result.
\begin{lemma}[Proof: SI] 
  \label{lem:circuitnetequiv}
  Two circuitnets $\mathcal{C}_1$ and $\mathcal{C}_2$ for fluffles $G_1$
  and $G_2$ yield the same CS matrix $\SM[\child]$ if and only if
  $\mathcal{C}_1 \bumpeq \mathcal{C}_2$.
\end{lemma}
For example, consider the autocatalytic core of type V \eqref{eq:coretype5}
depicted in Fig.~\ref{fig:altFluf}. The two circuitnets
$\mathcal{C}_1=\{C_1,C_2,C_3\}$ and $\mathcal{C}_2=\{C_4,C_5\}$ both are
associated to the same fluffle and thus in particular they are
CS-equivalent and give rise to the same CS matrix:
\begin{equation}\label{eq:TypVAgain}
  \SM[\child] =
  \begin{pmatrix}
    -1 &  1 &  1 \\
    1 & -1 &  1 \\
    1 &  1 & -1 \\ 
  \end{pmatrix},
\end{equation}
since $E_1(\mathcal{C}_1)=E_1(\mathcal{C}_2)$. Remarkably, 28 (!)
circuitnets in the power set $\mathfrak{P}(\{C_1, C_2, C_3, C_4, C_5\})$ --
all except for the empty set and the three singletons $\{C_1\}$, $\{C_2\}$,
and $\{C_3\}$, which correspond to elementary circuits involving only two
species -- induce the same CS matrix (Eq.~\ref{eq:TypVAgain}) and therefore
belong to the same equivalence class. Notably, even distinct elementary
circuits, such as $C_4$ and $C_5$ in this example, can be CS-equivalent. En
passant, we further note that all autocatalytic cores 
$\SM[\child]$, with the sole exception of type III, admit a single-circuit 
representative in the CS-equivalence class of circuitnets associated 
with $\SM[\child]$, see SI Example~\ref{ex:notjustcycle}. The next 
Lemma is central in exploiting the CS-equivalence to lighten the 
computational cost of our approach.

\begin{lemma}[Proof: SI] \label{lem:bumpeqEquiv}
  Let $\mathcal{C}_1$ and $\mathcal{C}_2$ be circuitnets for fluffles $G_1$
  and $G_2$, respectively, and let $C'$ and $C''$ be two elementary
  circuits.  Assume $\mathcal{C}_1 \bumpeq \mathcal{C}_2$, $C'\bumpeq C''$
  and $G_1'\coloneqq \bigcup(\mathcal{C}_1\cup\{C'\})$ is a fluffle. Then
  $G_2'\coloneqq \bigcup(\mathcal{C}_2\cup\{C''\})$ is a fluffle as well
  with $\mathcal{C}_1\cup\{C'\} \bumpeq \mathcal{C}_2\cup\{C''\}$.
  \label{lem:bumpeq}
\end{lemma}
The converse of Lemma~\ref{lem:bumpeq} does not hold. A counterexample is
again the autocatalytic core of type V \eqref{eq:coretype5} depicted in
Fig.~\ref{fig:altFluf}. Consider the two circuitnets
$\mathcal{C}_2\coloneqq \{C_4,C_5\}$ and $\mathcal{C}_3\coloneqq \{C_1,
C_2\}$. Note that the circuitnets $\mathcal{C}_3$ is obtained from
$\mathcal{C}_1=\{C_1,C_2,C_3\}$ by removing the elementary circuit $C_3$,
red in Fig.~\ref{fig:altFluf}. Even if $\mathcal{C}_2\bumpeq \mathcal{C}_3$
holds, the parallel removal of any elementary circuit from the two
circuitnets destroy the equivalence relation as $C_i\not\bumpeq C_j$ for
$i=4,5$, and $j=1,2$: the elementary circuits in $\mathcal{C}_2$ comprise
three species while the elementary circuits in $\mathcal{C}_3$ comprise
two.

However, we may still ensure that a representative of each CS-equivalence
class can be reached by adding an elementary circuit to the representative
of some CS-equivalence class of a smaller fluffle. More precisely, we have
the following Lemma.
\begin{lemma}[Proof: SI] 
  For every CS-equivalence class $[\mathcal{C}]$ there is a representative
  $\hat{\mathcal{C}}$ such that there exists a CS-equivalence class
  $[\mathcal{C}']$ with representative $\hat{\mathcal{C}}'$ and an
  elementary circuit $C^*$ such that $\hat{\mathcal{C}}'\cup\{C^*\} \bumpeq
  \hat{\mathcal{C}}$ and $|V(\bigcup(\mathcal{C}'))| <
  |V(\bigcup(\mathcal{C}))|$.
  \label{lem:bumpeq_enum}
\end{lemma}

Finally, we may naturally extend the notion of CS-equivalence to
fluffles.
\begin{definition}\label{def:bumpeq}
  Two fluffles $G, G'$ are CS-equivalent if and
  only if $E_1(G)=E_1(G')$. 
\end{definition}
Moreover, we observe that each $\bumpeq$-equivalence class has a natural
fluffle representative, whose graph is indeed $\king(\kappa)$:
\begin{proposition}
  \label{prop:kingrep}
  Let $G$ be a fluffle and let $\child$ be the CS defined by
  $E_1(G)$. Consider now the graph $\bar{G}$ with vertex set
  $V(\bar{G})=V(G)$ and edge set $E_1(G)\cup E_2(V(G))$, i.e., $e \in
  E_2(V(G))$ if and only if $e=(r_1,x_2)$ with $r_1,x_2\in V(G)$. Then
  $\bar{G}$ is a fluffle with $\bar{G}\bumpeq G$, and such that
  $\bar{G}=\king(\child)$.
\end{proposition}
\begin{proof}
  Since $G$ is a fluffle of $\king$, then $E_1(G)$ is a perfect matching
  corresponding to the CS $\child$, and thus $\bar{G}=\king(\child)$ by
  Eq.~\eqref{eq:kingchild}. By Prop.~\ref{prop:usefluffles}, $\bar{G}$ is a
  fluffle as well. As $E_1(\bar G)=E_1(G)$ holds by construction, we have
  $\bar{G}\bumpeq G$.
\end{proof}
Taken together, we can avoid the enumeration of
  $\mathfrak{F}$ in favor of enumerating fluffle representatives only.

\subsection{Metzler matrices and induced fluffles}
\label{ssect:induced}

Recall that by Prop.~\ref{prop:Metzlerpartaut} $\SM[\child]$ corresponds to
the \emph{induced} subgraph $\king[\child]$, while $\Metzler{\SM[\child]}$
corresponds to the fluffle $\king(\child)$ defined on the same vertex set
$X_{\kappa}\cupdot R_{\kappa}$.  Therefore, $\king(\child)$ is always a
spanning subgraph of $\king[\child]$. However, in general, we have
$\king[\child]\neq\king(\child)$, as shown in Fig.~\ref{fig:CSGraphs}.
Nevertheless, many properties of $\king(\child)$ translate to the induced
subgraphs. In this section, we collect some of these implications, which
will be useful below in the context of autocatalytic cores.

As an immediate consequence of Lemma~\ref{lem:CSgraph} we obtain:
\begin{corollary}
  An induced subgraph $\king[X'\cupdot R']$ is child-selective if and only
  if it contains a spanning subgraph $G$ with edge set $E'$ such that
  $E'\cap (X'\times R')$ is a perfect matching in $G$.
\end{corollary}

\begin{theorem}\label{lem:PMMetzler}
  An induced subgraph $\king'\subseteq \king$ has a Metzler CS matrix
  $\SM[\child']$ if and only if $E_1(\king')$ is a perfect matching in
  $\king'$.
\end{theorem}
\begin{proof}  
  If $E_1$ is a perfect matching in $\king'$ then each $x\in X(\king')$ has
  only one outgoing edge and each $r \in R(\king')$ has only one incoming
  edge. Hence, for each $r\in R(\king')$ there is a unique $x\in X(\king')$
  with $s_{xr}^->0$, while for all $y\in X(\king')$ with $x\neq y,
  s_{yr}^-=0$. Thus $\SM[\child]=\Metzler{\SM[\child]}$ and hence a Metzler
  matrix. If, on the other hand, $\SM[\child]$ is a Metzler matrix, for
  each $r\in R(\king')$ there is exactly one $x\in X(\king')$ with
  $s_{xr}^->0$, which implies that $r$ has exactly one incoming edge given
  by $(\kappa^{-1}(r),r)$. Since $\kappa$ is a bijection, $\kappa^{-1}$ is
  in particular injective and hence these edges form a perfect matching,
  which by construction coincides with $E_1$.
\end{proof}

If $\SM[\child]=\Metzler{\SM[\child]}$ then $\king(\child)$ already
contains all edges of $\king[\child]$. Conversely, if
$\king(\child)=\king[\child]$, then $E_1$ is a perfect matching in this
subgraph and hence $\SM[\child]$ is Metzler.
\begin{corollary}\label{cor:AutMatrixKing}
  Let $\child$ be a CS. Then the following statements are equivalent:
  \begin{itemize}
  \item[(i)] $\SM[\child]$ is a Metzler matrix
  \item[(ii)] $\king[\child] = \king(\child)$.
  \end{itemize}
\end{corollary}

Since $\king(\child)$ is a spanning subgraph of $\king[\child]$, the
following to observations are straightforward:
\begin{corollary}\label{cor:ksquaresc}
  If $\child$ is a CS and $\Metzler{\SM[\child]}$ is irreducible then
  $\king[\child]$ is strongly connected.
\end{corollary}
\begin{corollary}
  If $\child$ is a CS and $\Metzler{\SM[\child]}$ is irreducible then
  $\king[\child]$ is a strong block.
\end{corollary}
The statement of Lemma~\ref{lem:cutvertex}, however, does not hold for
$\king[\child]$ (see Fig.~\ref{fig:NonCSStrCon}, right). That is, a
strongly connected induced subgraph $\king[\child]$ is not necessarily a
strong block.

Moreover, the fact that $\king[\child]$ is strongly connected does not
imply that $\Metzler{\SM[\child]}$ is irreducible. As a counterexample,
consider the disjoint union of two even elementary circuits, one containing
a substrate vertex $x_2$ and the other $x_3$
(Fig.~\ref{fig:ambiguous-core}, right). This graph is child-selective, with
\(\kappa(z)\) defined as the successor of each substrate vertex \(z\) along
its circuit. Adding the two edges \((x_2,\kappa(x_3))\) and
\((x_3,\kappa(x_2))\) (blue edges in Fig.~\ref{fig:ambiguous-core}, right) 
produces a strongly connected graph that remains child-selective under 
the same map \(\kappa\). Nevertheless, \(\king(\child)\) remains the disjoint 
union of the two circuits and thus is not strongly connected.  In this example, 
an alternative child-selection exists by defining \(\kappa'(z)=\kappa(z)\) for
\(z\notin\{x_2,x_3\}\), \(\kappa'(x_2)=\kappa(x_3)\), and
\(\kappa'(x_3)=\kappa(x_2)\), though such a construction is not always
possible. If only $(x_3, \kappa(x_2))$ of the two blue edges was
  present, the subgraph would still be strongly connected if there existed
an additional edge \((x_2,r)\) to a reaction \(r\neq \kappa(x_3)\) in the
second circuit (green edge).  Suppose there is a perfect matching \(M\)
that includes \((x_2,r)\).  Then necessarily \((\kappa^{-1}(r),r)\notin
M\). Since \(\kappa^{-1}(r)\neq x_3\), the vertex \(\kappa^{-1}(r)\) has
only one successor, \(r\), along its circuit, implying that \(M\) cannot be
a perfect matching in \(E_1\).

We can, however, rephrase Lemma~\ref{lem:nosourcesink} in terms of the
induced subgraph $\king[\child]$, making use again of the fact that
$\king(\child)$ is a spanning subgraph of $\king[\child]$:
\begin{corollary}
  A CS $\child=(X_\kappa,E_\kappa,\kappa)$ is autocatalytic if and only if
  the following conditions both hold:
  \begin{enumerate}
  \item there is a positive vector $v>0$ such that $\SM[\child]v >0$.
  \item $\king[\child]$ does not possesses source and sink vertices;
  \end{enumerate} 
\end{corollary}

Recall that autocatalytic cores are, in particular, irreducible Metzler CS
matrices. Autocatalytic cores thus satisfy
$\SM[\child]=\Metzler{\SM[\child]}$, i.e., $\king[\child]=\king(\child)$.
In other words, all autocatalytic cores correspond to \emph{induced}
fluffles. Prop.~\ref{prop:kingrep} thus implies that the natural
representative of the CS-equivalence class of an irreducible Metzler CS
matrix, and hence in particular of an autocatalytic core, is an
induced fluffle.

Finally, we show that irreducible Metzler CS matrices $\SM[\child]$ are
autocatalytic whenever they contain an autocatalytic core. Our starting
point is the following technical result, which then enables us to state the
main result of this subsection.
\begin{lemma}[Proof: SI] 
  Let $\king(\child^*)$ be a fluffle with irreducible autocatalytic
  Metzler CS matrix $\SM[\child^*]$ and let $\king(\child)$ be
  obtained from $\king(\child^*)$ by adding a single ear with initial
  vertex in $R(\king(\child^*))$, terminal vertex in $X(\king(\child^*))$,
  and a non-empty set of internal vertices, together with all
  reaction-to-metabolite edges in $R(\king(\child))\times
  X(\king(\child))$.  If $\SM[\child^*]$ is an autocatalytic CS matrix and
  $\SM[\child]$ is a Metzler matrix, then $\SM[\child]$ is autocatalytic
  irreducible CS matrix.
\label{lem:expandAutoMetzler} 
\end{lemma}
\begin{theorem}[Proof: SI] 
  Let $\SM[\child]$ be an irreducible Metzler CS matrix and suppose
  $\SM[\child]$ contains an autocatalytic core $\SM[\child^*]$ as a
  principal submatrix. Then $\SM[\child]$ is autocatalyic.
  \label{thm:avoidMetzlerCheck}
\end{theorem}

We conclude this section with a close look at autocatalytic cores. Consider a
fluffle $G$ and recall that for any fluffle $H$ with
$E_1(H)\subseteq E_1(G)$ the matrix $\SM[E_1(H)]$ is a principal submatrix
of $\SM[E_2(G)]$. Since an autocatalytic core is irreducible and thus a
superposition of elementary circuits, we immediately observe the following
corollary:
\begin{corollary}
  Let $G$ be an autocatalytic core in $\king$ and $\mathcal{C}$ a
  circuitnet for $G$. Then $\SM[E_1(C)]$ is a Metzler CS matrix for every
  $C\in\mathcal{C}$.
  \label{cor:Metzler-core}
\end{corollary}
It is therefore of interest to consider the following subclass of
elementary circuits:
\begin{definition}
  \label{def:Metzlercircuit}
  A Metzler circuit in $\king=(X\cupdot R,E)$ is an elementary circuit
  without a chord of the form $(x,r)\in X\times R$.
\end{definition}
Cor.~\ref{cor:Metzler-core} thus implies that every autocatalytic core can
be constructed from the Metzler circuits in $\king$ alone. In a recent
  follow-up paper, we show that this is indeed feasible in combination with a
  specialized algorithm for enumerating partially chordless circuits
  \cite{golnik_using_2026}.

\section{Algorithms}
\label{sec:Algorithm}

\subsection{Basic Algorithms} 
\label{ssect:basic}
\subsubsection{Elementary circuits}

Johnson's algorithm \cite{johnson_finding_1975} enumerates all elementary
circuits of a directed graph with linear delay. Since we only need an
arbitrary representative of each CS-equivalence class, it suffices to
record the sets $E_1(C)$. We denote the set of representatives by
$\mathbf{C}$.

\subsubsection{Recognition of fluffles}

Based on Thm.~\ref{thm:woffle-union} and Prop.~\ref{prop:kingrep},
the key task is to expand a representative fluffle $G$ by a representative
circuit $C$ and to test whether the union $G\cup C$ is again a fluffle. The
following result greatly simplifies this task:
\begin{lemma}[Proof: SI] 
  \label{lem:AB}
  Let $G$ be a fluffle and $C$ an elementary circuit. Then $G\cup C$
  is a fluffle if and only if $\emptyset\ne V(G)\cap V(C) = V(E_1(G) \cap
  E_1(C))$.
\end{lemma}
As a consequence of Lemma~\ref{lem:AB}, it suffices to consider only the edge
sets $E_1(G)$ and $E_1(C)$, and their incident vertex sets when constructing
representatives of CS-equivalence classes of fluffles. This
considerably simplifies the practical implementation since we do not have to
maintain graph data structures for the fluffles. By
Prop.~\ref{prop:kingrep}, we may use $\bar G=\king(\child)$ as canonical
representative of $[G]$, where $\child$ is the CS defined by $E_1(G)$.  If
desired, circuitnets for the fluffle $\bar G$ can be re-constructed in
linear time by means of a directed ear decomposition.

\subsubsection{Representatives of CS-equivalence classes}

\begin{algorithm}[tb]\small
  \caption{\small Assembly of equivalence classes.}
  \label{alg:RecEqClAss} 

    \SetKwInOut{Require}{Require}
    \SetKwInOut{Output}{Output}
    \SetKw{Continue}{continue}

    \Require{$\mathbf{C}:$ Set of elementary circuits}
    \Output {$\mathcal{E}:$ Set CS-equivalence classes} 
    $Q \leftarrow \emptyset$; $\mathcal{E} \leftarrow \emptyset$ \tcp*{Init. empty queue and Output}
    Initialize hash map $M$ \tcp*{Edge $\rightarrow$ element. circuits}   
    \For{$C\in \mathbf{C}$}{
      \If{$E_1(C) \notin \mathcal{E}$}{
        $\mathcal{E} \leftarrow \mathcal{E} \cup \{E_1(C)\}$; $Q \leftarrow Q \cup \{E_1(C)\}$
        \For{$e \in E_1(C)$}{
            \eIf{$e \in M$}{
                $M[e] \leftarrow M[e] \cup \{E_1(C)\}$\;
            }{
                $M[e] \leftarrow \{E_1(C)\}$\;
            }
        }   
      }
    }
    \While{$Q\neq\emptyset$}{
        $E_1(G) \leftarrow Q.pop()$\;
        \For{$E_1(C) \in \bigcup_{e \in E_1(G)} M[e]$ } {
        \If{$E_1(C) \subseteq E_1(G)$ or $E_1(G) \subseteq E_1(C)$ } {        
             \Continue
        }
        $V(G) \leftarrow V(E_1(G))$; $V(C) \leftarrow V(E_1(C))$\;
        \If{ $\emptyset \ne V(G)\cap V(C) = V(E_1(G) \cap E_1(C))$ }{  
             $E_1(G') \leftarrow E_1(G) \cup E_1(C)$\;
             \If{$E_1(G') \notin \mathcal{E}$}{
                 $\mathcal{E} \leftarrow \mathcal{E} \cup \{E_1(G')\}$; $Q \leftarrow Q \cup \{E_1(G')\}$\;
             }
         }
         } 
     }
\end{algorithm}

As a consequence of Lemma~\ref{lem:bumpeq_enum} and Def.~\ref{def:bumpeq},
a representative fluffle for each CS-equivalence class of
$\mathfrak{F}$ can be obtained, sparsely, by computing unions of elementary
circuits with representative fluffles of CS-equivalence classes with
fewer vertices. To this end, we start from the set $\mathbf{C}$ of
representatives of elementary circuits and initialize a queue $Q$ with
these elementary circuits. $Q$ will contain an arbitrary representative of
each CS-equivalence class of fluffles for further expansion. We
maintain a separate set $\mathcal{E}$ of all representatives as the output
of the algorithm. The queue $Q$ is processed in first-in-first-out order.

For each $G\in Q$, all representative elementary circuits $C\in\mathbf{C}$
are tested for $G\cup C$ constituting a fluffle by means of
Lemma~\ref{lem:AB}. If so, $G'=G\cup C$ serves as a representative of the
CS-equivalence class $[G']$. If there is a fluffle $G''\in
\mathcal{E}$ such that $G''\in[G']$, i.e., $E_1(G'')=E_1(G')$ then $G'$ is
discarded without changes to $Q$.  Otherwise $[G']$ is appended to $Q$.  If
a concurrent hashset to $Q$ is maintained with all representative edgesets
as keys, then the comparison of $E_1(G')$ and $E(G'')$ can be performed in
constant time.  The algorithm terminates when $Q$ is empty, i.e., all
maximal fluffles have been found. Finding overlapping elementary circuits
$C$ for each fluffle $G$ can be sped up by a hashmap $M$ linking edges in
$E_1$ to sets of circuits they are contained in. We only have to consider
pairs where at least one such edge is shared. A pseudocode for this
procedure is given in Alg.~\ref{alg:RecEqClAss}.

\subsubsection{Identification of autocatalytic matrices
  and autocatalytic cores}

Each entry in $\mathcal{E}$ obtained by Alg.~\ref{alg:RecEqClAss}
constitutes a candidate $\SM[\child]$ for an irreducible autocatalytic CS
matrix.  Since every irreducible autocatalytic CS matrix $\SM[\child]$
contains an autocatalytic core, it either satisfies
$\SM[\child]=\Metzler{\SM[\child]}$ or it strictly contains a principal
submatrix $\SM[\child']$ with this property. In the latter case, $\child'$
is a restriction of $\child$ and thus $\king(\child')$ is a proper subgraph
of $\king(\child)$. In fact, by Cor.~\ref{cor:AutMatrixKing},
$\king(\child')=\king[\child']$ must also be an induced subgraph of
$\king(\child)$. Moreover, we have $E_1(\king(\child'))\subseteq
E_1(\king(\child))$ if and only if $\child'$ is a restriction of
$\child$. Thus, we have the following necessary condition:
\begin{corollary}
  If $G$ is a fluffle that defines an autocatalytic CS matrix, then
  there is an induced subgraph $G'$ of $G$ such that $G'$ is also
  an induced subgraph of $\king$. 
\end{corollary}
On the other hand, if $\SM[E_1(G)]$ is a Metzler matrix, i.e., if $G$
is an induced fluffle representative, and $G$ contains an autocatalytic
core, then $G$ is itself autocatalytic by Thm.~\ref{thm:avoidMetzlerCheck}.

These simple observations suggest computing the Hasse diagram with respect
to set inclusion, $\Hasse(\mathcal{E})$, for sets $\mathcal{E}$ of fluffle
equivalence classes. Traversing $\Hasse(\mathcal{E})$ in bottom-up order,
one then checks, for each $\subseteq$-minimal candidate $G$ in
$\mathcal{E}$:
\begin{itemize}
\item[(a)] whether $G = \king[V(E_1(G))]$, which is equivalent to $G$
  giving rise to a Metzler CS matrix $\SM[E_1(G)]$, and, if so,
\item[(b)] whether $\SM[E_1(G)]$ is Hurwitz unstable.
\end{itemize}
A minimal element in $(\mathcal{E},\subseteq)$ that satisfies (a) and (b)
is an autocatalytic core.  Moreover, any Metzler matrix that contains an
autocatalytic submatrix is automatically autocatalytic, and in this
case, test (b) can be omitted. Taken together, an explicit test for
autocatalyticity needs to be performed only for inclusion-minimal induced
fluffle representatives that are Metzler and for non-minimal non-Metzler
matrices in $\mathcal{E}$.  If one of the conditions (a) or (b) is
violated, $G$ is removed from $\Hasse(\mathcal{E})$ and each parents of $G$
is connected to each immediate descendant of $G$.  Upon completion of the
traversal, all minimal elements in $\Hasse(\mathcal{E})$ are autocatalytic
cores and thus Metzler matrices. Moreover, all descendants of a Metzler
matrix in $\Hasse(\mathcal{E})$ are again Metzler. Similarly, all parents
of a non-Metzler matrix are again non-Metzler. For each of the non-Metzler
matrices we explicitly test whether they are autocatalytic using an LP
solver \cite{virtanen_scipy_2020} to determine whether there is a vector
$v>0$ such that $\SM[\child]v>0$.

Determining the complete structure of the Hasse diagram, however, severely
compromises performance, and simply testing all CS-equivalence classes for
their autocatalytic capacity would be more efficient. In contrast,
predecessor relations are sufficient to avoid unnecessary testing and can
be obtained without additional costs. By Alg.~\ref{alg:RecEqClAss}, for
each element $G$ retrieved from $Q$, subset relations with all
CS-equivalence classes of elementary circuits intersecting in at least one
$e\in E_1(G)$ are determined. These subset relations, however, define
predecessor relations in $\Hasse(\mathcal{E})$ and can be utilized to avoid
unnecessary testing for autocatalysis as suggested. By construction, only
CS-equivalence classes of elementary circuits can be leaves. If their
set of predecessors is empty, they can be excluded from testing whenever
their associated CS matrix is non-Metzler. In the Metzler case, all
predecessors are recursively screened for the autocatalytic capacity
of their associated Metzler matrix.  Whenever this is the case, the search
is stopped and testing can be omitted. Additional flags avoid visiting and
testing an element twice.

\subsection{Direct enumeration of autocatalytic cores}

Since autocatalytic cores are necessarily Metzler matrices, it is possible
to modify Alg.~\ref{alg:RecEqClAss} to enumerate autocatalytic cores only:
first, the enumeration of elementary circuits is restricted to Metzler
circuits (Def.~\ref{def:Metzlercircuit}) since by
Cor.~\ref{cor:Metzler-core} every induced fluffle, and thus every candidate
for an autocatalytic core is a union of Metzler circuits. Moreover, if
$C_1,C_2,\dots,C_{n}$ is a circuitnet for a Metzler fluffle $G$ and each
$C_i$ is a Metzler fluffle, then any fluffle $G_k=\bigcup_{i=1}^k
C_{\pi(i)}$ (Cor.~\ref{cor:superimposingfluffle}) leads to a corresponding
Metzler matrix $\SM[E_1(G_k)]$ since $G_k$ is a fluffle subgraph of an
induced fluffle. Thereby, $\SM[E_1(G_k)]$ is a principal submatrix of a
Metzler matrix and hence itself a Metzler matrix. Therefore, all
non-Metzler fluffles can be discarded immediately. Moreover, if $E_1(G)$ is
an autocatalytic core, none of its extensions can be cores. Hence, $E_1(G)$
is only pushed on the queue $Q$ if it corresponds to an induced fluffle and
$\SM[E_1(G)]$ is not autocatalytic, while only the induced autocatalytic
cases are added to the output $\mathcal{E}$. This procedure, however, is
not guaranteed to detect all predecessor relationships between
autocatalytic matrices. The resulting false positive candidates can be
identified in a post-processing step by checking whether there are subset
relationships among the core candidates in $\mathcal{E}$. This inclusion
testing can be parallelized to increase efficiency. However, empirical
tests revealed that a different strategy performs better: pushing all
elementary circuits and larger fluffles, along with their associated
autocatalytic Metzler CS matrices, into the queue while restricting
their processing to inclusion-relation detection only. This approach
drastically reduces the number of required set-inclusion tests between
candidates and therefore offers a substantial performance advantage.

\subsection{Extensions for large CRNs} 

With increasing network size, the number of expected elementary circuits
grows exponentially. Exhaustive enumeration of elementary circuits as
required by Alg.~\ref{alg:RecEqClAss} therefore becomes infeasible for
large CRNs. A natural restriction is to limit the size of circuits, at the
cost of also limiting the size of resulting fluffles in the assembly. In
biochemical networks, one may expect that autocatalytic subsystems are
predominantly confined to functional modules or pathways.  We therefore aim
to enumerate circuits first within such modules and only then extend the
search for circuits to connections between modules.  We proceed in two
steps: the network is clustered and elementary circuits within units are
enumerated exhaustively. Then circuits crossing (typically sparsely
connected) borders of neighboring clusters (as identified by the cluster
partition tree) are enumerated with size restrictions.

A useful decomposition of a CRN should ideally generate subnetworks of
roughly equal size while preserving the cycle structure within modules as
much as possible. Moreover, as mentioned, the modules should be
biochemically meaningful, i.e., encapsulate specific metabolic
functions. This problem has received considerable attention in applications
to metabolic networks \cite{holme_subnetwork_2003, schaeffer_graph_2007,
  sridharan_identification_2011}. Here, we re-implemented the partitioning
algorithm described previously \cite{sridharan_identification_2011}, which is based
on spectral methods \cite{newman_modularity_2006}. A detailed description
and pseudocode for cluster and cycle enumeration algorithms are provided in
SI Sec.~\emph{\nameref{sec:Algorithms}}.

\subsection{Implementation details}

The algorithms detailed above are implemented as a Python package
\ourtool. It is structured in different submodules. First, a metabolic
model is imported in \code{xml}-format in \code{partitionNetwork.py} via
\code{libsbml} \cite{bornstein_libsbml_2008} and translated into a
\code{networkX} \cite{SciPyProceedings_11} \code{DiGraph}. After
modularisation by means of leading eigenvector computations using
\code{NumPy} \cite{harris_array_2020} in \code{partitionComputations.py},
the partition tree and all relevant parameters are \emph{pickled} and saved
separately for each strongly connected component. In the second step, each
strongly connected component is now analyzed separately in the module
\code{partitionAnalysis.py}. The submatrix of the stoichiometric matrix is
extracted, and then the set $\mathcal{E}$ of elementary circuits
enumerated. Depending on the strategy chosen by the user, the associated
CS matrices and their autocatalytic capacities are determined concurrently
or downstream after assembly. During assembly, feasible combinations of
CS-equivalence and CS-equivalence classes of elementary circuits are
combined. Finally, autocatalytic capacity is computed using the real part
of the largest eigenvalue or by solving a linear programming problem
with \code{SciPy} \cite{virtanen_scipy_2020}. If feasible, partitioning,
enumeration of elementary circuits, and assembly of larger equivalence
classes is processed in parallel using \code{ConcurrentFutures}
\code{ProcessPoolExecutors}. Heavy computations are processed with Cython
\cite{behnel2011cython}.

\section{Showcase applications}
\label{sect:appl}

\begin{figure}
  \centering
  \includegraphics[width=\columnwidth]{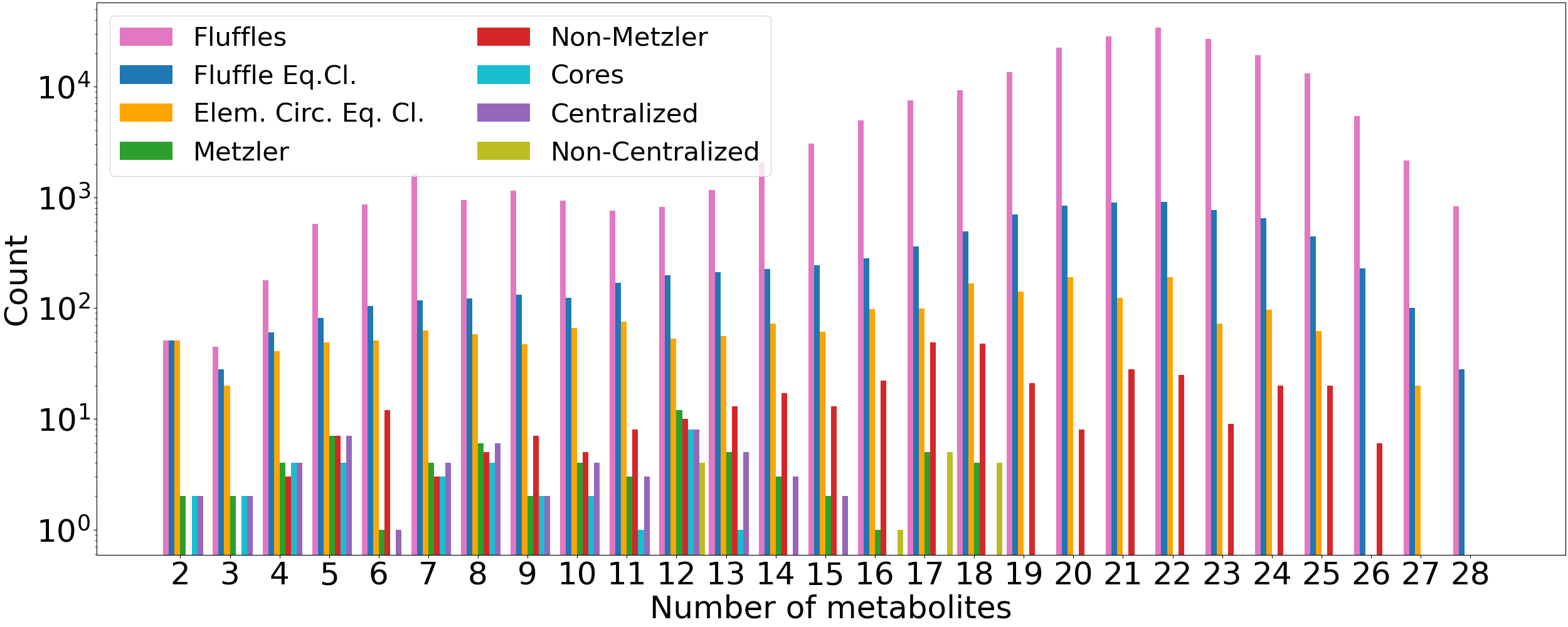}
  \caption{Size distribution of CS-equivalence classes for elementary
      circuits and fluffles, and autocatalytic Metzler and non-Metzler
    matrices in the \textit{E.\ coli} core metabolism. See SI
    Fig.~\ref{fig:EqClFluffles} for additional details.}
  \label{fig:EqClFluffles}
\end{figure}

To demonstrate the performance of \ourtool, and thus the practical use of
the algorithms described above, we investigated four metabolic networks:
the \textit{E.\ coli} core model \cite{orth_reconstruction_2010}, a larger
model of \textit{E.\ coli} DH5$\alpha$ \cite{monk_multiomics_2016}, a
metabolic model of human erythrocytes \cite{bordbar_iabrbc283_2011}, and a
model of the archaeon \textit{Methanosarcina barkeri}
\cite{feist_modeling_2006}. In all cases, we removed small, highly connected
molecules (e.g., CO\textsubscript{2} and H\textsubscript{2}O) as well as
exchange metabolites such as ADP and NADH, since they are of little
relevance for the biological interpretation of autocatalysis. A full list
for each model is provided in SI Sec.~ref{sec:addComp}.

Our version of the \textit{E.\ coli} core metabolism CRN comprises 36
metabolites and 71 reactions. Its K{\"o}nig graph contains 2,021 elementary
circuits, all belonging to distinct CS-equivalence classes. The enumeration
algorithm identified 202,206 fluffles, grouped into 8,551 CS-equivalence
classes. For a summary of their size distribution, we refer to
Fig.~\ref{fig:EqClFluffles}. Even for relatively small networks,
restricting to representatives of CS-equivalence classes provides a drastic
reduction in computational resources: fluffle enumeration took about 12,500
seconds, while enumeration of CS-equivalence classes required only about 25
seconds and only 5 seconds when being computed in parallel. In total,
\ourtool required 17.4s for completion, including decomposition of the
network and construction of the stoichiometric matrix from the reaction
data.

Autocatalytic subsystems are common in metabolic networks. In the central
carbon metabolism, represented by our \textit{E.\ coli} core model, 158 of
the 2,021 elementary circuits (7.9\%) are autocatalytic. Of these, 42 are
associated with Metzler CS matrices and 114 with non-Metzler CS
matrices. Overall, approximately 5\% (426 of 8,551) of CS-equivalence
classes are autocatalytic; 67 of these have a Metzler $\SM[\child]$, while
359 are non-Metzler. Of the 67 autocatalytic Metzler matrices, 53 were
centralized and only 14 were non-centralized (see 
Def.~\ref{def:centautocata} below for more details). Interestingly, the
ratio of autocatalytic Metzler CS to non-Metzler matrices roughly
halves when moving from elementary circuits to all equivalence classes,
from 1/3 to 1/6, which corresponds to the overall decrease in the fraction
of equivalence classes associated with Metzler matrices, from 6.4\% to
1.8\%. Among the 67 autocatalytic CS-Metzler matrices, we identified 33
autocatalytic cores. One of these, shown in SI Fig.~\ref{fig:TypeIV}, is
an autocatalytic core of Type IV. Previously, no example of this type had
been reported in the literature
\cite{blokhuis_universal_2020,unterberger_stoechiometric_2022}.

We compared our implementation with the ILP formulation of
Gagrani et al.~\cite{gagrani_polyhedral_2024}. To this end, we computed the
stoichiometric matrix and passed these data as input to the ILP, which
found 31 autocatalytic cores in 1.69s. Restricting \ourtool to enumerating
autocatalytic cores exclusively, only 0.342s (averaged over 1000
iterations) were required, and 33 cores were identified, which
included all found by Gagrani et al. \cite{gagrani_polyhedral_2024}. The two additional
cores are depicted in Fig.~\ref{fig:AddCores} of the Supplemental
  Material. We comment on potential reasons for these differences in the
Appendix below.

\begin{figure}
  \centering
  \includegraphics[width=\columnwidth]{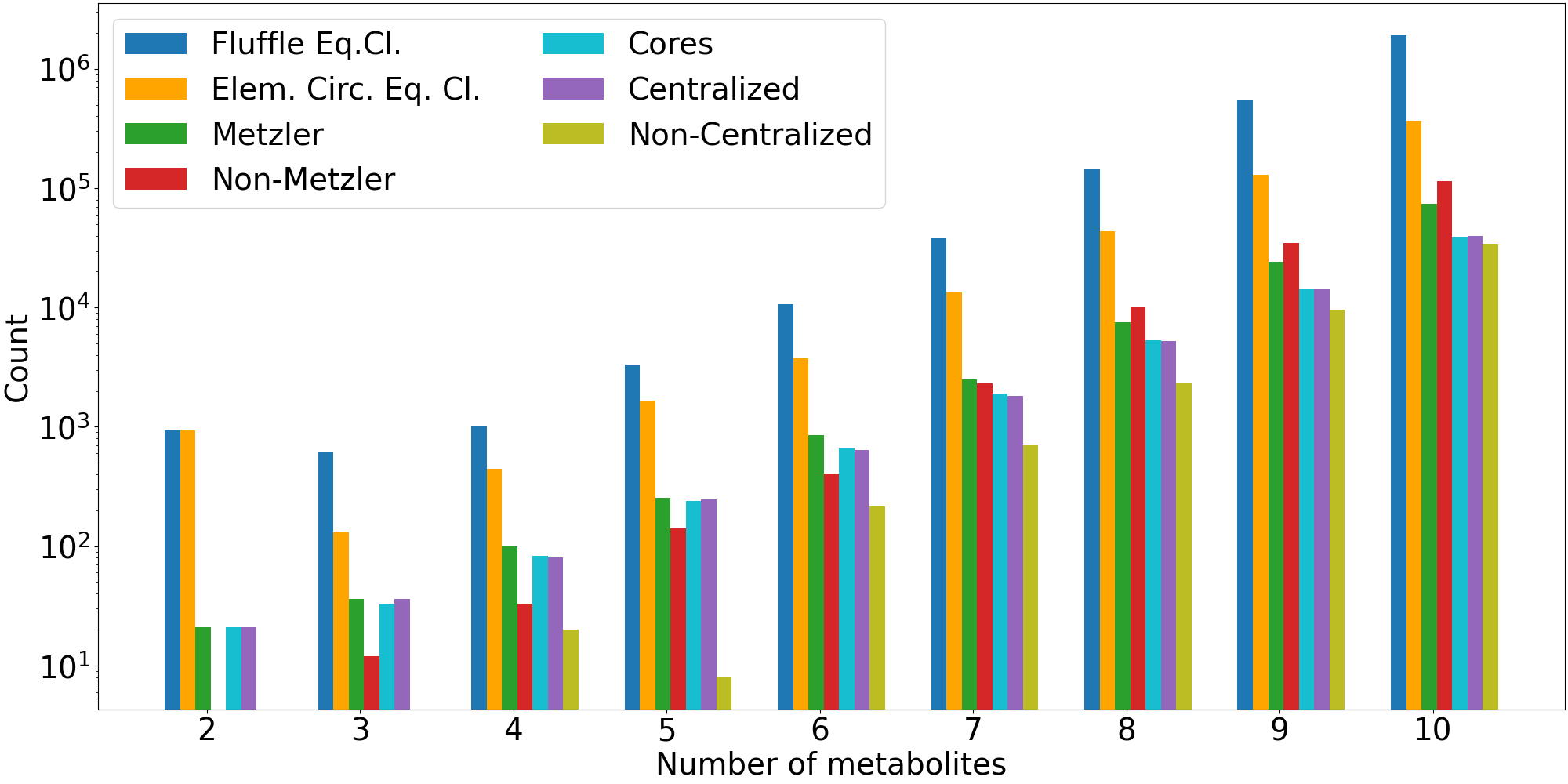}
  \caption{Size distribution of CS-equivalence classes for elementary
    circuits and fluffles, as well as autocatalytic CS-Metzler and
    non-Metzler matrices in the largest connected component of our
    modified \textit{E.\ coli DH5$\alpha$} network.}
  \label{fig:EqEColiC}
\end{figure} 

The larger \textit{E.\ coli} DH5$\alpha$ network comprised $2779$ reactions
and $1951$ metabolites. After removing small and highly connected
metabolites (see SI Sec.~\emph{\nameref{sec:addComp}}), as performed for the
\textit{E.\ coli} core network, $10$ strongly connected components with at
least $2$ reactions remained. In total, we retained $1142$ reactions and
$622$ metabolites. The largest strongly connected component comprised of
$568$ metabolites and $1061$ reactions. Overall, $2,647,664$ CS-equivalence
classes with at most $10$ metabolites and $10$ reactions could be detected;
the majority $(94.8\%; 2,516,295)$ comprised CS-equivalence classes with
associated non-Metzler matrices; $161,589 \ (6.4\%)$ of them autocatalytic
and $2,354,706 \ (93.6\%)$ non-autocatalytic. In contrast, $131,369
\ (5.2\%)$ of the enumerated CS-equivalence classes were associated with a
Metzler matrix; $109,391 \ (83.3\%)$ autocatalytic and only $21,978
\ (16.7\%)$ non-autocatalytic matrices. The majority of autocatalytic
Metzler matrices ($56\%$; $61,903$) form autocatalytic cores. Centralized
autocatalysis dominated with $57\%$ non-centralized autocatalysis with
$43\%$ slightly among the autocatalytic Metzler matrices. Overall, $10\%$
of all CS-equivalence classes were autocatalytic. The size distributions
are depicted in Fig.~\ref{fig:EqEColiC}. In total, \ourtool required
1h:49min:24s with maximum consumption of 16.9 Gb internal memory. 
To compare with the ILP of \cite{gagrani_polyhedral_2024}, we used the same
approach as for the smaller network and supplied the stoichiometric matrix
of the largest strongly connected component as input. After 14 hours of
running time, 4700 autocatalytic cores had been enumerated up to a size of
8 species and reactions. At this point, 20s were required for the
computation of the next core. We, therefore, terminated the enumeration
process. In contrast, the restriction of our algorithm to enumerating only
cores took 4min:5s with a maximum memory consumption of 302 Mb. The
majority of the time was spent on the enumeration of elementary circuits,
while assembly of larger equivalence classes and post-processing finished
in 26s and 17s, respectively.

\begin{figure}
  \centering
  \includegraphics[width=\columnwidth]{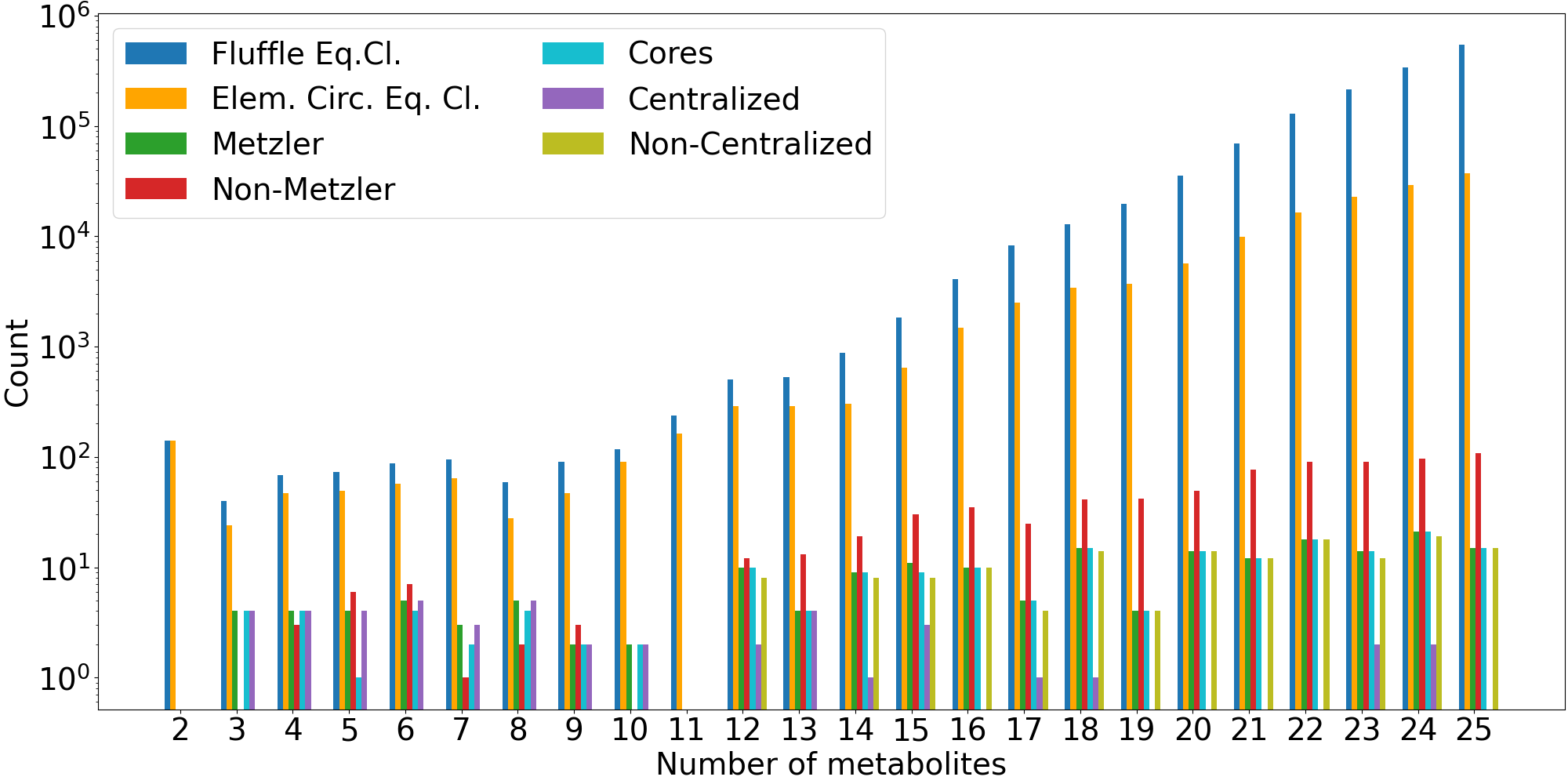}
  \caption{Size distribution of CS-equivalence classes for elementary
    circuits and fluffles, as well as autocatalytic CS Metzler and
    non-Metzler matrices in the largest connected component of our
    modified \textit{erythrocyte} network.}
  \label{fig:Erythrocytes}
\end{figure} 

To investigate the frequency of autocatalysis in non-bacterial species, we
applied the algorithm to another network, of human erythrocytes
\cite{bordbar_iabrbc283_2011}, which contained 342 metabolites and 469
reactions. After removal of all small metabolites, two larger strongly
connected components remained: one composed of 69 metabolites and 112
reactions, respectively, covering central carbon metabolism, including
glycolysis, PPP, and amino-acid metabolism, and a second component composed
of 72 metabolites and 135 reactions, largely covering lipid
metabolism. In summary, 1,379,913 CS-equivalence classes with a maximum
size of 25 metabolites/reactions were enumerated. In contrast to the
\textit{E.\ coli} networks, only $940$ ($0.068\%$) were autocatalytic. The
network reflecting central carbon metabolism exhibited approximately $8\%$
(103/1258) autocatalytic CS-equivalence classes (35/150 Metzler and 113/183
non-Metzler), which is in line with the results obtained from the
\textit{E.\ coli core} network. However, the network reflecting lipid
metabolism contained 1,378,647 CS-equivalence classes, of which only
$0.06\%$ (837) were autocatalytic; of these 156 are Metzler and 681
non-Metzler matrices. Size distributions for both networks together are
depicted in Fig.~\ref{fig:Erythrocytes}.

\begin{figure}
  \centering
  \includegraphics[width=\columnwidth]{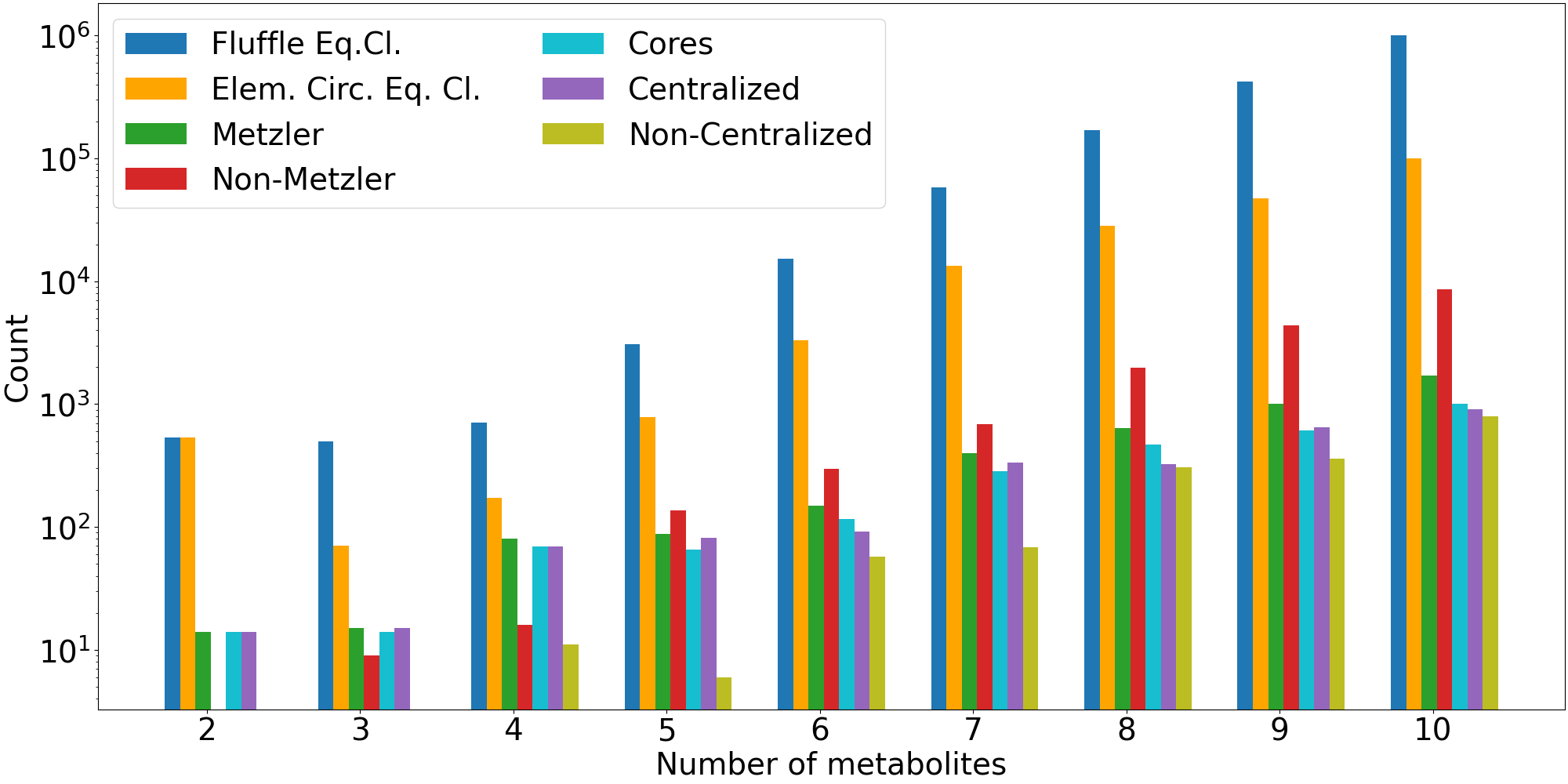}
  \caption{Size distribution of CS-equivalence classes for elementary
    circuits and fluffles, as well as autocatalytic CS-Metzler and
    non-Metzler matrices in the largest connected component of our
    modified \textit{Mathanosarcina barkeri} network.}
  \label{fig:Methanosarcina}
\end{figure} 

Finally, we investigated whether autocatalysis could be found in a
member of the Archaea domain. To this end, we took advantage of the
metabolic model of \textit{Methanosarcina barkeri}
\cite{feist_modeling_2006} with 690 metabolites and 692 originally, of
which 249 and 402, respectively, remained in the largest strongly connected
component. We restricted the size of all CS-equivalence classes to 10
metabolites and reactions. Within this connected component, only 1.2\%
(20,194) of all CS-equivalence classes (1,677,604) were found to be
autocatalytic; 4105 with a Metzler and 16,089 with a non-Metzler
matrix. Nearly three-quarters of the Metzler matrices (5483) were
autocatalytic, while for the non-Metzler matrices this is the case for only
1\% (16,089/1,656,158). Two-thirds of the CS-equivalence classes with
autocatalytic Metzler matrices correspond to autocatalytic cores. The size
distribution is depicted in Fig.~\ref{fig:Methanosarcina}.

\FloatBarrier

\section{Discussion}
\label{sect:conclusion}

We have presented a detailed mathematical analysis of autocatalytic
substructures in large CRNs. Starting from the stoichiometric matrix $\SM$,
we identify a specific class of subgraphs in the bipartite (K{\"o}nig)
representation of the CRN, called fluffles, which are necessary to support
irreducible autocatalytic subnetworks. Fluffles fall into equivalence
classes determined solely by the corresponding child-selections. These
correspond to Metzler matrices that form autocatalytic cores if and only if
they are induced subgraphs of the CRN, while larger irreducible
autocatalytic subnetworks only need to contain a Metzler part, or
equivalently, a spanning fluffle, as well as a smaller autocatalytic core.

Based on these structural insights, we developed an algorithmic approach to
produce representative fluffles by superimposing elementary circuits. This
purely graph-theoretical method avoids the complex ILP formulation used 
previously to detect autocatalytic cores \cite{gagrani_polyhedral_2024, 
kosc_thermodynamic_2025}. Furthermore, it extends to a much broader 
class of autocatalytic subsystems beyond the autocatalytic cores. 
Tests on four metabolic networks, a small model of the
\textit{E.\ coli} core metabolism and three much larger CRNs comprising up
to more than 600 metabolites and 1100 reactions showed that our algorithmic
approach is feasible in practise. For the small network, a complete
analysis is obtained within about 17 seconds. For the large network, the
computation had to be limited to moderate-size fluffles, with up to 10, 10,
and 25 metabolites and reactions for \textit{E.\ coli}, Methanosarcina Barkeri, 
and human erythrocytes, respectively. Clearly, this does not capture
all autocatalytic cores, since cores in the smaller \textit{E.\ coli}
network ranged up to 13 metabolites, i.e., almost half of the size of the
CRN. Nevertheless, in the \textit{E.\ coli DH5$\alpha$} model, we
identified more than 100,000 irreducible autocatalytic CS subnetworks, more
than half of which are autocatalytic cores. These results reinforce the
conclusion of earlier studies (in particular those based on different
definitions of autocatalysis \cite{barenholz_design_2017}) on the
ubiquitous nature of autocatalysis in metabolic CRNs.

A direct comparison of the restricted variant of Alg.~\ref{alg:RecEqClAss}
that computes autocatalytic cores only with the ILP formulation 
\cite{gagrani_polyhedral_2024} turned out favorably for our approach with
respect to resource consumption. An evaluation of the \textit{E.\ coli}
core network shows furthermore that the ILP does not enumerate all
autocatalytic cores and struggles with larger network sizes. In fact, the
algorithm described previously \cite{gagrani_polyhedral_2024} focuses on
enumerating only the minimal subsets of reactions that contain an
autocatalytic core, without imposing any restriction on the species
involved. Consequently, for reaction sets \(R_1 \subset R_2\) associated
with two autocatalytic cores \(A_1\) and \(A_2\) based on different sets of
species, the ILP formulation would identify only the minimal set
\(R_1\). We briefly elaborate on this issue in the Appendix.

Once autocatalytic subsystems have been identified, they can provide
further insight into the potential behaviors of the CRN. For example, the
close connection between autocatalysis and sustained oscillations has been
explored before \cite{blokhuis_stoichiometric_2025}. Building on this work, one
can state the following sufficient condition:
\begin{proposition}[Proof: SI]
  \label{prop:StabOszil}
  Let $\SM[\child]$ be a Hurwitz-stable autocatalytic CS matrix. Then there
  exists a choice of parameters such that the system
  \eqref{eq:ODEdynamics}, 
  $\dot{z}=f(z) \coloneqq \SM \cdot v(z)$, admits periodic solutions.
\end{proposition} 
This result sets the stage for identifying minimal subnetworks that
  are responsible for ``interesting'' dynamical behavior such as periodic oscillations.

  
For large CRNs, in particular models of complete metabolisms, an exhaustive
enumeration of fluffle CS-equivalence classes is probably infeasible even
on an HPC system. This is certainly true in the (chemically unrealistic)
worst-case scenario, since it is possible to construct CRNs in which all
autocatalytic cores have size $2$ but there are exponentially many
autocatalytic Metzler matrices: it suffices to consider a CS matrix
$\SM[\child]$ such that $\SM[\child]_{ii}=-1$, $\SM[\child]_{ij}=2$ for
$i>j$, and $\SM[\child]_{ij}=1$ for $i<j$; in this case, only the 2x2
principal submatrices of $\SM[\child]$ are autocatalytic cores while every
principal submatrix of $\SM[\child]$ is an irreducible, autocatalytic
Metzler matrix.

By Cor.~\ref{cor:Metzler-core}, all autocatalytic cores are superpositions
of Metzler circuits. Since worst-case instances may also contain very large
numbers of elementary circuits, this raises the question of whether it is
possible to enumerate the subset of Metzler circuits without enumerating
all elementary circuits. In our analysis of the large \textit{E.\ coli}
network, we pragmatically limited the length $L$ of the elementary
circuits. Current versions of Johnson's algorithm allow such a cut-off. In
particular, the algorithm of Gupta \& Suzumura \cite{gupta_finding_2021}
for sparse graphs, with running time $O((c+|V|)L\bar{d}^L)$ where $\bar{d}$
is the average degree, is attractive for applications to CRNs. So far,
there is no comparably efficient approach to produce elementary circuits
ordered by size. A related algorithm to enumerate chordless elementary
circuits, optionally restricted to length $L$, is described in
\cite{detecting_Bisdorff_2010}. We show that this is indeed feasible
  in a follow-up paper \cite{golnik_using_2026}.

All autocatalytic cores derive from induced fluffles, and more broadly from
(not necessarily maximal) induced strong blocks in $\king$. A recent
linear-delay algorithm for enumerating strongly connected induced subgraphs
\cite{tada_linear_2025, shota_linear_2025} may thus serve as an alternative
starting point for the efficient generation of candidate subsets for
autocatalytic cores.

As expected, the computational examples in the previous section 
identified a large number of irreducible autocatalytic subnetworks in
metabolic systems. It remains an open question what fraction of those is of
biological relevance. Clearly, this will depend on the metabolic fluxes
that are realized and on the fluxes that can potentially be realized
given specific food sets \cite{gagrani_polyhedral_2024,kosc_thermodynamic_2025}. 
The enumeration of irreducible autocatalytic subnetworks at least makes it 
possible to address such questions computationally in systematic studies.

\begin{small}
\section{Declarations}

\subsection{Availability of Data and Materials} 

The implementation, models, and all necessary data is available at
https://github.com/hollyritch/autogato.

\subsection{Competing interests}

The authors declare that they have no competing interests.

\subsection{Funding}
This work was supported in part by the Novo Nordisk Foundation (grant
no.\ 0066551, MATOMIC), the German Research Foundation (grant
no.\ 234823413), and the German Federal Ministry of Education and Research
(BMBF) within the German Network for Bioinformatics Infrastructure (de.NBI)
under grant number w-de.NBI{\textbackslash}303{\textbackslash}203018.
High-performance computing resources were provided by the German
Network for Bioinformatics Infrastructure (de.NBI).

\subsection{Authors' contributions}

All authors contributed to the conceptualization, methodology, and writing of
the manuscript.
\end{small}

\section*{Supporting information}

The file \emph{Supplements.pdf} is available free of charge and contains:

\begin{itemize}
  \item {\textbf{Background: Graphs and Matrices} (Basic Notation, Graph Theoretical Constructions)}
  \item {\textbf{Proofs of Statements in the Main Text}}
  \item {\textbf{The Set System of Circuitnets of Fluffles} ((strong) accessibility, commutability, no confluence)}
  \item {\textbf{Additional Computational Data} - Information on the metabolic models analyzed in this contribution for \emph{E. coli} (core metabolism), \emph{E. coli} (DH5$\alpha$), \emph{human erythrocytes}, and \emph{Methanosarcina barkeri}}
  \item{\textbf{Examples}
 	\begin{itemize}
		\item {Example for a type III autocatalytic core that does not admit a single elementary-circuit CS representative}
		\item {Example for an autocatalytic core type IV in the pentose-phosphate-pathway (PPP) of the \emph{E. coli} core model}
	\end{itemize}
	}
\item{\textbf{Algorithmic Overview} - Pseudocode and proof of correctness for the following algorithms: \textit{Partitioning}, \textit{Enumeration of elementary circuits}, \textit{Detecting autocatalytic capacity}}
\item {\textbf{Figures}
  	\begin{itemize}
		\item{Figure S1: Depiction of a type V autocatalytic core}
		\item{Figure S2: Counterexample that the set system of circuitnets of fluffles is confluent}
		\item{Figure S3: Length distribution of elementary and fluffles and their corresponding equivalence classes in the \emph{E. coli} core model}
		\item{Figure S4: Depiction of the two autocatalytic cores identified by \ourtool but not by the ILP approach of Gagrani et al. \cite{gagrani_polyhedral_2024} in the \emph{E. coli} core model}
		\item{Figure S5: Schematic depiction of an autocatalytic core type IV and an example from the PPP of the \emph{E. coli} core model}
		\item{Figure S6: Illustration for the fusion of two vertices of the partitioning tree with four intersecting metabolites}
	\end{itemize}
	}
\end{itemize}

\bibliography{AutocatalysisMath.bib}

\newpage

\section{Appendix}
\begin{appendix}

\section{Centralized autocatalysis}
\label{sec:centAut}

Upon analyzing autocatalytic structures one is bound to notice
a qualitative, dichotomic difference between two types. Consider the
following two examples:
\begin{center}
\begin{tabular}{lcl}
  \textbf{Example I:} & \qquad\qquad\qquad\qquad &  \textbf{Example II:}        \\
  $x_1\rightarrow x_2+x_3$ &  &  $x_1\rightarrow x_2+x_3$    \\
  $x_2\rightarrow x_1$     &  &  $x_2\rightarrow x_1 + x_3$  \\
  $x_3\rightarrow x_1$     &  &  $x_3\rightarrow x_1+x_2$    \\
\end{tabular}
\end{center}
with associated CS matrices, respectively,
\begin{equation}
  \begin{pmatrix}
    -1 & 1 & 1\\
    1 & -1 & 0\\
    1 & 0 & -1
  \end{pmatrix}
  \quad\quad\quad\quad
  \begin{pmatrix}
    -1 & 1 & 1\\
    1 & -1 & 1\\
    1 & 1 & -1
  \end{pmatrix}.
\end{equation}
In Example I, species $x_1$ produces $x_2$ and $x_3$, and both in turn
react back to $x_1$. The autocatalytic process is thus \emph{centralized}
around $x_1$: every reaction cycle passes through $x_1$, so that
autocatalysis can be interpreted as an amplification mechanism for $x_1$.
In contrast, in Example II, no single species plays such a role, due to a
stronger interconnection of the network: there is a clear symmetry among
indices $1,2,3$, and the network topology is invariant under permutations
of these labels.  To capture and formalize this intuitive difference, we
introduce the notion of \emph{centralized autocatalysis}. Our first
formulation is based on \emph{permutation cycles} of the associated
CS matrix. Subsequently, we will provide an equivalent graph-theoretic
characterization.  Recall the standard Leibniz formula for the determinant
of a $k\times k$ CS matrix $\SM[\child]$:
\begin{equation}\label{eq:leibnizCS}
  \det \SM[\child] = \sum_{\pi \in P_k} \operatorname{sgn}(\pi)
  \prod_{m=1}^{k} \SM[\child]_{m, \pi(m)},
\end{equation}
where $P_k$ denotes the symmetric group on $k$ elements. Every non-identity
permutation $\pi\neq\mathrm{id}$ can be decomposed as a product of disjoint
cyclic permutations of at least two elements,
\begin{equation}
  \pi= C_1 \cdot ... \cdot C_{n_\pi}.
\end{equation}
We write $\mathcal{C}_\pi$ for the set of permutation cycles of $\pi$.  For
simplicity of notation, we consider each cycle $C_i$ as an element itself
of $P_k$, i.e., as a permutation of its $k=|X_\kappa|$ elements. We say
that a permutation cycle $C$ \emph{contributes} (to $\det\SM[\child]$) if
\begin{equation}
  \prod_{m=1}^{k} \SM[\child]_{m, C(m)}\neq 0.
\end{equation}
We denote with $\mathfrak{K}[\child]$ the set of non-trivial permutation
cycles contributing to $\det\SM[\child]$, i.e. with length $\ge2$.  In
turn, permutation cycles without a contribution in the Leibniz formula
\eqref{eq:leibnizCS} can be ignored in the following. We are now in the
position to formalize the distinction between the two examples above:
\begin{definition}
  \label{def:centautocata}
  Let $\SM[\child]$ be a $k\times k$ irreducible autocatalytic
  Metzler CS matrix and denote by $M_{\kappa}$ the set of all
  $m^*\in\{1,\dots,k\}$ such that every permutation cycle $C$ contributing
  to $\det\SM[\child]$ satisfies $\SM[\child]_{m^*,C(m^*)}>0$. Then
  $\SM[\child]$ is \emph{centralized} if $M_{\kappa}\ne\emptyset$. The
  elements of $M_{\kappa}$ are the \emph{autocatalytic centers} of
  $\SM[\child]$ and we say that $\SM[\child]$ is \emph{centered at
  $M_{\kappa}$} provided $M_{\kappa}\ne\emptyset$.
\end{definition}

As a direct consequence of the definition, moreover, for centralized
autocatalysis, the determinant computed via the Leibniz formula
\eqref{eq:leibnizCS} naturally provides the sum over the weights of the
different cycles in the network, as the next proposition states.
\begin{proposition}[Proof: SI] 
  \label{prop:centralizedDet}
  Let $\SM[\child]$ be a $k\times k$ irreducible autocatalytic Metzler
  CS matrix that exhibits \emph{centralized autocatalysis}. Then
\begin{equation}
  \frac{\operatorname{det}\SM[\child]}{\prod_{m=1}^{k}\SM[\child]_{mm}}
  =1-\sum_{C}\prod_{m\in C}
  \frac{\SM[\child]_{m, C(m)}}{\vert \SM[\child]_{mm} \vert}
  \label{eq:propcentralized}
\end{equation}
where the sum runs on all permutation cycles.
\end{proposition}
In essence, Prop.~\ref{prop:centralizedDet} differs from the Leibniz
formula \eqref{eq:leibnizCS} because the sum runs over permutation cycles,
only, instead over all permutations. Prop.~\ref{prop:centralizedDet} holds
for centralized autocatalysis. Nevertheless, equality
\eqref{eq:propcentralized} does not solely apply to centralized
autocatalysis as stated in Thm.~\ref{thm:NgheClasses}.  Hence, it does not
provide a characterization of centralized autocatalysis.

Although examples of stoichiometric coefficients different from $(0,1)$ do
exist in metabolic networks, they are very rare. A well-known example
is the condensation of two acetyl-CoA molecules into acetoacetyl-CoA in the
synthesis of HMG-CoA during cholesterol or isopentenyl pyrophosphate (IPP)
biosynthesis \cite{lange_isoprenoid_2000}. In the special case of unit
stoichiometric coefficients, Prop.~\ref{prop:centralizedDet} simplifies to
an easily interpretable statement regarding the \emph{number} of
contributing permutation cycles of the stoichiometric matrix.
\begin{corollary}\label{cor:centralizedDettrivial}
Let $\SM[\child]$ be a $k\times k$ irreducible autocatalytic Metzler
CS matrix that exhibits \emph{centralized autocatalysis}, and such
that
\[\SM[\child]_{ij}\in\{-1,0,1\} \quad \text{for all $(i,j)$}.\]
Then
\[\operatorname{det}\SM[\child](-1)^k=1-\#_C,\]
where $\#_C$ is the number of permutation cycles $C$ such that
$\prod_{m=1}^{k} \SM[\child]_{m, C(m)}\neq 0$.
\end{corollary}
\begin{proof}
It directly follows from Prop.~\ref{prop:centralizedDet}.
\end{proof}

For a $k$-CS $\child=(X_\kappa, R_\kappa, \kappa)$, centralized
autocatalysis can be characterized in graph-theoretical terms using a
correspondence between contributing permutation cycles in its CS matrix
$\SM[\child]$ and directed elementary circuits in the induced subgraphs
$\king[\child]\coloneqq \king[X_\kappa \cup R_\kappa]$ of the K{\"o}nig
graph of the CRN. The key observation is that if $x$ and $y$ are
consecutive vertices in a permutation cycle $C$ that contributes to
$\SM[\child]$, then $(x,\kappa(x),y)$ is a path in $\king[\child]$ and,
\emph{vice versa}, if $(x,r,y)$ is a path in $\king[\child]$ of an
irreducible autocatalytic child-selection, then $r=\kappa(x)$.  Denoting
the elementary circuits (viewed as subgraphs of $\king[\child]$) by
$\mathbf{C}(\child)$, we obtain the following formal statement:
\begin{lemma}[Proof: SI] 
  \label{lem:ElemPermCycles}
  Let $\SM[\child]$ be an autocatalytic Metzler CS matrix and
  $\mathfrak{K}[\child]$ the set of contributing permutation cycles of
  length $\geq 2$. Then there is a one-to-one correspondence between
  $\mathfrak{K}[\child]$ and $\mathbf{C}(\child)$ such that a contributing
  permutation cycle $(x_1,x_2,\dots x_k)$ corresponds to the elementary
  circuit \linebreak
  $(x_1,\kappa(x_1),x_2,\kappa(x_2),\dots,\kappa(x_{k-1}),x_k,\kappa(x_k),x_1)$
  in $\king[\child]$.
\end{lemma}

As a direct consequence of Lemma \ref{lem:ElemPermCycles}, we can now
rephrase Def.~\ref{def:centautocata} in graph-theoretical terms:
\begin{corollary}
  \label{def:CenAutTriv}
  Let $\Auto\coloneqq \SM[\child]$ be an autocatalytic Metzler CS
  matrix. Then $\Auto$ is \emph{centralized} if and only if there is a
  vertex $x^*\in X_\kappa$ such that $x^*\in X(C)$ for all $C\in
  \mathbf{C}(\child)$.
\end{corollary}
We collect all center species $x^*$ in a set $X^*_{\kappa}$ and refer to it
as the autocatalytic center of $\Auto$. Note that $x^*$ and $X^*_{\kappa}$
correspond to $m^*$ and $M_{\kappa}$ above.

To simplify the notation, we introduce a normalized version of the matrix
$\Auto$ by setting $N(\Auto)_{ij}\coloneqq \Auto_{ij}/|\Auto_{jj}|$ for all
$i,j$. With this notation, Prop.~\ref{prop:centralizedDet} and
Lemma~\ref{lem:ElemPermCycles} immediately imply
\begin{corollary}\label{cor:CentAutNonTriv}
  Let $\Auto\in \{\mathbb{Z}\}^{k\times k}$ be centralized in
  $X^*_\kappa$, denote by  $\mathbf{C}_{x^*}(\child)$ the set of cycles
  containing $x^*$, and let $s_C(y)$ be the successor of $y$ along
  the elementary circuit $C$. Then 
  \begin{equation}\label{eq:CondNonTriv}
    \det(\Auto) \cdot (-1)^{k-1} = \sum_{C\in \mathbf{C}_{x^*}(\child)}
    \prod_{y\in X(C)} \left \vert N(\Auto)_{y,s_C(y)}\right \vert
  \end{equation}
  for all autocatalytic centers $x^*\in X^*_{\kappa}$. In particular if
  $\Auto\in \{-1,0,1\}^{k\times k}$ then
  \begin{equation}
    \det(\Auto) \cdot (-1)^{k-1} = {\vert \mathbf{C}_{x^*}(\child) \vert}
    -1,
  \end{equation} 
\end{corollary}

We conclude this section by connecting the concept of centralized
autocatalysis with the classification of autocatalytic cores proposed by
Blokhuis et al. \cite{blokhuis_universal_2020}. Their five types, in essence, 
correspond to the following five motifs (up to different stoichiometric
coefficients):
\begin{equation}
\textbf{Type I:} \quad\quad x_1\rightarrow x_2 \rightarrow 2x_1 \quad\quad \begin{pmatrix}
  -1 & 2\\
  1 & -1
\end{pmatrix}\qquad
\end{equation}
\begin{equation}
  \textbf{Type II$^*$}: \quad 
  \begin{cases} 
    x_1\rightarrow x_2+x_3 \\
    x_2\rightarrow x_3\\
    x_3\rightarrow x_1
  \end{cases}
  \quad \begin{pmatrix}
    -1 & 0 & 1\\
    1 & -1 & 0\\
    1 & 1 & -1
  \end{pmatrix}
\end{equation}
\begin{equation}
  \textbf{Type III:} \quad 
  \begin{cases} 
    x_1\rightarrow x_2+x_3 \\
    x_2\rightarrow x_1\\
    x_3\rightarrow x_1
  \end{cases}
  \quad \begin{pmatrix}
    -1 & 1 & 1\\
    1 & -1 & 0\\
    1 & 0 & -1
  \end{pmatrix}
\end{equation}
\begin{equation}
  \textbf{Type IV:} \quad 
  \begin{cases} 
    x_1\rightarrow x_2+x_3 \\
    x_2\rightarrow x_1 + x_3\\
    x_3\rightarrow x_1
  \end{cases}
  \quad \begin{pmatrix}
    -1 & 1 & 1\\
    1 & -1 & 0\\
    1 & 1 & -1
  \end{pmatrix}
\end{equation}
\begin{equation}\label{eq:coretype5}
  \textbf{Type V:} \quad 
  \begin{cases} 
    x_1\rightarrow x_2+x_3 \\
    x_2\rightarrow x_1 + x_3\\
    x_3\rightarrow x_1+x_2
  \end{cases}
  \quad \begin{pmatrix}
    -1 & 1 & 1\\
    1 & -1 & 1\\
    1 & 1 & -1
  \end{pmatrix}
\end{equation}
Examples 1 and 2 above are of Type III and Type V, respectively. 
For Type II, the classification of Blokhuis et al. \cite{blokhuis_universal_2020} 
allows for a more general structure with additional internal cycles. In 
what follows, we consider only the simplest Type II$^*$ motif above 
and refer to the original paper for a description of the general case. 
The next result states which of such motifs are centralized and 
establishes the range of validity of Eq.~\eqref{eq:propcentralized}a.
\begin{theorem}[Proof: SI] 
  \label{thm:NgheClasses}
  An autocatalytic core of type I, {{II}$^*$},
  III, or IV is centralized. Moreover, Eq.~\eqref{eq:propcentralized}
  holds for all five types of cores.
\end{theorem}
We underline that Type II, in its more general form, is typically 
\emph{not} centralized due to the presence of additional internal cycles. 
Similarly, Eq.~\eqref{eq:propcentralized} does not need to hold 
in this case. For brevity, we omit the straightforward demonstration 
of this observation, which lies outside the scope of this work.

\paragraph{Algorithmic considerations.} 
To test whether an autocatalytic Metzler matrix $\SM[\child]$ on
$(X_{\kappa},R_{\kappa})$ exhibits centralized autocatalysis, counting of
elementary circuits passing each $x\in X_{\kappa}$ is required. Since
$\kappa:X_{\kappa}\to R_{\kappa}$ is bijective, there is a 1-1
correspondence between circuits in the induced subgraph $\king[\child]$ and
the graph with vertex set $X_{\kappa}$ and edges $(x,y)$ whenever
$\SM[\kappa]_{yx}>0$. It therefore suffices to enumerate the set
$\mathcal{C}$ of elementary circuits in $\king[\child]$, e.g., using
Johnson's algorithm and to store the vertices $V(C)\cap X_{\kappa}$ in a
bit vector $\zeta_C$ for each $C\in\mathcal{C}$. The component-wise
conjunction of these vectors
\begin{equation}
  \zeta^* \coloneqq \bigwedge_{C\in\mathbf{C}} \zeta_C 
\end{equation}
identifies the set of autocatalytic centers as $M=\{x|\zeta^*_x=1\}$.
The autocatalytic Metzler matrix $\SM[\child]$ is therefore centralized
if and only if there is an $x\in X_{\kappa}: \zeta_x^*\ne 0$.

\section{Minimal Autocatalytic Systems (MAS)}
  \label{sec:MAS}

  The ILP approach of \cite{gagrani_polyhedral_2024} is based on
  subnetworks $(X',R')$ of a CRN $(X,R)$ where $R'\subseteq R$ and
  $X':=X(R')$ is defined as the set of species participating in the set of
  reactions $R'$ either as reactant or as product (or both). From the
  perspective of the present paper, any child-selection (CS)
  $\child=(X_\kappa,R_\kappa,\kappa)$ uniquely defines a subnetwork
  $(X(R_{\kappa}),R_{\kappa})$ as above, where typically $X_\kappa\subset
  X(R')$. The converse is, however, not true: the same subnetwork
  $(X(R'),R')$ may support another child-selection $\tilde{\child}$,
  i.e.\ with $R_{\tilde{\kappa}}=R_{\kappa}$. Thus, different 
  child-selections may be supported by the same subnetwork
  $(X(R'),R')$. Moreover, $\tilde{\child}$ need not be defined on the same
  set of species, i.e.\ we may have $X_{\tilde{\kappa}}\ne X_{\kappa}$ as
  long as $X_{\tilde{\kappa}}\subseteq
  X(R_{\tilde{\kappa}})=X(R_{\kappa})$. The key concept in the work of
  Gagrani \emph{et al.} \cite{gagrani_polyhedral_2024} is introduced by Def.~3.2 
  (III.2 in the arXiv preprint) as follows:
  \emph{``A minimal autocatalytic subnetwork (MAS) is defined to be the
  subnetwork with the least number of reactions containing a particular
  autocatalytic core.''} We rephrase this statement here as follows:
  \begin{definition}[Def.~3.2 \cite{gagrani_polyhedral_2024}]
    \label{def:prafulMAS}
    A \emph{minimal autocatalytic subnetwork} (MAS) is a subnetwork
    $(X(R'),R')$ containing an autocatalytic core that is induced by  
    an inclusion-minimal set of reactions with this property. 
  \end{definition}
  Inclusion minimality of the reaction set $R'$ implies that an MAS does not
  contain more reactions than specified by the child-selection of its
  defining autocatalytic core, i.e.\ $R'=R_{\kappa}$ for a child-selection
  $\child$ for which $\SM[\child]$ is an autocatalytic core. Since
  Def.~\ref{def:prafulMAS} imposes minimality only on the set of reactions,
  not all autocatalytic cores are associated with an MAS. In particular, an
  autocatalytic core $\SM[\child_2]$ with reaction set $R_{\kappa_2}$ is
  not associated with an MAS if $R_{\kappa_2}$ strictly contains a reaction
  set $R_{\kappa_1} \subset R_{\kappa_2}$ that supports another
  autocatalytic core $\SM[\child_1]$. By minimality of cores,
  $\SM[\child_1]$ necessarily involves a set $X_{\kappa_1}$ of reactants
  that is not contained in $X_{\kappa_2}=\kappa_2^{-1}(R_{\kappa_2})$. To
  see that such cases indeed exist, consider the following simple reaction
  network:
  \begin{equation}
    \begin{cases}
      x_1 + x_3 &\underset{1}{\longrightarrow}\quad x_2 + x_4\\
      x_2 &\underset{2}{\longrightarrow}\quad 2x_1 + x_3\\
      x_4 &\underset{3}{\longrightarrow}\quad x_3,
    \end{cases}
  \end{equation}
  with stoichiometric matrix
  \begin{equation}
    \SM=\begin{pmatrix}
    -1 & 2 & 0\\
    1 & -1 & 0\\
    -1 & 1 & 1\\
    1 & 0 & -1
    \end{pmatrix}.
  \end{equation}
  The first two columns and first two rows, corresponding to the CS
  \(\child_1=(\{x_1,x_2\},\{1,2\},\kappa(x_1,x_2)=(1,2))\), form the 
  autocatalytic core 
  \begin{equation}
    \SM[\pmb{\kappa}_1]=
    \begin{pmatrix}
      -1 & 2\\
      1 & -1
    \end{pmatrix}.
  \end{equation}
  The MAS-oriented ILP implementation of Gagrani \emph{et al.}
  \cite{gagrani_polyhedral_2024} therefore discards any reaction set
  strictly containing $\{1,2\}$. However, by considering all three columns
  and rows $2$, $3$, and $4$, we obtain another autocatalytic core:
  \begin{equation}
    \SM[\child_2]=
    \begin{pmatrix}
      -1 & 1 & 0\\
      1 & -1 & 1\\
      0 & 1 & -1
    \end{pmatrix},
  \end{equation}
  associated with the CS
  \begin{equation*}
    \child_2=(\{x_2,x_3,x_4\},\{1,2,3\},\kappa(x_2,x_3,x_4)=(2,1,3)),
  \end{equation*}
  where $X_{\kappa_1}=\{x_1,x_2\}\not\subseteq \{x_2,x_3,x_4\}$. Indeed,
  the ILP of Gagrani \emph{et al.} does not detect the autocatalytic core
  $\SM[\child_2]$. The two missing autocatalytic cores in the
  \textit{E.\ coli} core metabolism can also be explained in this manner.

  We note, finally, that it is a simple task to determine the set of all
  MAS in a given CRN if all autocatalytic cores have been computed. By
  considering the sets of reactions associated to all autocatalytic cores,
  it suffices to remove the sets that are non-minimal with respect to the
  inclusion relation. The graph-theoretic approach thus can also be used to
  enumerate MAS in addition to autocatalytic cores.

\end{appendix}

\clearpage
\section{TOC Graphic}

\begin{figure}[htb]
	\centering
	\includegraphics[width=\textwidth]{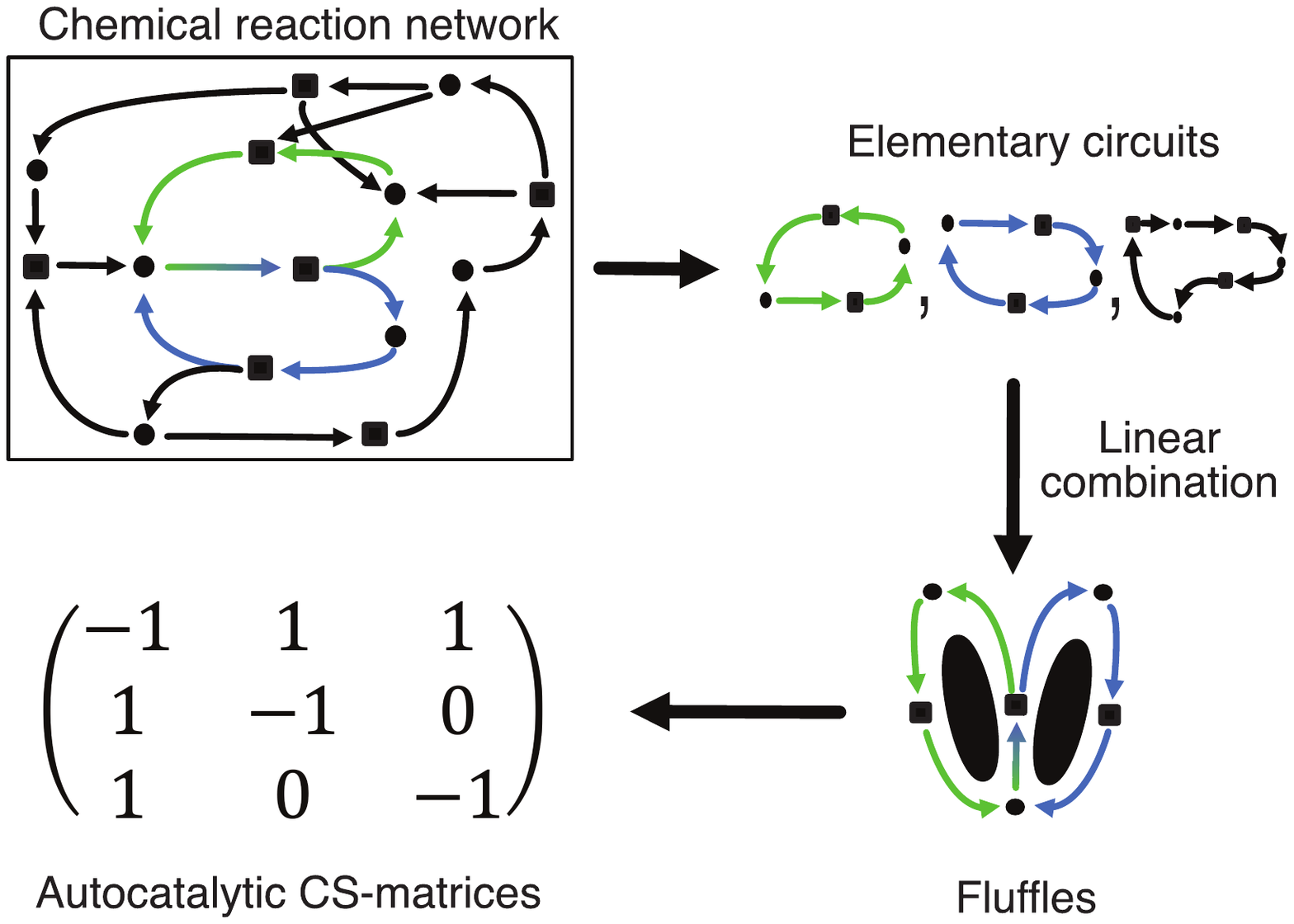}	
\end{figure}

\newpage

\begin{center}
{\Huge \textbf{Supplementary Information}\\}
\end{center}

\section{Background: Graphs and Matrices}

The main text follows well-established textbook terminology and
notation. For completeness, we briefly review the basic definitions
and well-known facts about matchings and circuits in directed graphs.

\subsection{Basic Notation} 

\par\noindent\textbf{Graphs.} We consider here directed graphs $G=(V,E)$
with vertex set $V$ and edge set $E\subseteq V\times V$ without loops,
i.e., $(v,v)\notin E$ for all $v\in V$. Where necessary, we write $V(G)$
and $E(G)$. Whenever there is a directed edge $(u,v)$ from vertex $u$ to
$v$, we say that $u$ is an in-neighbor of $v$, and $v$ is an out-neighbor
of $u$. The in-degree and out-degree of a vertex are the number of
its in-neighbors and out-neighbors, respectively. $H$ is a subgraph of $G$
if $H$ is graph, $V(H)\subseteq V(G)$, and $E(G)\subseteq E(H)$. The
subgraph $H$ is spanning if $V(H)=V(G)$ and induced if $u,v\in V(H)$ and
$(u,v)\in E(G)$ implies $(u,v)\in E(H)$. The adjacency matrix $\mathbf{A}$
of $G$ is the $V\times V$ matrix with entries $\mathbf{A}_{uv}=1$ if
$(u,v)\in E(G)$ and $\mathbf{A}_{uv}=0$ otherwise.

A walk of length $h\ge 0$ in $G$ is an alternating sequence
$(v_0,e_1,v_1,\dots,e_h,v_h)$ of vertices and edges such that
$e_i=(v_{i-1},v_i)$ for $1\le i\le h$. A walk is closed if $v_0=v_h$. It is
a \emph{path} if $i\ne j$ implies $v_i\ne v_j$ and thus also $e_i\ne e_j$.
A closed walk is an \emph{elementary circuit} if $i\ne j$ implies $v_i\ne
v_j$ for $i,j\ne 0$, i.e., if $(v_1,e_2,\dots,e_h,v_h)$ is a path.  A graph
is \emph{strongly connected} if there is a path from $u$ to $v$ for $u,v\in
V$. The underlying undirected graph is obtained from $G$ by ignoring
the direction of the edges. It is equivalent to the \emph{symmetrized graphs}
$G_s$ obtained by setting $V(G_s)=V(G)$ and $(u,v),(v,u)\in E(G_s)$
whenever $(u,v)\in E(G)$. A graph $G$ is \emph{connected} if its underlying
undirected graph is connected in the usual sense for undirected graphs, or
equivalently, if its symmetrized graph $G_s$ is strongly connected.

\medskip
\par\noindent\textbf{Hypergraphs.} Chemical reaction networks can be
represented as directed hypergraphs with vertex set $X$,
representing the chemical species, and a set of directed hyperedges $R$
denoting the reactions. A directed hyperedge $(E^-,E^+)$ is a pair of
non-empty subsets $E^-,E^+\subseteq V$ denoting the reactants and products,
respectively.  The \emph{K{\"o}nig representation} of a directed hypergraph
$(X,R)$ is the directed bipartite graph with vertex set $V=X\cupdot R$ and
edge set $E$ such that $(x,r)\in E$ if there is $r=(E^-,E^+)\in R$ such
that $x\in E^-$, and $(r,x)\in E$ if there is $r=(E^-,E^+)\in R$ such that
$x\in E^+$. Throughout, we denote the K{\"o}nig graph of the CRN under
consideration by $\king$. 

\medskip
\par\noindent\textbf{Linear Algebra.} Given a matrix $A$ with rows and
columns indexed by ordered sets $N$ and $M$, respectively, we denote by
$A[K,L]$ the submatrix with rows indexed by $K\subseteq N$ and columns
$L\subseteq M$. When $|K|=|L|$ the determinant of $A[K,L]$ is called a
\emph{minor}. If $A$ is a square matrix, i.e., $|N|=|M|$, a square
submatrix $A[K,L]$ with $K=L$ is called \emph{principal submatrix} and its
determinant \emph{principal minor}. Since only one set $K$ is needed to
define a principal submatrix, we refer to a principal submatrix as
$A[K]$. A permutation matrix $P$ of size $n$ is a square matrix that has
exactly one entry $1$ in each row and each column, and $0$ elsewhere. A
square matrix is \emph{irreducible} if there exists no permutation matrix
$P$ of size $n$ such that $PAP^{-1}$ is an upper triangular
matrix. Considering corresponding row and column indices as vertices and
introducing an edge $(i,j)$ if $A_{ij}\ne 0$ yields a graph $\Gamma$ that
is strongly connected if and only if $A$ is irreducible. A matrix $A\in
\mathbb{C}^{n\times n}$ is \emph{Hurwitz stable} if all of its eigenvalues
have negative real part. It is called \emph{Hurwitz unstable} if it
possesses at least one eigenvalue with positive real part.

\subsection{Graph Theoretical Constructions}

\medskip
\par\noindent\textbf{Matchings.} A \emph{matching} in $G=(V,E)$ is a subset
$M\subseteq E$ of edges such that each vertex is incident to at most one
edge. The notion of matchings is the same in directed and undirected
graphs, i.e., the direction of the edges does not play a role.  A matching
$M$ is \emph{perfect} if every vertex is incident with an edge in
$M$. While the existence of a perfect matching can be verified in
polynomial time (using any algorithm for computing a maximum matching),
counting the number of perfect matchings is \#P-complete. Enumeration can be
achieved with constant amortized time \cite{fink_constant_2025}. In the
applications below, we consider matchings on directed bipartite
graphs. Writing the vertex partition of the bipartite graph as $V\eqqcolon
X\cupdot R$, we are interested only in matchings $M\subseteq
E_1\coloneqq (X\times R)\cap E$. Clearly, this is equivalent to matchings
in the subgraph $G_1=(V, E_1)$.

\medskip
\par\noindent\textbf{Cycle Bases and Ear Decompositions in Digraphs.} 
Recall that a digraph is \emph{strongly connected} if every vertex is reachable
from every other vertex by means of a directed path. Equivalently, $G$ is
(weakly) connected if the underlying undirected graph is connected, and
every vertex $x\in V(G)$ is contained in an elementary circuit. The cycle
space of digraph is usually defined over $\mathbb{Z}$ as the kernel of the
directed incidence matrix $\mathbf{H}$ with entries $H_{xe}=-1$ if $x$ is
the tail of the edge $e$, $H_{xe}=+1$ if $x$ is the head of $e$, and $0$
otherwise. The (incidence vectors of the) elementary circuits define the
extremal rays of the non-negative cone $\{z\;|\;\mathbf{H}z=0, z_e\ge 0\,
\forall e\in E\}$
\cite{galluccio_pqodd_1996,loebl_remarks_2001,gleiss_circuit_2003}.  Every
strongly connected digraph has a basis of the cycle space that consists of
elementary circuits only, see Thm.~9 of Berge et al. \cite{Berge:73}.

A \emph{directed ear} in a digraph $G$ is a directed path in which all
internal vertices have in-degree $1$ and out-degree $1$, while the initial
vertex has out-degree at least $2$ and the terminal vertex has in-degree at
least $2$. An ear is called \emph{open} if its initial and terminal vertices
are distinct, and \emph{closed} otherwise. A digraph $G$ is strongly connected if
it can be obtained from a single directed cycle by successively adding
(open or closed) ears
\cite{galluccio_pqodd_1996,loebl_remarks_2001,havet_constrained_2019}. The
digraph $G$ is a \emph{strong block} if it is strongly connected and has no
cut vertices. Equivalently, any two vertices in $G$ lie on some elementary
circuit. Moreover, $G$ is a strong block if and only if it can be
constructed by means of an open ear decomposition
\cite{grotschel_minimal_1979}. For the ear decompositions, an \emph{ear
basis} is obtained by completing each ear (after it has been attached) to
an elementary circuit by a directed path from the terminal to the initial
vertex of the ear.

In the main text, we will use the following straightforward property of ear
decompositions for which we could not find a convenient reference, and thus
it is proved here:
\begin{lemma}
  Let $G'$ be a subgraph of $G$ and suppose both $G'$ and $G$ are strong
  blocks.  Then $G'$ can be extended to $G$ by adding a sequence of open
  ears.
\end{lemma}
\begin{proof}
  Consider a vertex $x\in V(G)\setminus V(G')$. Since $G$ is a strong
  block, there is a vertex $y\in V(G')$ and an elementary circuit $C$ that
  contains both $x$ and $y$. Let $u\in V(G')$ be the first
  predecessor of $x$ on $C$ and $v\in V(G')$ the first successor on $C$ in
  $G'$. Then the path $P=(u,\dots,x,\dots,v)$ is an ear. Clearly
  $G''=G'\cup P$ is again a strong block. Thus, all vertices in
  $V(G)\setminus V(G')$ can be added to $V(G')$ sequentially attaching
  ears. The resulting graph $G^*$ is a spanning subgraph of $G$. Any
  missing edges have both endpoints in $V(G^*)=V(G)$ and thus are ears. As
  an immediate consequence, any circuit basis of $G'$ can be extended to a
  circuit basis of $G$ by adding elementary circuits composed of an ear as
  described above and a directed path connecting its attachment vertices in
  the previously constructed subgraph.
\end{proof}
The same argument works for strongly connected graphs if open and closed
ears are allowed.

\section{Proofs of Statements in the Main Text}

\begin{mainitem}{Proposition}{prop:autored}
Let $\child=(X_{\kappa}, R_{\kappa}, \kappa)$ be a $k$-CS whose associated
CS-matrix $\SM[\child]$ is reducible, Metzler, and autocatalytic. Then
there exists a $k'$-CS $\child'=(X_{\kappa'}, R_{\kappa'}, \kappa')$ with
$X_{\kappa'}\subset X_{\kappa}$, $R_{\kappa'}\subset R_{\kappa}$, and
$\kappa'(X_{\kappa'})=\kappa(X_{\kappa'})$, such that its associated
CS-matrix $\SM[\child']$ is an irreducible autocatalytic Metzler matrix.
\end{mainitem}
\begin{proof}
  Since $\SM[\child]$ is reducible, there exists a permutation matrix such that 
  \begin{equation}
    P\SM[\child]P^{-1} = 
    \begin{pmatrix}
      A & 0 \\
    B & C
    \end{pmatrix}
  \end{equation}
  with irreducible $A$. Let now $X_{\kappa'} \subset X_{\kappa}$ represent
  the species corresponding to the rows of $A$. Then the triple $\child'
  \coloneqq (X_{\kappa'}, R_{\kappa'} \coloneqq \kappa(X_{\kappa'}),
  \kappa' \coloneqq \kappa|_{X_{\kappa '}})$ is a CS satisfying
  $\kappa(X_{\kappa'}) = \kappa'(X_{\kappa'})$. Moreover, $\SM[\child']
  \coloneqq A$ has negative diagonal entries since $\SM[\child]$ has
  negative diagonal entries; a fact that is not changed upon
  simultaneous rearrangement of rows and columns. $\SM[\child]$ being a Metzler matrix
  implies that $\SM[\child']$ has only non-negative off-diagonal entries,
  thus $\SM[\child']$ is Metzler. An analogous argument can be made for
  $\SM[\child'']\coloneqq C$.

  In addition, $\SM[\child]$ being autocatalytic implies that there exists
  $v>0:\SM[\child]v>0 \Rightarrow P\SM[\child]v>0$. We let $w\coloneqq Pv$,
  then $w>0$:
  \begin{equation}
    0 < P \SM[\child] v =P \SM[\child] P^{-1} P v =  \begin{pmatrix}
    \SM[\child'] & 0 \\
    B & \SM[\child'']
    \end{pmatrix}w = \begin{pmatrix}
    \SM[\child']w_{1} \\
    Bw_1 + \SM[\child''] w_2
    \end{pmatrix}
  \end{equation}
  Thus $\SM[\child']w_1>0$. Assume now there is a column in $\SM[\child']$
  without a positive entry. Then one reaction of $\SM[\child']$ has no
  product. Hence, there exists a permutation matrix $P'$ such that:
  \begin{equation}
    P'\SM[\child'] P'^{-1}=
    \begin{pmatrix}
      A' & \overset{\rightarrow}{0} \\		
      C' & x	
    \end{pmatrix}
  \end{equation}
  with $x<0$ and $\overset{\rightarrow}{0}\in 0^{(m-1)\times 1}$. However,
  this implies that $\SM[\child']$ is reducible, which is a
  contradiction. Thereby, $\SM[\child']$ is autocatalytic.
\end{proof}

\begin{figure}
	\centering
%
	\centering
		\begin{tikzpicture}
			\node[draw, circle] (x1) at (0,0) {$x_1$};
			\node[draw, circle] (x2) at (4,1) {$x_2$};
			\node[draw, circle] (x3) at (4,-1) {$x_3$};
		
			\node[draw, rectangle] (r1) at (2,0) {$r_1$};
			\node[draw, rectangle] (r2) at (4,3) {$r_2$};
			\node[draw, rectangle] (r3) at (4,-3) {$r_3$};
			
			\draw[->] (x1) -- (r1);
			\draw[->] (x2) -- (r2);
			\draw[->] (x3) -- (r3);
		
			\draw[->] (r1) -- (x2);
			\draw[->] (r1) -- (x3);
			
			\draw[->] (r2) to [bend right = 45] (x1);
			\draw[->] (r2) to [bend left = 45] (x3);
				
			\draw[->] (r3) to [bend right = 45] (x2);
			\draw[->] (r3) to [bend left = 45] (x1);
		\end{tikzpicture}
	\caption{Depiction of a 
	type V autocatalytic core} 
    \label{fig:coreV}
\end{figure}
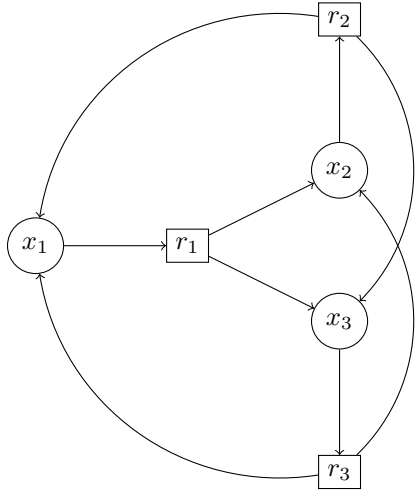

\begin{mainitem}{Lemma}{lem:CSgraph}
  Let $\king'=(X'\cupdot R',E_1'\cupdot E_2')$ be a subgraph of $\king$ with
  reactant vertices $X'$, reaction vertices $R'$, and edges $E_1\subseteq
  X'\times R'$ and $E_2'\subseteq R'\times X'$ 
  such that
  \begin{enumerate}
  \item $|X'|=|R'|$;
  \item every $x\in X'$ has out-degree $1$ and every $x\in R'$ has
    in-degree $1$.
  \end{enumerate}
  Then $\king'$ is child-selective with $\kappa(x)=r$ for
  $(x,r)\in E_1'$. 
\end{mainitem}
\begin{proof}
  Since $x\in X'$ has out-degree $1$ and $\king'$ is bipartite, there is a
  unique $\kappa(x)\in R'$. Analogously, for every $r\in R'$ there is a
  unique $\mu(r)\in X'$ and we have $(x,\kappa(x))\in E_1'$ for all $x\in
  X'$ as well as $(\mu(r),r)\in E_1'$ for all $r\in R'$. Thus we have
  $\kappa(\mu(r))=r$ and $\mu(\kappa(x))=x$, i.e., $\mu(\kappa)$ is the
  identity on $X'$ and $\kappa(\mu)$ is the identity on $R'$. Hence
  $\kappa$ is a bijection and $\child=(X',R',\kappa)$. By construction, we
  have $\king'=\king(\child)$.
\end{proof}

\begin{mainitem}{Lemma}{lem:scirreduc}
  Let $\child=(X_\kappa,E_\kappa,\kappa)$ be a CS. Then $\king(\child)$ is
  strongly connected if and only if $\Metzler{\SM[\child]}$ is irreducible.
\end{mainitem}
\begin{proof}
  There is a path $(x,r,y)$ in $\king(\child)$ if and only if $r=\kappa(x)$
  and $y$ is a product of $r$ in $X_{\kappa}$, which is the case if and
  only if $\SM_{y\kappa(x)}=\SM[\child]_{yx} >0$, 
  and hence if and only if
  $\mathbf{M}_{xy}\coloneqq \Metzler{\SM[\child]}_{xy}>0$.
  \\ First suppose
  $\king(\child)$ is strongly connected. Then there is a path from $x$ to
  $y$ for all $x,y\in X_{\kappa}$, say $(x=z_0,r_1,z_1,\dots,r_k,y=z_k)$
  with $r_i=\kappa(z_{i-1})$. Thus $\mathbf{M}_{z_{i-1}z_{i}}>0$ and hence
  $\mathbf{M}$ is irreducible.  Conversely, suppose $\mathbf{M}$ is
  irreducible. Then for every pair $x,y\in X_{\kappa}$ there is a sequence
  of vertices $z_i\in X_{\kappa}$ with $x=z_0$ and $y=z_k$ such that
  $\mathbf{M}_{z_{i-1}z_{i}}>0$ and hence
  $(x=z_0,\kappa(z_0),z_1,\kappa(z_1),\dots,\kappa(z_{k-1}),y=z_k)$ is a
  path in $\king(\child)$. Moreover, for every reaction vertex $r$ there is
  an edge $(x,r)$ with $x=\kappa^{-1}(r)$ and an edge $(r,y)$ since there
  is $y\in X_{\kappa}$ with $\mathbf{M}_{\kappa^{-1}(r),y}>0$. Hence all
  vertices of $\king(\child)$ are reachable from each other.
\end{proof}

\begin{mainitem}{Lemma}{lem:cutvertex}
  If $\king(\child)$ is strongly connected, then it does not contain a
  cut vertex.
\end{mainitem}
\begin{proof}
  Indirectly assume that $v$ is a cut vertex and $\king(\child)$ is
  strongly connected: this assumption implies that $v$ has at least two
  in-edges and two out-edges. If $v$ is reaction vertex, $v=r$, its
  in-edges in $\king(\child)$ are of the form $(\kappa^{-1}(r),r)$.  Since
  $\kappa$ is a bijection, there is at most one such edge, and thus we
  reach a contradiction. If $v$ is in turn a substrate vertex, i.e.,
  $v=x$, then all its out-edges are of the form $(x,\kappa(x))$,
  i.e., again, there is exactly one such edge, and thus we reach a
  contradiction analogously. Therefore, $\king(\child)$ cannot
  contain a cut vertex.
\end{proof}

\begin{mainitem}{Lemma}{lem:nosourcesink}
  A CS $\child=(X_\kappa,E_\kappa,\kappa)$ is autocatalytic if and only if
  the following conditions both hold:
  \begin{enumerate}
  \item there is a positive vector $v>0$ such
    that $\SM[\child]v >0$.
  \item 
    $\king(\child)$ does not possess source and sink vertices;
  \end{enumerate} 
\end{mainitem}
\begin{proof}
  Property 1 is identical to property (i) of
  Def.~M\ref{def:autocatmatrix}.\\ First, suppose $\child$ is
  autocatalytic. Property (i) of Def.~M\ref{def:autocatmatrix}
  further implies that no substrate vertex is a source in $\king(\child)$
  because the row $\SM[\child]_x$, corresponding to substrate $x$,
  satisfies $\SM[\child]_x v>0$ and thus $x$ is a product in at least one
  reaction.  By Eq.~(M\ref{eq:kingchild}), this implies that $x$ is
  not a source in $\king(\child)$. Property (ii) in
  Def.~M\ref{def:autocatmatrix}, on the other hand, implies that no
  reaction vertex is a sink. By Thm.~M\ref{thm:CSPerfMatch}, the CS
  bijection $\kappa$ explicitly guarantees the existence of a perfect
  matching in the set of reactant-to-reaction edges in
  $\king(\child)$. Thus, no substrate vertex is a sink and no reaction
  vertex is a source $\king(\child)$. In summary, statement 2 of the lemma
  is satisfied.\\ Conversely, suppose conditions 1 and 2 hold. Since
  $\child$ is a CS and there is no source or sink vertex in $\king(\child)$,
  then for every reaction vertex $r$, there is an edge $(x,r)$ and an edge
  $(r,y)$ and thus for every reaction $r$ (that is, for every column) we
  have $\SM[\child]_{xr}<0$ and $\SM_{ry}>0$, i.e., condition (ii) in
  Def.~M\ref{def:autocatmatrix} is satisfied. Together with (i),
  $\SM[\child]$ is autocatalytic.
\end{proof}

\begin{mainitem}{Theorem}{thm:Woffle}
  A graph $G$ is a fluffle if and only if it is bipartite with vertex set
  $X\cupdot R$ and it has an ear decomposition such that every ear
  initiates in a reaction vertex $r\in R$ and terminates in a substrate
  vertex $x\in X$. In this case, all directed open ear decompositions have
  this property.
\end{mainitem}
\begin{proof} 
  Let $G$ be a fluffle. Then $G$ is bipartite and a strong block. Thus, in
  particular, it has a directed open ear decomposition
  \cite{grotschel_minimal_1979}. Now consider \emph{any}
  decomposition ($P_1=C$, $P_2$, \dots, $P_h$), where $h\ge 1$ and $P_1$ is
  an elementary circuit and $P_h$, $h\neq 1$, is a path. As any
  ear decomposition starts from an elementary circuit $C$, it follows that
  for any substrate vertex $x\in V(C)\cap X$ of $C$, an out-neighbor of $x$
  is contained as well in $V(C)$. By condition (ii) in
  Prop.~\ref{prop:usefluffles}, such out-neighbors must be unique in the
  fluffle, and thus any substrate vertex $x\in V(C)\cap X$ cannot be an
  initial vertex of any ear.  Respectively and in total analogy, a reaction
  vertex $r\in V(C)\cap R$ of the elementary circuit $C$ cannot be a
  terminal vertex of an ear, since its only in-neighbor is also located
  along $V(C)$. Now let $G_2$ be the graph obtained by attaching the ear
  $P_2$ to $C$. Since $G_2$ is 2-connected, every substrate vertex $x\in
  V(G_2)\cap X$ already has a unique out-neighbor in $V(G_2)$ and every
  reaction vertex $r\in V(G_2)\cap R$ has a unique in-neighbor in
  $V(G_2)$. Thus, the next ear can only initiate at a vertex $r'\in
  V(G_2)\cap R$ and terminate at a vertex $x'\in V(G_2)\cap
  R$. Inductively, this argument holds true for all subsequent
  ears. Moreover, all ear decompositions are of this form.
  
  Conversely, assume that $G$ is bipartite with vertex partition $X\cupdot
  R$, let ($C$, $P_2$, \dots, $P_h$) be an ear decomposition of $G$ such
  that each $P_i$ initiates in a reaction vertex $r\in V(G)\cap R$ and
  terminates in a substrate vertex $x\in V(G) \cap X$. Since $G$ has an ear
  decomposition, then $G$ is a strong block, i.e., $G$ satisfies (iii) in
  Prop.~\ref{prop:usefluffles}. Moreover, as $G$ is bipartite, the
  vertices along each ear $P_i$ alternate between reaction vertices in $R$
  and substrate vertices in $X$, and since the first and last vertex of
  $P_i$ belong to different sets we have $|V(P_i)\cap X|=|V(P_i)\cap
  R|=|V(P_i)|/2$. Writing $G_1=C$ and $G_i$ for the graph obtained by
  attaching the ear $P_i$ to $G_{i-1}$, we have $|V(G_1)\cap X|=|V(G_1)\cap
  R|=|C|/2$ and $|V(G_i)\cap R|=|V(G_{i-1})\cap R|+(|P_i|/2-1)$ as well as
  $|V(G_i)\cap X|=|V(G_{i-1})\cap X|+(|P_i|/2-1)$, where the $-1$ accounts
  for the fact that initial and terminal vertices of $P_i$ are already
  present in $G_{i-1}$. By induction, it follows immediately that
  $|V(G_i)\cap R|=|V(G_i)\cap X|$ for all $i$, and thus $|R(G)|=|X(G)|$,
  i.e., $G$ satisfies (i). By construction, every substrate vertex $x\in X$
  has a single out-neighbor $r_x$. We have $r_x\in V(C)$ if $x\in V(C)$ and
  $r_x\in V(P_i)$ if $x\in V(P_i)$. Similarly, every reaction vertex $r\in
  R$ has a single in-neighbor $x_r$ satisfying $x_r\in V(C)$ if $r\in V(C)$
  and $x_r\in V(P_i)$ if $r\in V(P_i)$. Thus $G$ satisfies property (ii).
  Taking (iii), (i), and (ii) together, $G$ is a fuffle.
\end{proof}

\begin{mainitem}{Lemma}{lem:subwoffle}
  Let $G$ be a fluffle in $\king$ and $G'$ a subgraph of $G$ that
  is a strong block. Then $G'$ is a fluffle.
\end{mainitem}
\begin{proof}
  Trivially, $G'$ is bipartite, each substrate vertex $x\in X(G')$ has
  out-degree at most $1$ and each reaction vertex $r\in R(G')$ has
  in-degree at most $1$. Since $G'$ is a strong block by assumption, it has
  no vertices with in-degree or out-degree $0$, i.e., every $x\in X(G')$
  has out-degree $1$ and every $r\in R(G')$ has in-degree $1$. Thus $G'$
  satisfied condition \textit{(ii)}. Moreover, $|X(G')|=|R(G')|$ is
  satisfied because indirectly $|X(G')|<|R(G')|$ would imply that there is
  a reaction vertex $r\in R(G')$ without in-edge and $|X(G')|>|R(G')|$ would
  imply that there is a substrate vertex $x\in X(G')$ without out-edge,
  both leading to a contradiction to property \textit{(ii)}, which we just
  proved. Hence $G'$ also satisfies \textit{(i)} and is a fluffle.
\end{proof}

\begin{mainitem}{Theorem}{thm:woffle-union}
  Let $G$ be a fluffle with vertex partition $X\cupdot R$ and $C$ an
  elementary circuit such that $\emptyset \subset G\cap C \subset C$. Then,
  the connected components of $G\cap C$ are directed paths $P_i$. Moreover,
  $G\cup C$ is a fluffle if and only if all such paths $P_i$ start from a
  substrate vertex $x_i \in X$ and terminate with a reaction vertex $r_i
  \in R$.
\end{mainitem}
\begin{proof}
  Trivially, any connected component of proper subsets of an elementary
  circuit is a path whenever it starts and terminates
  with a vertex. In particular, then, the connected components $P_i$ of the
  intersection $\emptyset \subset G\cap C \subset C$ are paths. Without
  loss of generality, we can arrange such paths $P_i$ in circular order
  along $C$. Let then $Q_i$ identifies the path in $C$ that
  starts with the terminal vertex of $P_{i-1}$ and terminates with
  the starting vertex of $P_{i}$ (here the index $i$ is to be intended in a
  cyclic group). Note that the $Q_i$ are in a 1-to-1 relation with
  the connected components of the complement $C\setminus G$ of $C\cap G$ in
  $C$: they are obtained by adding to each connected component of
  $C\setminus G$ the starting and terminal vertices to obtain a
  path. Clearly, by construction, the paths $Q_i$ are ears for $G$. Since
  $G$ is a fluffle, then it admits itself an ear decomposition. Because the
  paths $Q_i$ are vertex-disjoint, any arbitrary ear decomposition for $G$
  can then be extended to an ear decomposition for $G\cup C$ by adding in
  arbitrary order the paths $Q_i$. Theorem~M\ref{thm:Woffle}
  therefore implies that $G\cup C$ is a fluffle if and only if each ear
  $Q_i$ initiates in a reaction vertex $r\in R$ and terminates in a
  substrate vertex $x\in X$. By complementary construction, this is
  equivalent to each path $P_i$ initiating in a substrate vertex $x\in X$
  and terminating in a reaction vertex $r\in R$.
\end{proof}

\begin{mainitem}{Lemma}{lem:circuitnetequiv}
  Two circuitnets $\mathcal{C}_1$ and $\mathcal{C}_2$ for fluffles $G_1$
  and $G_2$ yield the same CS matrix $\SM[\child]$ if and only if
  $\mathcal{C}_1 \bumpeq \mathcal{C}_2$.
\end{mainitem}
\begin{proof}
  Two circuitnets yield the same CS-matrix if and only if
  Eq.~(M\ref{eq:samecs}) holds. That is, for the graphs
  $\bigcup(\mathcal{C}_1)=(X^{(1)}\cup R^{(1)}, E^{(1)}_1 \cup E^{(1)}_2)$,
  $\bigcup(\mathcal{C}_2)=(X^{(2)}\cup R^{(2)}, E^{(2)}_1 \cup E^{(2)}_2)$
  it holds that $X^{(1)}=X^{(2)}$, $R^{(1)}=R^{(2)}$,
  $ E^{(1)}_1=E^{(2)}_1$, and in particular
  $\mathcal{C}_1\bumpeq \mathcal{C}_2$. In turn, since $\mathcal{C}_1$ and
  $\mathcal{C}_2$ are circuitnets for fluffles $G_1$ and $G_2$, then
  $E_1=(x_i,r_j)$ is a perfect matching (Prop.~M\ref{prop:usefluffles}
  and Eq.~(M\ref{eq:kingchild})) and thus fully specifies $(X,R)$.
\end{proof}

\begin{mainitem}{Lemma}{lem:bumpeq}
  Let $\mathcal{C}_1$ and $\mathcal{C}_2$ be circuitnets for fluffles $G_1$
  and $G_2$, respectively, and let $C'$ and $C''$ be two elementary
  circuits.  Assume $\mathcal{C}_1 \bumpeq \mathcal{C}_2$, $C'\bumpeq C''$
  and $G_1'\coloneqq \bigcup(\mathcal{C}_1\cup\{C'\})$ is a fluffle. Then
  $G_2'\coloneqq \bigcup(\mathcal{C}_2\cup\{C''\})$ is a fluffle as well
  with $\mathcal{C}_1\cup\{C'\} \bumpeq \mathcal{C}_2\cup\{C''\}$.
\end{mainitem}
\begin{proof} 
  First consider any fluffle $G=\bigcup(\mathcal{C})$ with circuitnet
  $\mathcal{C}$ and let $C$ be an elementary circuit such that $G\cup C$ is
  a fluffle.  By construction, we have $E_1(G\cup C)=E_1(G)\cup
  E_1(C)$. Moreover, $E(G)\cap E(C)\ne\emptyset$ since $G\cup C$ is
  fluffle, and hence a strong block. In particular, via
  Thm.~M\ref{thm:woffle-union}, $G$ and $C$ share at least one
  directed path $P$ initiating at a substrate-vertex and terminating at a
  reaction-vertex, i.e., $E_1(G)\cap E_1(C)\ne\emptyset$. Hence we have
  $E_1(G_1)\cup E_1(C')= E_1(G_1) \cup E_1(C'') =E_1(G_2)\cup
  E_1(C')=E_1(G_2)\cup E_1(C'')$ and $\emptyset \neq E_1(G_1)\cap
  E_1(C')=E_1(G_2)\cap E_1(C'')$, and thus
  Thm.~M\ref{thm:woffle-union} applies for $G'_2$, concluding that
  $G'_2$ is a fluffle. Def.~M\ref{def:csequivalence} implies
  $\mathcal{C}_1\cup\{C'\} \bumpeq \mathcal{C}_2\cup\{C''\}$.
\end{proof}

\begin{mainitem}{Lemma}{lem:bumpeq_enum}
  For every CS-equivalence class $[\mathcal{C}]$ there is a representative
  $\hat{\mathcal{C}}$ such that there exists a CS-equivalence class
  $[\mathcal{C}']$ with representative $\hat{\mathcal{C}}'$ and an
  elementary circuit $C^*$ such that $\hat{\mathcal{C}}'\cup\{C^*\} \bumpeq
  \hat{\mathcal{C}}$ and $|V(\bigcup(\mathcal{C}'))| <
  |V(\bigcup(\mathcal{C}))|$.
\end{mainitem}
\begin{proof}
  The statement trivially holds for circuitnets that are single-elementary
  circuits. Let $G$ be the fluffle associated with the circuitnet
  $\mathcal{C}=\{C_1,\dots,C_h\}$, $h>1$, listed according to the ordering
  in Def.~M\ref{def:cnet}. Pick now the first circuitnet
  $\hat{\mathcal{C}} \subseteq \mathcal{C}$ for $G$ such that any strict
  subset of $\hat{\mathcal{C}}$ is \emph{not} anymore a circuitnet for
  $G$. Clearly, we can always find such a suitable candidate
  $\hat{\mathcal{C}}$ from any circuitnet $\mathcal{C}$ for $G$ by
  iteratively removing the single elementary circuit $C^*_i$ with the
  highest index $i$ and checking whether the remaining set is a circuitnet
  for the very same fluffle $G$. Once $\hat{\mathcal{C}}$ is found, a
  further removal of the elementary circuit $\hat{C}^*_i$ with highest
  index $i$ identifies a circuitnet
  $\hat{\mathcal{C}}'=\hat{\mathcal{C}}\setminus \hat{C}^*_i$ for
  $G'\subset G$ which, by Cor.~M\ref{cor:superimposingfluffle}, is itself
  a fluffle.  Moreover, since $G'\subsetneq G$ and the removal of
  $\hat{C}^*_i$ removes an ear, $|V(G')|<|V(G)|$.
\end{proof}

\begin{mainitem}{Lemma}{lem:expandAutoMetzler}
  Let $\king(\child^*)$ be a fluffle with irreducible autocatalytic Metzler
  CS matrix $\SM[\child^*]$ and let $\king(\child)$ be obtained from
  $\king(\child^*)$ by adding a single ear with initial vertex in
  $R(\king(\child^*))$, terminal vertex in $X(\king(\child^*))$, and a
  non-empty set of internal vertices, together with all
  reaction-to-metabolite edges in $R(\king(\child))\times
  X(\king(\child))$.  If $\SM[\child^*]$ is an autocatalytic CS matrix and
  $\SM[\child]$ is a Metzler matrix, then $\SM[\child]$ is an autocatalytic
  irreducible CS matrix.
\end{mainitem}
\begin{proof}
  Since $\king(\child^*)=\king[\child^*]$ is in particular a strong block,
  the addition of an ear makes the resulting graph $G$ also a strong block,
  and thus a fluffle by Thm.~M\ref{thm:woffle-union} because the
  ear can be extended to an elementary circuit by any directed path
  in $\king(\child^*)$ from its terminal to its initial vertex. Inserting
  the additional $R$-to-$X$ edges does not affect the fluffle
  property, completing it to the corresponding representative
  CS-equivalence class, i.e., $\king(\child)$. Since $\SM[\child^*]$ is
  irreducible by assumption, $\Metzler{\SM[\child]}$ is also irreducible.

  Now suppose $\Metzler{\SM[\child]}$ is a Metzler matrix and let
  $\mathbf{A}$ be the matrix obtained by renumbering the vertices such that
  the initial vertex of the ear is $k=|X(\king[\child^*])|$, its terminal
  vertex is $1$, and the substrate vertices are ordered consecutively along
  the directed ear from $k+1$ to $l$. By construction $\mathbf{A}$ has the
  form
    \begin{equation*}
      \begin{pmatrix}
        \framebox{$\mathbf{A^*}$ } &  &     &    \dots & & & \vec{g}_l \\
        f_{k+1} & -a_{k+1}   &   &    \dots & & &\\  
                &  f_{k+2}   & -a_{k+2} & \dots & & &\\
                &  & f_{k+3} &    \ddots   & & & & \\
                &  &  &\ddots    &-a_{l-2} &   & \\
                &  &  &\dots    & f_{l-2}   & -a_{l-1} &  \\
                &  &  &\dots &  & f_{l-1}     & -a_{l} \\    
    \end{pmatrix}
  \end{equation*}
  where the vector $\vec{g}\ge 0$ has a strictly positive first entry, and
  all $a_i$ and $f_i$ are strictly positive. Moreover, all entries that are
  left blank are non-negative, since $\mathbf{A}$, like 
  $\SM[\child]$, is a Metzler matrix.
  Multiplying $\mathbf{A}$ with a strictly positive vector
  $\vec{u}=(\vec{u}^*,u_{k+1},u_{k+2},\dots,u_l)^{\top}$ yields
  $\mathbf{A}\vec{u}=\vec{z}=(\vec{z}^*,z_{k+1},z_{k+2},\dots,z_l)^{\top}$.
  Taking into account that the blank entries
  yield only non-negative contributions, we obtain component-wise
  inequalities for $\vec{z}$ from the terms that are shown explicitly:
    \begin{equation*}
    \begin{split} 
      \vec{z}^* & \ge \mathbf{A}^*\vec{u}^* \\
      z_{j} &\ge f_{j} u_{j-1}-a_{j}u_{j} \quad \text{for } k+1\le j\le l \\
    \end{split}
  \end{equation*}
  Hence $\mathbf{A}^*\vec{u}^*>0$ implies $\vec{z}^*>0$.
  The second set of inequalities implies $z_{j}>0$ whenever
  $0<u_{j}<(f_j/a_j)\cdot u_{j-1}$, for
  $k+1\le j\le l$. With $u_{k}=u^*_k$ fixed, we can recursively choose
  $0<u_{j}=(f_{j}/a_{j})\cdot u_{j-1}-\epsilon_j$ with $\epsilon_j>0$ small
  enough for $k+1\leq j \leq l$. By induction from
  $j=k+1$ down to $j=l$, therefore there is always a positive
  choice of $u_j$ that yields a positive entry $z_j$. In summary,
  therefore, if $\mathbf{A^*}\vec{u}^*>0$, i.e., if $\SM[\child^*]$
  is autocatalytic, then there is $\vec{u}>0$ such that
  $\mathbf{A}\vec{u}>0$, i.e., such that $\SM[\child]$ is autocatalytic.
\end{proof}

\begin{mainitem}{Theorem}{thm:avoidMetzlerCheck}
  Let $\SM[\child]$ be an irreducible Metzler CS matrix and suppose
  $\SM[\child]$ contains an autocatalytic core $\SM[\child^*]$ as a
  principal submatrix. Then $\SM[\child]$ is autocatalyic.
\end{mainitem}
\begin{proof}
  It suffices to recall that any fluffle can be obtained from a sub-fluffle
  by adding ears. Going from a fluffle to the canonical representative of
  its CS-equivalence class amounts to adding edges of the form
  $(r,x)$, i.e., ears without internal vertices. Thus if
  $\SM[\child^*]$ is an autocatalytic core that is a principal
  submatrix of $\SM[\child]$ there is a sequence of ears, and thus
  a corresponding sequence of child-selections $\child^*=\child_0,
  \child_1,\dots,\child_h=\child$, such that $\king(\child_i)$ is obtained
  from $\king(\child_{i-1})$ by adding an ear with a non-empty set of
  interior vertices. Since $\SM[\child_h]$ is Metzler and all
  $\SM[\child_i]$, $0\le i\le h$ are principal submatrices of
  $\SM[\child_h]$, each of the $\SM[\child_i]$ is an irreducible Metzler CS
  matrix.  Applying Lemma~M\ref{lem:expandAutoMetzler} to each of
  the steps from $\SM[\child_{i-1}]$ to $\SM[\child_{i}]$ for $1\le i\le h$
  now implies that $\SM[\child_{i}]$ is autocatalytic whenever
  $\SM[\child_{i-1}]$ is autocatalytic.
\end{proof}

\begin{mainitem}{Lemma}{lem:AB}
  Let $G$ be a fluffle and $C$ an elementary circuit. Then $G\cup C$
  is a fluffle if and only if $\emptyset\ne V(G)\cap V(C) = V(E_1(G) \cap
  E_1(C))$.
\end{mainitem}
\begin{proof}
  Since fluffles are connected by definition, we may assume that $V(G)\cap
  V(C)\ne \emptyset$. We observe $B\coloneqq V(E_1(G) \cap E_1(C))\subseteq
  V(E_1(G)) \cap V(E_1(C)) = V(G) \cap V(C) \eqqcolon A$.  First assume
  $A=B$. Thus $V(G)\cap V(C)\ne\emptyset$ implies that $E_1(G)\cap
  E_1(C)\ne\emptyset$ and thus $G\cup C$ is a strong block. Moreover, every
  edge in $E_1(C)$ is either contained in $G$ or disjoint from $G$, and
  thus every maximal path in the intersection $G\cap C$ initiates with a
  metabolite $x\in X$ and terminates with a reaction $r\in
  R$. Thm.~M\ref{thm:woffle-union} now implies that $G\cup C$ is a
  fluffle.  For the converse, assume that there is $z\in A\setminus B$. If
  $z\in X$, then there is a unique $y_1\in V(G)$ with $(z,y_1)\in E_1(G)$
  and $y_2\in V(C)$ with $(z,y_2)\in E_1(C)$.  We have $y_1\ne y_2$ since
  otherwise $(z,y_1)=(z,y_2)\in E_1(G)\cap E(G)$. Thus $z\in X$ has
  out-degree $2$ in $G\cup C$ and hence $G\cup C$ is not a
  fluffle. Similarly, if $z\in R$, there is $(y_1,z)\in E_1(G)$ and
  $(y_2,z)\in E_1(G)$ with $y_1\ne y_2$ and thus $z$ has in-degree $2$ in
  $G\cup C$, which therefore is not a fluffle.
\end{proof}

\begin{mainitem}{Proposition}{prop:StabOszil}
  Let $\SM[\child]$ be a Hurwitz-stable autocatalytic CS-matrix. Then there
  exists a choice of parameters such that the system
  M\eqref{eq:ODEdynamics}: $\dot{z}=f(z) \coloneqq \SM \cdot
  v(z)$ admits periodic solutions.
\end{mainitem}
\begin{proof}
  Blokhuis et al. \cite{blokhuis_stoichiometric_2025} showed that any CS-matrix
  that is Hurwitz-stable but $D$-unstable admits a parameter choice for
  which~M\eqref{eq:ODEdynamics} has periodic solutions. Here,
  $D$-unstable means that there exists a positive diagonal matrix $D$ such
  that $\SM[\child]D$ is Hurwitz-unstable. Autocatalyticity of
  $\SM[\child]$ implies $D$-instability: indeed, $\SM[\child]$ contains an
  autocatalytic core $\Auto[\kappa']$ as a principal submatrix, which is
  Hurwitz-unstable by Prop.~M\ref{cor:SingularMetzlerHurwitz}. Without
  loss of generality, let $\Auto[\kappa']$ be the leading $k'$-dimensional
  principal submatrix of $\SM[\child]$, and define
  $D(\varepsilon)=\operatorname{diag}Po(1_1,\dots,1_{k'},
  \varepsilon_{k'+1},\dots,\varepsilon_k)$. For $\varepsilon=0$,
  $\SM[\child]D(0)$ is Hurwitz-unstable, as is $\Auto[\kappa']$. By
  continuity of eigenvalues, $\SM[\child]D(\varepsilon)$ remains
  Hurwitz-unstable for $\varepsilon$ small enough, so $\SM[\child]$ is
  $D$-unstable. Stability together with $D$-instability implies the claim.
\end{proof} 

\begin{mainitem}{Proposition}{prop:centralizedDet}
  Let $\SM[\child]$ be a $k\times k$ irreducible autocatalytic Metzler
  CS matrix that exhibits \emph{centralized autocatalysis}.  Then
\begin{equation}
  \frac{\operatorname{det}\SM[\child]}{\prod_{m=1}^{k}\SM[\child]_{mm}}=
  1-\sum_{C}\prod_{m\in C}
  \frac{\SM[\child]_{m, C(m)}}{\vert \SM[\child]_{mm} \vert}
\end{equation}
where the sum runs on all permutation cycles.
\end{mainitem}
\begin{proof} 
We recall the notation $P_k$ for the permutation group on $k$ elements. The
first step is noting that if $\SM[\child]$ is centralized with center
$m^*$, then for each permutation $\pi\in P_{k}$ with nonzero contribution,
i.e. such that $\prod_{m=1}^{k} \SM[\child]_{m, \pi(m)}\neq 0,$ we get that
there exists exactly one permutation cycle $C_\pi$ with $\pi=C_\pi$, i.e.,
any permutation with nonzero contribution is a single-cycle permutation. To
confirm this, assume indirectly that there exists a permutation $\pi$ with
nonzero contribution and such that $\pi=C_1\cdot...\cdot C_i$, with
$i\ge2$. In particular, $C_1$ and $C_2$ have disjoint support and thus
$C(m^*)\neq m^*$ cannot hold for both $C_1$ and $C_2$, which leads to a
contradiction with the definition of centralized autocatalysis.  The second
step is just computing
\begin{equation}
    \dfrac{\det \SM[\child]}{\prod_{m=1}^k \SM[\child]_{mm}},
\end{equation}
where the numerator is expanded via the Leibniz formula.
\begin{equation}
  \begin{split}
    \dfrac{\det \SM[\child]}{\prod_{m=1}^k \SM[\child]_{mm}}
    & = \dfrac{\sum_{\pi\in P_k} \operatorname{sgn}(\pi)\prod_{m=1}^k
      \SM[\child]_{m,\pi(m)}}{\prod_{m=1}^k \SM[\child]_{mm}}\\
    & = \dfrac{\prod_{m=1}^k \SM[\child]_{mm}}{\prod_{m=1}^k
      \SM[\child]_{mm}}
      +\dfrac{\sum_{C}
      \operatorname{sgn}(C)\prod_{m\in C} \SM[\child]_{m,C(m)}
      \prod_{m\not\in C}\SM[\child]_{mm}}{\prod_{m=1}^k \SM[\child]_{mm}}\\
    & = 1+\sum_{C} (-1)^{|C|-1} \prod_{m\in C}\dfrac{
      \SM[\child]_{m,C(m)}}{
      \SM[\child]_{mm}}\\  
    & = 1-\sum_{C} \prod_{m\in C} \dfrac{ \SM[\child]_{m,C(m)}}{
      |\SM[\child]_{mm}|}\\ 
  \end{split}
\end{equation}
\end{proof}

\begin{mainitem}{Lemma}{lem:ElemPermCycles}
  Let $\SM[\child]$ be an autocatalytic Metzler CS matrix and
  $\mathfrak{K}[\child]$ the set of permutation cycles with non-zero
  contribution of length $\vartheta \geq 2$ (i.e., nontrivial cycles).  Then there is
  a one-to-one correspondence between $\mathfrak{K}[\child]$ and the
  elementary circuits of the induced subgraph $\king[\child]$ such that a
  permutation cycle $(x_1,x_2,\dots, x_\vartheta)$ corresponds to the elementary
  circuit
  $(x_1,\kappa(x_1),x_2,\kappa(x_2),\dots,\kappa(x_{\vartheta-1}),
  x_\vartheta,\kappa(x_\vartheta),x_1)$ in $\king[\child]$.
\end{mainitem}

\begin{proof}
Let $\pi_c=(x_1,x_2,...,x_\vartheta)$ be a nontrivial permutation cycle with nonzero contribution for the CS-matrix $\SM[\child]$, i.e., with index $i$ following the labeling along the cycle,
\begin{equation}\label{eq:nonzerocont}
    \prod_{i=1}^\vartheta \SM[\child]_{\pi_c(i)i}\neq 0.
\end{equation}
Since $\SM[\child]$ is Metzler, the nonzero off-diagonal entries are positive and correspond to products of the reaction with the respective column index. In particular, then, eq.~\eqref{eq:nonzerocont} holds if and only if $\SM[\child]_{\pi_c(i)i}=s^{+}_{\pi_c(i)i}> 0$ for any $i=1,...,\vartheta.$ Moreover, since $\SM[\pmb{\kappa}]$ is a CS-matrix, $\SM[\pmb{\kappa}]_{ii}=s^-_{i\kappa(i)}> 0$ for any $i=1,...,\vartheta.$ These two observations are equivalent to the existence of edges in the induced subgraph $\king[\child]$. Respectively, the existence of the edge $(x_i,\kappa(x_i))$ is equivalent to $s^-_{i\kappa(i)}> 0$ and the existence of the edge $(\kappa(x_i),x_{\pi_c(i)})$ is equivalent to $s^{+}_{ \pi_c(i)i}>0$. Following the index $i$ along any contributing permutation cycle thus identifies one elementary circuit in the induced subgraph $\king[\child]$. In turn, following the index $i$ along any elementary circuit in the induced subgraph $\king[\child]$ identifies one contributing nontrivial permutation cycle. The bijection follows.
\end{proof}

To see that the statement does not necessarily hold for a non-Metzler matrix, consider:
\begin{equation}
	\SM[\child] = \begin{pmatrix}
				-1 & -1  \\
				1 & -1  \\				
				\end{pmatrix}
\end{equation}
which has one permutation cycle with nonzero contribution, i.e., $(x_1,x_2)$, while the induced subgraph $\king[\child]$ has no elementary circuit.

\begin{mainitem}{Theorem}{thm:NgheClasses}
  An autocatalytic core of type I, {II}$^*$,
  III, or IV is centralized. Moreover,
  Eq.~(\ref{eq:propcentralized}), i.e.
  \begin{equation*}
\frac{\operatorname{det}\SM[\child]}{\prod_{m=1}^{k}\SM[\child]_{mm}}
  =1-\sum_{C}\prod_{m\in C}
  \frac{\SM[\child]_{m, C(m)}}{\vert \SM[\child]_{mm} \vert},
  \end{equation*}
  holds for all five types of cores.
\end{mainitem} 
\begin{proof}
  We prove the theorem identifying the five autocatalytic cores exactly
  with the following five motifs and associated CS-matrices, respectively,
\begin{equation}
\textbf{Type I:} \quad\qquad x_1\underset{1}\longrightarrow x_2 \underset{2}\longrightarrow 2x_1 \qquad\quad \begin{pmatrix}
-1 & 2\\
1 & -1
\end{pmatrix}
\end{equation}
\begin{equation}
\textbf{Type II$^*$}: \quad\quad 
\begin{cases} 
x_1\underset{1}\longrightarrow x_2+x_3 \\
x_2\underset{2}\longrightarrow x_3\\
x_3\underset{3}\longrightarrow x_1
\end{cases}
\quad\quad \begin{pmatrix}
-1 & 0 & 1\\
1 & -1 & 0\\
1 & 1 & -1
\end{pmatrix}
\end{equation}
\begin{equation}
\textbf{Type III:} \quad\quad 
\begin{cases} 
x_1\underset{1}\longrightarrow x_2+x_3 \\
x_2\underset{2}\longrightarrow x_1\\
x_3\underset{3}\longrightarrow x_1
\end{cases}
\quad\quad \begin{pmatrix}
-1 & 1 & 1\\
1 & -1 & 0\\
1 & 0 & -1
\end{pmatrix}
\end{equation}
\begin{equation}
\textbf{Type IV:} \quad\quad 
\begin{cases} 
x_1\underset{1}\longrightarrow x_2+x_3 \\
x_2\underset{2}\longrightarrow x_1 + x_3\\
x_3\underset{3}\longrightarrow x_1
\end{cases}
\quad\quad \begin{pmatrix}
-1 & 1 & 1\\
1 & -1 & 0\\
1 & 1 & -1
\end{pmatrix}
\end{equation}
\begin{equation}
\textbf{Type V:} \quad\quad 
\begin{cases} 
x_1\underset{1}\longrightarrow x_2+x_3 \\
x_2\underset{2}\longrightarrow x_1 + x_3\\
x_3\underset{3}\longrightarrow x_1+x_2
\end{cases}
\quad\quad \begin{pmatrix}
-1 & 1 & 1\\
1 & -1 & 1\\
1 & 1 & -1
\end{pmatrix}.
\end{equation}
For a better visualization, we prove the theorem using the
  correspondence of elementary cycles in $\king[\child]$ and the 
  permutation cycles as cycles established in
  Lemma~M\ref{lem:ElemPermCycles} above for type I-IV.\\
For type I, there is only one (permutation) cycle, 
\begin{equation}x_1 \rightarrow r_1 \rightarrow  x_2\rightarrow r_2 \rightarrow  2x_2,
\end{equation}
and thus the autocatalytic core is centralized with centers both $\{x_1, x_2\}$.\\
For type II$^*$, there are two (permutation) cycles:
\begin{equation}
  \begin{cases}
    x_1\rightarrow r_1  \rightarrow x_2\rightarrow r_2 \rightarrow x_3 \rightarrow  r_3\rightarrow x_1;\\
    x_1 \rightarrow r_1 \rightarrow x_3 \rightarrow r_3 \rightarrow x_1,
  \end{cases}   
\end{equation} with $\{x_1,x_3\}$ being both centers.\\
For type III, there are two (permutation) cycles:
\begin{equation}
  \begin{cases}
    x_1\rightarrow r_1  \rightarrow x_2\rightarrow r_2 \rightarrow x_1; \\
    x_1 \rightarrow r_1 \rightarrow x_3 \rightarrow r_3 \rightarrow x_1,
  \end{cases}   
\end{equation}
with $\{x_1\}$ being a center.\\
For type IV, there are three (permutation) cycles:
\begin{equation}
  \begin{cases}
    x_1\rightarrow r_1 \rightarrow  x_2 \rightarrow r_2 \rightarrow  x_1;\\  
    x_1\rightarrow r_1 \rightarrow  x_3\rightarrow r_3 \rightarrow x_1;\\ 
    x_1\rightarrow r_1 \rightarrow x_2 \rightarrow r_2 \rightarrow x_3
    \rightarrow r_3 \rightarrow x_1,
  \end{cases}
\end{equation}
with $\{x_1\}$ being a center.\\
For type V, there are five permutation cycles: 
\begin{equation}\
  \begin{cases}
x_1\rightarrow r_1 \rightarrow  x_2 \rightarrow r_2\rightarrow x_1;\\
    x_1\rightarrow r_1 \rightarrow x_3 \rightarrow r_3 \rightarrow x_1;\\
    x_2\rightarrow r_2 \rightarrow x_3\rightarrow r_3\rightarrow x_2;\\ 
    x_1\rightarrow r_1 \rightarrow x_2 \rightarrow r_2 \rightarrow x_3
    \rightarrow r_3 \rightarrow x_1;\\
    x_1\rightarrow r_1 \rightarrow x_3\rightarrow r_3 \rightarrow x_2
    \rightarrow r_2\rightarrow x_1,
  \end{cases}
\end{equation}
with no species being a center, see Fig.~\ref{fig:coreV}. One easily
  checks that those are also exactly the five elementary circuits in
  Fig.~\ref{fig:coreV}. Finally, an explicit and straightforward computation
shows the validity of Eq.~\eqref{eq:propcentralized} for all five
types. The straightforward generalization with different stoichiometric
coefficients and by addition of monomolecular intermediates, e.g., by
substituting $x_1 \rightarrow x_2 $ with $x_1\rightarrow
I_1\rightarrow\dots\rightarrow I_n\rightarrow x_2$, is omitted for
simplicity of presentation. For the latter, it suffices to say that any
argument based on the number of (permutation) cycles is indeed
insensitive to the addition of intermediates.

\end{proof}

\section{The Set System of Circuitnets of Fluffles}
\label{sect:flufflesetsystems}

Here, we collect some properties of the set system $\mathfrak{F}\subseteq
2^{\mathcal{C}}$ of circuitnets whose union form fluffles. This is of
interest because certain simple properties guarantee simple enumeration or
the existence of efficient algorithms to find maximal elements. For our
purposes, the following properties of the set system
$(X,\mathfrak{A})$ with basis set $X$ and $\mathfrak{A}\in
2^X$ are of interest:
\begin{itemize}
\item[(i)] $(X, \mathfrak{A})$ is \emph{accessible} if for all
  $A \in \mathfrak{A}, A \neq \emptyset:$ there is $a \in A$ such that
  $A\setminus \{a\} \in \mathfrak{A}$.
\item[(ii)] $(X, \mathfrak{A})$ is \emph{strongly accessible}
  \cite{boley_listing_2010} if it is accessible, and in addition, for any
  $A,B\in \mathfrak{A}$ with $A\subsetneq B$ there is $b\in B\setminus A$
  such that $A\cup \{b\}\in \mathfrak{A}$.
\item[(iii)] A strongly accessible set system $(X, \mathfrak{A})$ is called
  \emph{commutable} \cite{conte_listing_2019} if for any nonempty
  $A, B \in \mathfrak{A}$ and
  $a, b \in X: A \cup \{a\} \in \mathfrak{A}, A\cup \{b\} \in \mathfrak{A}$
  and $
  A \cup \{a,b\} \subseteq B$ implies
  $A \cup \{a,b\} \in \mathfrak{A}$.
\item[(iv)] A commutable set system $(X,\mathfrak{A})$ is called
  \emph{confluent} if for all $A,B,C\in \mathfrak{A}$ with
  $B\neq \emptyset$ that
  $B\subseteq A, B\subseteq C \Rightarrow A\cup C\in \mathfrak{A}$.
\item[(v)] $(X, \mathfrak{A})$ is an \emph{independence system} or
  hereditary, if $A\in\mathfrak{A}$ and $\emptyset \ne B\subseteq A$
  implies $B\in\mathfrak{A}$.
\end{itemize}
Note that confluence is not comparable to the other properties.
 
\begin{theorem}
  \label{thm:Fcommutable}
  $\mathfrak{F}$ is a commutable set system.
\end{theorem}
\begin{proof}
  If $\mathcal{C}\in\mathfrak{F}$ then there is a $C\in\mathcal{C}$ such
  that $\mathcal{C}'\coloneqq \mathcal{C}\setminus\{C\}$ is a again a
  circuitnet. By Lemma~M\ref{lem:subwoffle}, $\bigcup(\mathcal{C}')$ is
  again a fluffle, i.e., $\mathcal{C}'\in\mathfrak{F}$. That is,
  $\mathfrak{F}$ is accessible.

  Now suppose $\mathcal{C}',\mathcal{C}\in\mathfrak{F}$ and
  $\mathcal{C}'\subsetneq\mathcal{C}$. Then there is a circuit
  $C_1\in\mathcal{C}$ that shares an edge with a cycle
  $C_2\in\mathcal{C}\setminus\mathcal{C}'$, since otherwise $\mathcal{C}$
  cannot be a strong block. Thus $\mathcal{C}'\cup\{C_2\}$ is a strong
  block and Lemma~M\ref{lem:subwoffle} implies that 
  $\mathcal{C}'\cup\{C_2\}$ is a fluffle. Since $\mathfrak{F}$ is accessible, 
  it is also strongly accessible.
  
  Let $\mathcal{C},\mathcal{D}\in\mathfrak{F}$,
  $\mathcal{C}\subset\mathcal{D}, C_1,C_2\in
  \mathcal{D}\setminus\mathcal{C}$,
  $\mathcal{C}\cup\{C_1\}\in\mathfrak{F}$,
  $\mathcal{C}\cup\{C_2\}\in\mathfrak{F}$, and 
  $\mathcal{C}\cup \{C_1,C_2\}\subseteq \mathcal{D}$. Since 
  $\mathcal{C}\ne\emptyset$, the union of the circuits in $\mathcal{C}\cup\{C_1\}$ and
  $\mathcal{C}\cup\{C_2\}$ are two strong blocks that share a strong block,
  namely the union of the circuits in $\mathcal{C}$. Thus
  $\mathcal{C}'\coloneqq \mathcal{C}\cup\{C_1,C_2\}$ is also a strong
  block, and hence $\mathcal{C}'$ is a circuitnet. Since
  $\mathcal{C}'\subseteq \mathcal{D}\in\mathfrak{F}$, the union of the
  circuits is a fluffle by Lemma~M\ref{lem:subwoffle}, and thus
  $\mathcal{C}'\in\mathfrak{F}$. Together with strong accessibility, this
  implies that $\mathfrak{F}$ is a commutable set system
\end{proof}

The example in Fig.~\ref{fig:notconfl} shows that the set system of fluffles
$\mathfrak{F}$ is not confluent.

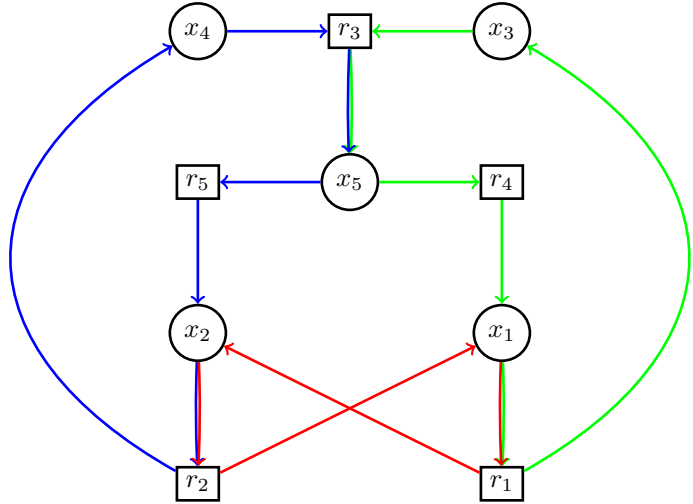
\begin{figure}
  \centering
  \begin{tikzpicture}[line width = 1pt]
    \node[draw, rectangle] (a) at (3,0) {$r_1$};     
    \node[draw, rectangle] (b) at (-1,0) {$r_2$};
    \node[draw, rectangle] (c) at (1,6) {$r_3$};
    \node[draw, rectangle] (e) at (3,4) {$r_4$};
    \node[draw, rectangle] (d) at (-1,4) {$r_5$};

    \node[draw, circle] (x) at (3,2) {$x_1$};    
    \node[draw, circle] (y) at (-1,2) {$x_2$};
    \node[draw, circle] (v) at (3,6) {$x_3$};
    \node[draw, circle] (u) at (-1,6) {$x_4$};
    \node[draw, circle] (w) at (1,4) {$x_5$};

    \draw[->, green] (a) to [bend right =60, looseness=1.5] (v);
    \draw[->, green] (v) -- (c);
    \draw[->, green] (c) to  [out = 272, in = 88] (w);
    \draw[->, green] (w) -- (e);
    \draw[->, green] (e) -- (x);
    \draw[->, green] (x) to [out = 272, in = 88] (a);
    
    \draw[->, blue] (u) -- (c);
    \draw[->, blue] (c) to [out = 268, in = 92] (w);
    \draw[->, blue] (w) -- (d);
    \draw[->, blue] (d) -- (y);
    \draw[->, blue] (y) to [out = 268, in = 92] (b);
    \draw[->, blue] (b) to [bend left = 60, looseness=1.5] (u);
    
    \draw[->, red] (a) -- (y);
    \draw[->, red] (y) to [out = 272, in = 88] (b);
    \draw[->, red] (b) -- (x);
    \draw[->, red] (x) to [out = 268, in = 92] (a);
    
  \end{tikzpicture}
  \caption{Counterexample to the hypothesis that $\mathfrak{F}$ is
    confluent. Consider the two child-selective elementary circuits
    $C_1=(r_1,x_3,r_3,x_5,r_4,x_1,r_1)$ (green) and
    $C_2=(r_2,x_4,r_3,x_5,r_5,x_2,r_2)$ (blue). Then $\{C_1,C_2\}$ is a
    circuitnet, but $\{C_1,C_2\} \notin \mathfrak{F}$ since $d_{in}(r_3) =
    2$ or $d_{out}(x_5) = 2$.  However, the circuit
    $C_3=(x_1,r_1,x_2,r_2,x_1)$ (red) is a fluffle. In addition,
    $\{C_1,C_3\},\{C_2,C_3\}\in\mathfrak{F}$, but
    $\{C_1,C_2,C_3\}\notin\mathfrak{F}$ since $\{C_1,C_2\} \notin
    \mathfrak{F}$.}
  \label{fig:notconfl}
\end{figure}

Instead of circuitnets, we use superpositions of the representatives of
fluffle CS equivalence classes with the representatives of the CS 
equivalence classes of elementary circuits. The following statement 
follows immediately from Lemmas~M\ref{lem:circuitnetequiv}, 
M\ref{lem:bumpeq_enum}, and Proposition~M\ref{prop:kingrep}:
\begin{corollary}
  Let $\mathcal{C}=\{C_1,\dots,C_h\}$ be a circuitnet for a fluffle
  $G=\bigcup(\mathcal{C})$. Then the representative $[G]_{\bumpeq}$ of its
  CS-equivalence class is 
  \begin{equation}
    \left[\bigcup(\mathcal{C})\right]_{\bumpeq}=\bigcup_{i=1}
         [C_i]_{\bumpeq}
         \label{eq:repsum}
  \end{equation}
\end{corollary}
For every circuitnet $\mathcal{C}\in\mathfrak{F}$ we therefore can define a
corresponding set of representatives $[\mathcal{C}]_{\bumpeq}:=\{ [C]\, |
\, C\in\mathcal{C}\}$. Note that some of the representatives in
$\mathcal{C}$ may be redundant. Now we consider the corresponding set system 
$[\mathfrak{F}]\coloneqq\{[\mathcal{C}]\,|\, \mathcal{C}\in\mathfrak{F}\}$.
As an immediate consequence of equ.(\ref{eq:repsum}), the arguments in the
proof of Thm.~\ref{thm:Fcommutable} carry over to $[\mathfrak{F}]$, and
thus we may conclude that
\begin{corollary}
  The set system of circuitnet representatives $[\mathfrak{F}]$ is
  commutable.
\end{corollary}
This observation provides a formal basis for the stepwise enumeration 
on the system of fluffle represen\-tatives.

In principle one could also consider the set system $(E,\mathfrak{X})$ with
$\mathfrak{X}\coloneqq\{ E_1(G)| G \text{ is fluffle in }\king\}$ on the
edge set $E$ of $\king$. It is easy to see, however, that
$(E,\mathfrak{X})$ is not an accessible set system, since 
the deletion of each of the $(r,x)$-edges may lead to a subgraph that is not
a strong block and thus also not a fluffle. Thus fluffles cannot be
generated efficiently by exploring individual edge additions.

\section{Additional Computational Data} \label{sec:addComp}

\paragraph{\emph{E.\ coli} core model}
This CRN comprises 72 metabolites and 95 reactions
\cite{orth_reconstruction_2010}. We excluded the following set of small,
highly connected molecules, which are of minor interest for autocatalysis:
cytosolic NAD, NADH, NADP, NADPH, AMP, ADP, H\textsuperscript{+},
H\textsubscript{2}O, CO\textsubscript{2}, coenzyme A, phosphate, oxygen,
ubiquinone, and ubiquinol. Since this model describes the central carbon
metabolism of \emph{E.\ coli}, we did not remove ATP as the major energy
carrier.

\begin{figure}[h]
  \centering
  \begin{minipage}[c]{0.5\textwidth}	
    \includegraphics[width=\textwidth]{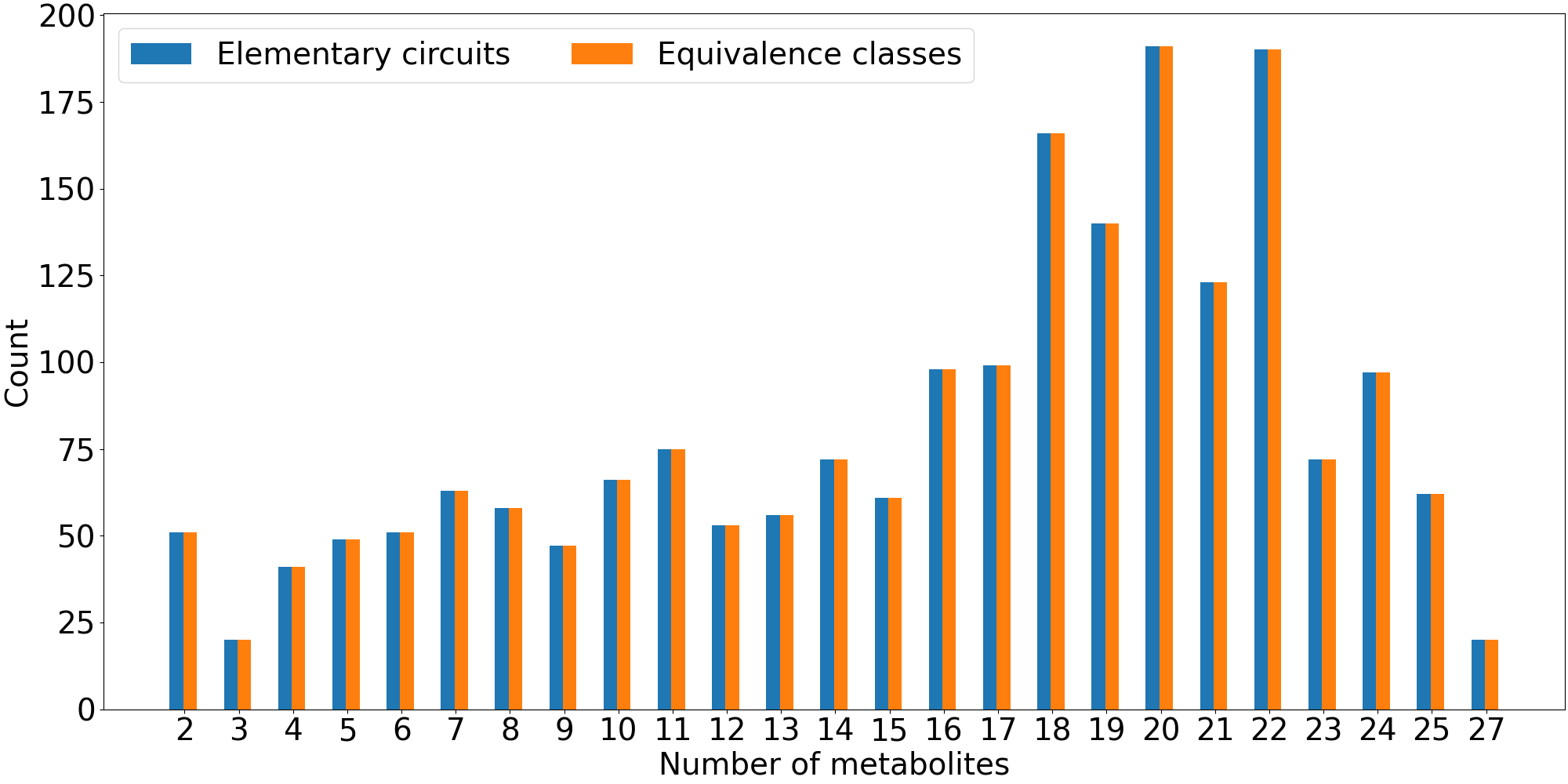}
  \end{minipage}%
  \hfill%
  \begin{minipage}[c]{0.5\textwidth}
    \includegraphics[width=\textwidth]{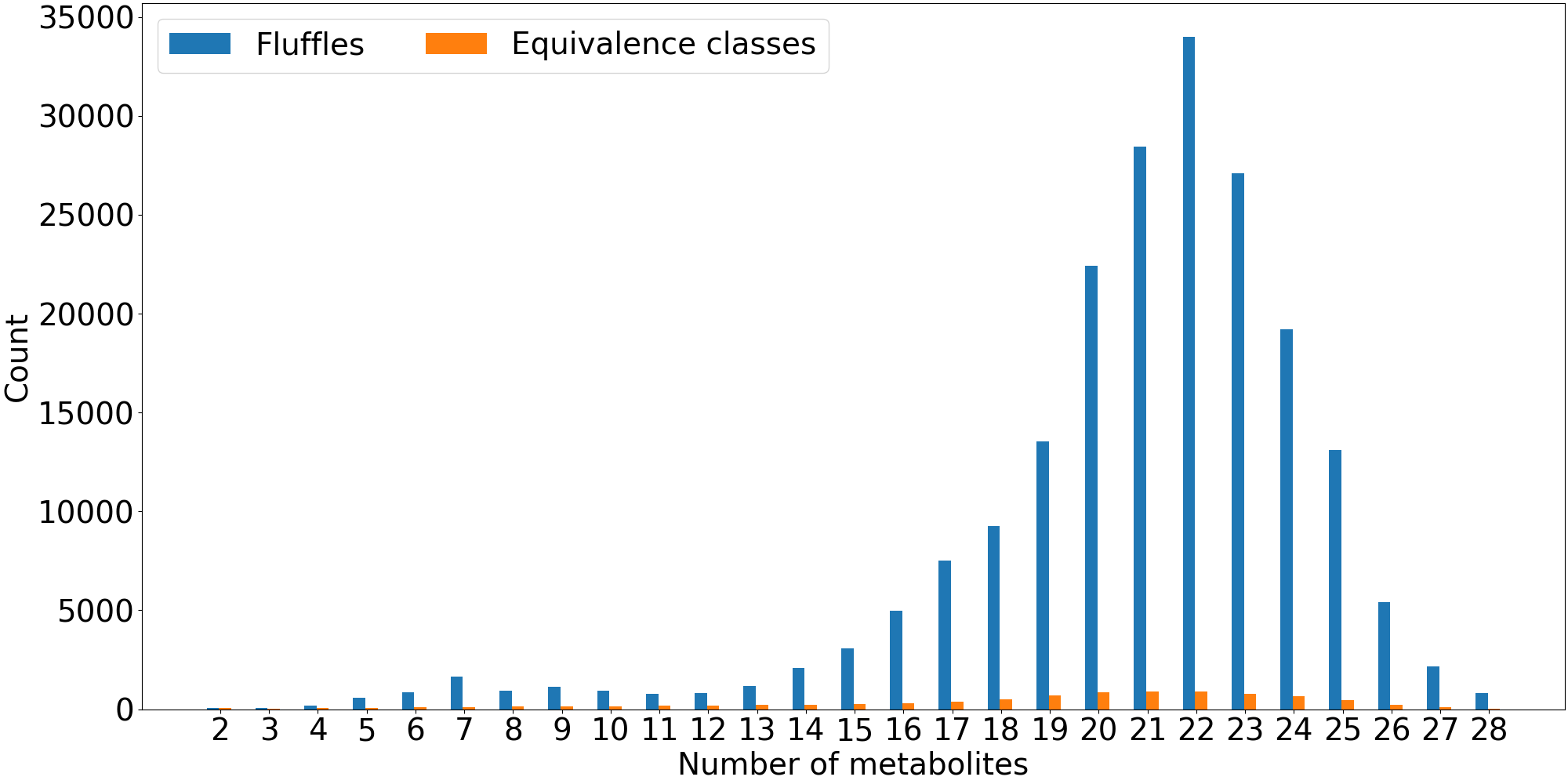}
  \end{minipage}
  \caption{Length distribution of elementary circuits (left) and size
    distribution of fluffles and their CS-equivalance classes
    (right) for the \textit{E.\ coli} core model.}
  \label{fig:EqClFluffles}
\end{figure}

Fig.~\ref{fig:EqClFluffles} in the main text summarizes the distribution
of elementary circuits and fluffles for the \textit{E.\ coli} core network
on a log scale. Here, we include the same data on a linear scale.

\begin{figure}
  \begin{minipage}[c]{0.4\textwidth}
    \begin{tikzpicture}[font=\footnotesize]
      \node[draw, circle] (x1) at (0,0) {F6P};
      \node[draw, circle] (x2) at (3,0) {G6P};
      \node[draw, circle] (x3) at (6,0) {6-PGL};
      \node[draw, circle] (x4) at (6,-3) {6-PGC};
      \node[draw, circle] (x5) at (6,-6) {Ru5P};
      \node[draw, circle] (x6) at (3,-6) {R5P};
      \node[draw, circle] (x7) at (0,-6) {S7P};
      \node[draw, circle] (x8) at (0,-3) {E4P};
			
      \node[draw, rectangle] (r1) at (1.5,0) {PGI};
      \node[draw, rectangle] (r2) at (4.5,0) {G6PDH};
      \node[draw, rectangle] (r3) at (6,-1.5) {PGL};
      \node[draw, rectangle] (r4) at (6,-4.5) {GND};
      \node[draw, rectangle] (r5) at (4.5,-6) {RPI};
      \node[draw, rectangle] (r6) at (1.5,-6) {TKT2};
      \node[draw, rectangle] (r7) at (3,-3) {TALA};
      \node[draw, rectangle] (r8) at (0,-1.5) {TKT1};
      
      \draw[->] (x1) to (r1);
      \draw[->] (r1) -- (x2);
      \draw[->] (x2) -- (r2);
      \draw[->] (r2) -- (x3);
      \draw[->] (x3) -- (r3);
      \draw[->] (r3) -- (x4);
      \draw[->] (x4) -- (r4);
      \draw[->] (r4) -- (x5);
      \draw[->] (x5) -- (r5);
      \draw[->] (r5) -- (x6);
      \draw[->] (x6) -- (r6);
      \draw[->] (r6) -- (x7);
      \draw[->] (x7) -- (r7);
      \draw[->] (r7) -- (x8);
      \draw[->] (r7) -- (x1);
      \draw[->] (x8) -- (r8);
      \draw[->] (r8) -- (x1);
    \end{tikzpicture}	
  \end{minipage}%
  \hfill%
  \begin{minipage}[c]{0.4\textwidth}
    \begin{tikzpicture}[font=\footnotesize]
      \node[draw, circle] (x1) at (0,0) {F6P};
      \node[draw, circle] (x2) at (4,0) {S7P};
      \node[draw, circle] (x3) at (4,-3) {R5P};
      \node[draw, circle] (x4) at (2,-5) {Ru5P};
      \node[draw, circle] (x5) at (0,-3) {Xu5P};
      
      \node[draw, rectangle] (r1) at (2,0) {TALA};
      \node[draw, rectangle] (r2) at (4,-1.5) {TKT1};
      \node[draw, rectangle] (r3) at (3,-4) {RPI};
      \node[draw, rectangle] (r4) at (1,-4) {RPE};
      \node[draw, rectangle] (r5) at (0,-1.5) {TKT2};
      
      \draw[->] (x1) to (r1);
      \draw[->] (r1) -- (x2);
      \draw[->] (x2) -- (r2);
      \draw[->] (r2) -- (x3);
      \draw[->] (r2) -- (x5);
      \draw[->] (x3) -- (r3);
      \draw[->] (r3) -- (x4);
      \draw[->] (x4) -- (r4);
      \draw[->] (r4) -- (x5);
      \draw[->] (x5) -- (r5);
      \draw[->] (r5) -- (x1);
      
    \end{tikzpicture}	
  \end{minipage}
  \caption{Two autocatalytic cores in the \textit{E.\ coli} core model were
    identified by our graph-theoretical algorithm, but not by the ILP
    formulation of Gagrani et al. \cite{gagrani_polyhedral_2024}}
  \label{fig:AddCores}
\end{figure}
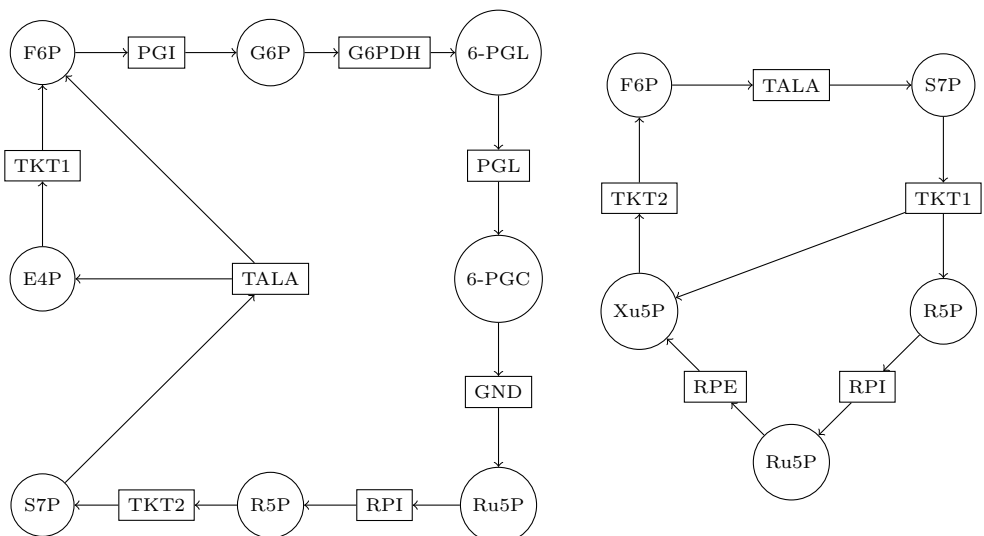

The comparison of our results with the ILP-formulation of Gagrani et al.
\cite{gagrani_polyhedral_2024} revealed two autocatalytic cores that were
detected by our graph-theoretical algorithm only, see
Fig.~\ref{fig:AddCores}.  Both are localized in the
Pentose-Phosphate-Pathway.

\paragraph{Large \emph{E.\ coli DH5$\alpha5$} model}
This CRN consists of $1,951$ metabolites and $2,779$ \cite{monk_multiomics_2016}.
Compared to the core model, the list of molecules that were removed was
augmented to allow for computational feasibility and meaningful
results. For brevity, we only provide the species identifier. The
assignment to the full names can be found at
\url{http://bigg.ucsd.edu/models/iEC1368\_DH5a/metabolites}.  The list
contains the metabolites from cytosol (c), periplasm (p), and extracellular space
(e): 
\par\noindent\begingroup\small M\_23camp\_p, M\_23ccmp\_p,
M\_23cgmp\_p, M\_23cump\_p, M\_2fe1s\_c, M\_2fe2s\_c, M\_35cgmp\_c,
\linebreak M\_3amp\_p, M\_3cmp\_p, M\_3fe4s\_c, M\_3gmp\_p, M\_3ump\_p, M\_4fe4s\_c,
M\_ACP\_c, M\_adp\_c, M\_alpp\_p, M\_amp\_c, M\_amp\_p, M\_apoACP\_c,
M\_arbtn\_e, M\_arbtn\_fe3\_e, M\_atp\_c, M\_btn\_c, M\_btnso\_c,\linebreak 
M\_ca2\_c, M\_ca2\_p, M\_camp\_c, M\_cdp\_c, M\_cl\_c, M\_cmp\_c,
M\_cmp\_p, M\_co2\_c, M\_co2\_p, M\_coa\_c, M\_colipa\_e, M\_cpgn\_e,
M\_cpgn\_un\_e, M\_ctp\_c, M\_cu2\_c, M\_cu2\_p, M\_cu\_p, M\_dadp\_c,
M\_damp\_c, M\_damp\_p, M\_datp\_c, M\_dcamp\_c, M\_dcdp\_c, M\_dcmp\_c,
M\_dcmp\_p, M\_dctp\_c, M\_dgdp\_c, \linebreak M\_dgmp\_c, M\_dgmp\_p, M\_dgtp\_c,
M\_didp\_c, M\_dimp\_c, M\_dimp\_p, M\_ditp\_c, M\_dnad\_c, \linebreak M\_dsbaox\_p,
M\_dsbard\_p, M\_dsbcox\_p, M\_dsbcrd\_p, M\_dsbdox\_c, M\_dsbdrd\_c,
M\_dsbgox\_p, \linebreak M\_dsbgrd\_p, M\_dtdp\_c, M\_dtmp\_c, M\_dtmp\_p, M\_dttp\_c,
M\_dudp\_c, M\_dump\_c, M\_dump\_p, \linebreak M\_dutp\_c, M\_enter\_e, M\_fad\_c,
M\_fadh2\_c', M\_fe2\_c, M\_fe2\_p, M\_fe3\_e, M\_fe3hox\_e,
M\_fe3hox\_un\_e, M\_fecrm\_e, M\_fecrm\_un\_e, M\_feenter\_e,
M\_feoxam\_e, M\_feoxam\_un\_e, M\_fmn\_c, M\_fmnh2\_c,\linebreak  M\_gdp\_c,
M\_gdp\_p, M\_gmp\_c, M\_gmp\_p, M\_gtp\_c, M\_gtp\_p, M\_h2\_c, M\_h2\_p,
M\_h2o2\_c, M\_h2o2\_p, M\_h2o\_c, M\_h2o\_e, M\_h2o\_p, M\_h2s\_c,
M\_h\_c, M\_h\_e, M\_h\_p, M\_hacolipa\_e, M\_halipa\_e, M\_hco3\_c,
M\_hdca\_e, M\_hqn\_c, M\_idp\_c, M\_imp\_c, M\_imp\_p, M\_itp\_c,
M\_lipa\_e, M\_lipidA\_e, M\_lipidAp\_e, \linebreak M\_metsox\_R\_\_L\_e,
M\_metsox\_S\_\_L\_e, M\_mql8\_c, M\_mqn8\_c, M\_n2o\_c, M\_na1\_c,
M\_na1\_p, \linebreak M\_nad\_c, M\_nadh\_c, M\_nadp\_c, M\_nadph\_c, M\_nh4\_c,
M\_nh4\_p, M\_nmn\_c, M\_nmn\_p, M\_no2\_c, M\_no2\_p, M\_no3\_c,
M\_no3\_p, M\_no\_c, M\_o2\_c, M\_o2\_p, M\_o2s\_c, M\_o2s\_p, M\_pi\_c,
M\_pi\_p, \linebreak M\_ppi\_c, M\_pppi\_c, M\_q8\_c, M\_q8h2\_c, M\_rbflvrd\_c,
M\_ribflv\_c, M\_s\_c, M\_sel\_c, M\_seln\_c, M\_slnt\_c, M\_so2\_c,
M\_so3\_c, M\_so3\_p, M\_so4\_c, M\_thm\_c, M\_thmmp\_c, M\_thmpp\_c,
M\_trnaala\_c, M\_trnaarg\_c, M\_trnaasn\_c, M\_trnaasp\_c, M\_trnacys\_c,
M\_trnagln\_c, M\_trnaglu\_c, M\_trnagly\_c, M\_trnahis\_c, \linebreak  M\_trnaile\_c,
M\_trnaleu\_c, M\_trnalys\_c, M\_trnamet\_c, M\_trnaphe\_c, M\_trnapro\_c,
M\_trnasecys\_c, \linebreak M\_trnaser\_c, M\_trnathr\_c, M\_trnatrp\_c,
M\_trnatyr\_c, M\_trnaval\_c, M\_tsul\_c, M\_tsul\_p, M\_udp\_c, \linebreak M\_ump\_c,
M\_ump\_p, M\_utp\_c.
\endgroup

\paragraph{\emph{Homo sapiens} erythrocyte model}
This CRN consists in total of $342$ metabolites and $469$
\cite{bordbar_iabrbc283_2011}. It was constructed by taking advantage
of the human RECON-1 metabolic model \cite{duarte_global_2007} and 
proteomic data from enucleated erythrocytes and covers two compartments 
only, the cytosol and the extracellular space. After removal of the below provided 
list of small molecules, $151$ metabolites and $261$ reactions in mainly two strongly
connected components remained. Again, we only provide the species
identifier. The assignment to the full names can be found at
\url{http://bigg.ucsd.edu/models/iAB_RBC_283/metabolites}. The list
contains the metabolites from the cytosol (c) and the extracellular space
(e): \par\noindent\begingroup\small M\_gdp\_c, M\_thmtp\_c, M\_nad\_c,
M\_ump\_c, M\_arg\_\_L\_e, M\_pi\_c, M\_3moxtyr\_e, M\_normete\_\_L\_e,
M\_cl\_c, M\_mal\_\_L\_e, M\_spmd\_e, M\_gluala\_e, M\_thmpp\_c, M\_thm\_e,
M\_imp\_c, M\_cdp\_c, M\_o2\_c, M\_band\_c, M\_utp\_c, M\_cl\_e,
M\_dnad\_c, M\_35cgmp\_c, M\_hco3\_c, M\_dopa\_e, M\_adp\_c, M\_na1\_e,
\linebreak M\_h\_c, M\_coa\_c, M\_ptrc\_e, M\_cmp\_c, M\_ala\_\_L\_e,
M\_nadp\_c, M\_nadh\_c, M\_k\_c, M\_ppi\_c, \linebreak M\_gmp\_c,
M\_nh4\_c, M\_co\_c, M\_ctp\_c, M\_k\_e, M\_bandmt\_c, M\_na1\_c,
M\_acnam\_e, M\_gtp\_c, \linebreak M\_nmn\_c, M\_camp\_c, M\_udp\_c,
M\_h2o\_c, M\_4pyrdx\_e, M\_mepi\_e, M\_h\_e, M\_ribflv\_c, M\_nrpphr\_e,
M\_h2o2\_c, M\_nadph\_c, M\_ca2\_c, M\_fad\_c, M\_ncam\_e, M\_ca2\_e,
M\_thmmp\_c, M\_thm\_c, M\_atp\_c, M\_amp\_c, M\_co2\_c, M\_pi\_e,
M\_fmn\_c, M\_gly\_e, M\_fe2\_c \endgroup

\paragraph{\emph{Methanosarcina Barkeri} model}
This CRN consists of $628$ metabolites and $690$ \cite{bordbar_iabrbc283_2011}.
Again, we only provide the species identifier. The assignment to the full names can 
be found at \url{http://bigg.ucsd.edu/models/iAF692/metabolites}. The list
contains metabolites from the cytosol (c) and the extracellular space (e):

 \par\noindent\begingroup\small 
M\_f420\_2\_c, M\_trnathr\_c, M\_h\_c, M\_h\_e, M\_f430p2\_c, M\_f420\_3\_c, 
M\_imp\_c, M\_s\_c, M\_cu2\_c, \linebreak M\_dctp\_c, M\_dtdp\_c,M\_trnaser\_c, 
M\_trnaarg\_c, M\_so3\_e, M\_trnaile\_c, M\_pppi\_c, M\_cobya\_c, \linebreak M\_mma\_e,
M\_ctp\_c, M\_ni2\_c, M\_dma\_e, M\_f420\_1\_c, M\_ala\_\_L\_e, M\_trnagly\_c,
M\_dcdp\_c,\linebreak  M\_tma\_e, M\_nmn\_c, M\_itp\_c, M\_h2\_c, M\_cd2\_e, M\_btn\_c,
M\_dcmp\_c, M\_dudp\_c, M\_no2\_c, \linebreak M\_cmp\_c, M\_tsul\_c, M\_dgtp\_c, 
M\_ch4\_e, M\_cobalt2\_c, M\_cbi\_e, M\_nh4\_c, M\_adp\_c, M\_n2\_e, 
\linebreak M\_nad\_c, M\_f420\_0\_c, M\_cd2\_c, M\_co2\_c, M\_dtmp\_c, M\_trnatrp\_c,
M\_trnalys\_c, M\_camp\_c, \linebreak M\_trnagln\_c, M\_ca2\_c, M\_k\_e, M\_h2s\_c, 
M\_f420\_5\_c, M\_trnaasp\_c, M\_mg2\_c, M\_co\_c, M\_f420\_4\_c, M\_pi\_c,
M\_dttp\_c, M\_f390a\_c, M\_o2\_c, M\_f420\_6\_c, M\_nadh\_c, M\_trnaala\_c,
M\_ind3ac\_e, M\_dgdp\_c, M\_h2o\_c, M\_cdp\_c, M\_f390g\_c, M\_fe2\_c, 
M\_meoh\_e, M\_com\_c, M\_dms\_e, M\_o2s\_c, M\_f420\_2h2\_c, M\_datp\_c, 
M\_cu2\_e, M\_cl\_e, M\_na1\_e, M\_hco3\_c, M\_so3\_c, M\_trnamet\_c, M\_pac\_e, 
M\_alac\_\_S\_e, M\_co1dam\_c, M\_dadp\_c, M\_gtp\_c, M\_trnaval\_c, M\_coa\_c, 
M\_nadp\_c, M\_thmpp\_c, M\_ppi\_c, \linebreak M\_f430p1\_c, M\_glyald\_e, M\_trnacys\_c,
 M\_fmn\_c, M\_thm\_c, M\_thmmp\_c, M\_f420\_7\_c, M\_dcamp\_c, M\_na1\_c, 
M\_nadph\_c,  M\_actn\_\_R\_e, M\_atp\_c, M\_dutp\_c, M\_cob\_c, M\_co2dam\_c, 
M\_dnad\_c, \linebreak M\_ump\_c, M\_cl\_c, M\_f430p3\_c, M\_damp\_c, M\_gmp\_c, M\_trnaglu\_c,
M\_k\_c, M\_gdp\_c, M\_idp\_c, \linebreak M\_s\_e, M\_trnaphe\_c, M\_f430\_c,  M\_cbl1hbi\_e,
M\_btn\_e, M\_ribflv\_c, M\_h2o2\_c, M\_udp\_c, M\_trnaleu\_c, M\_trnatyr\_c, M\_h2\_e,   
M\_trnahis\_c, M\_unknown\_rbfdeg\_e, M\_amp\_c, M\_ca2\_e, M\_dump\_c,  \linebreak
M\_unknown\_cbl1deg\_e, M\_ch4s\_e, M\_utp\_c, M\_trnapro\_c
\endgroup

\section{Examples}
\label{subsect:Examples}

\begin{example}[Autocatalytic core of Type III does not admit an elementary-circuit CS-representative]
\label{ex:notjustcycle}
The CS-equivalence class of the autocatalytic cores of Types I, II, IV, and
V contains a single circuit circuitnet; see the proof of
Thm.~M\ref{thm:NgheClasses} for a direct verification. In each of these
cases, the list of elementary circuits includes at least one (two for Type
V) circuit that passes through all species and reaction vertices. The only
exception is Type III, shown in Fig.~M\ref{fig:autoCore}g. Here, the
CS-matrix is
\begin{equation}
\SM[\child]=\begin{pmatrix}
-1 & 1 & 1\\
1 & -1 & 0\\
1 & 0 & -1
\end{pmatrix},
\end{equation}
and it admits exactly two elementary circuits:
\begin{equation}
\begin{cases}
x_1 \rightarrow r_1 \rightarrow x_2 \rightarrow r_2 \rightarrow x_1,\\
x_1 \rightarrow r_1 \rightarrow x_3 \rightarrow r_3 \rightarrow x_1,
\end{cases}
\end{equation}
neither of which traverses all vertices (both have length~2).
\end{example}

\begin{example}[Autocatalytic core Type IV in the \emph{E.\ Coli} core network] 
Most autocatalytic cores described in the literature are of types I, II, or
III. In the \emph{E.\ coli} core network, we found a single example of a
Type IV core in the pentose-phosphate-pathway (PPP). This example served as
motivation for introducing the concept of centralized autocatalysis.
 
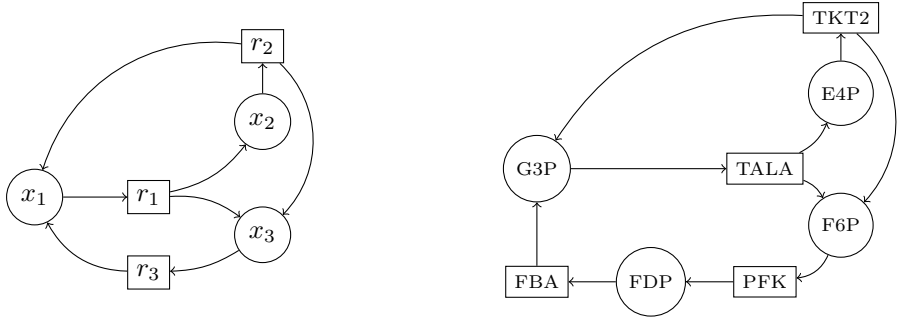
\begin{figure}[tb]
  \centering
  \begin{minipage}[c]{0.5\textwidth}
    \begin{tikzpicture}
      \node[draw, circle] (x1) at (0,0) {$x_1$};
      \node[draw, circle] (x2) at (3,1) {$x_2$};
      \node[draw, circle] (x3) at (3,-0.5) {$x_3$};		
      \node[draw, rectangle] (r1) at (1.5,0) {$r_1$};
      \node[draw, rectangle] (r2) at (3, 2) {$r_2$};
      \node[draw, rectangle] (r3) at (1.5,-1) {$r_3$};
      
      \draw[->] (x1) -- (r1);
      \draw[->] (r1) to [bend right = 20] (x2);
      \draw[->] (r1) to [bend left = 20] (x3);
      \draw[->] (x2) -- (r2);
      \draw[->] (r2) to [bend left = 45] (x3);
      \draw[->] (r2) to [bend right = 40] (x1);
      \draw[->] (x3) to [bend left = 15] (r3);
      \draw[->] (r3) to [bend left = 30] (x1);
    \end{tikzpicture}
  \end{minipage}%
  \hfill%
  \begin{minipage}[c]{0.5\textwidth}
    \begin{tikzpicture}[font=\footnotesize]
      \node[draw, circle] (x1) at (0,0) {G3P};
      \node[draw, circle] (x2) at (4, 1) {E4P};
      \node[draw, circle] (x3) at (4,-0.75) {F6P};	
      \node[draw, circle] (x4) at (1.5,-1.5) {FDP};
      
      \node[draw, rectangle] (r1) at (3,0) {TALA};
      \node[draw, rectangle] (r2) at (4, 2) {TKT2};
      \node[draw, rectangle] (r3) at (3,-1.5) {PFK};
      \node[draw, rectangle] (r4) at (0,-1.5) {FBA};		
      
      \draw[->] (x1) -- (r1);
      \draw[->] (r1) to [bend right = 20] (x2);
      \draw[->] (r1) to [bend left = 20] (x3);
      \draw[->] (x2) -- (r2);
      \draw[->] (r2) to [bend left = 45] (x3);
      \draw[->] (r2) to [bend right = 30] (x1);
      \draw[->] (x3) to [bend left = 30] (r3);
      \draw[->] (r3) to (x4);
      \draw[->] (x4) to (r4);
      \draw[->] (r4) to (x1);
    \end{tikzpicture}
  \end{minipage}
  \caption{Example of an autocatalytic core of Type IV according to the
    classification of Blokhuis et al. \cite{blokhuis_universal_2020} (left) and a
    topologically equivalent autocatalytic core detected in the
    pentose-phosphate-pathway (PPP) of the \textit{E.\ coli} core network
    (right). Abbreviations of metabolites: G3P glyceraldehyde 3-phosphate;
    E4P erythrose 4-phosphate; F6P fructose 6-phosphate; FDP fructose
    1,6-bisphosphatase. Reactions are labeled by the enzymes that catalyze
    them: TALA transaldolase A; TKT2 transketolase 2; PFK
    phosphofructokinase; FBA fructose-bisphosphate aldolase.}
  \label{fig:TypeIV}
\end{figure}

The Type IV autocatalytic core in Fig.~\ref{fig:TypeIV}
introduces one unit of G3P, which yields one unit of E4P and F6P each. E4P then
generates one unit of G3P and one unit of F6P. Each of the two units of F6P
finally produces a G3P, resulting in a 
gross yield of three G3P. There are
indeed three elementary circuits in this network that contain G3P, while
all other species are located on at most two elementary circuits. G3P
therefore differs from the residual species and, since all elementary
circuits coalesce in G3P, it forms the autocatalytic center.
\end{example} 

\FloatBarrier

\section{Algorithmic Overview}\label{sec:Algorithms}

The algorithm constitutes five main parts, computing the set of all autocatalytic cycles and their properties from the König graph of a CRN:
\begin{itemize}
\item preparing the network
\item decomposition of the network to biochemically functional units
\item enumeration of elementary circuits
\item enumeration of equivalence classes of fluffles
\item classification of enumerated fluffles
\end{itemize}
The individual steps will be described in the following sections.  

\begin{algorithm}[htb]
  \caption{Algorithm overview}
  \label{alg:AlgorithmOverview} 
    \SetKwInOut{Require}{Require}
    \SetKwInOut{Output}{Output}

    \Require{$\king(X,R)$, set of small molecules to remove $S$}
    \Output {$\mathcal{A}$ - Set of all autocatalytic Metzler matrices} 
    
    $\mathcal{A} \leftarrow \emptyset$\;
    $\mathcal{A}_Z \leftarrow \emptyset$\;
    $\king(X,R) \leftarrow$ \FuncSty{RemoveSmallMolecules}($S, \king(X,R)$)\ \tcp*{see Sec. \emph{\nameref{sec:network_prep}}}
    \For{SCC of $\king(X,R)$}{
    $\tree \leftarrow$ \FuncSty{PartitionNetwork}($\king[SCC], (V,\leq)$)\ \tcp*{see Sec. \emph{\nameref{sec:Partitioning}}}
    $\elem \leftarrow$ \FuncSty{EnumerateElementaryCircuits}($\tree$)\ \tcp*{see Sec. \emph{\nameref{sec:enum_elem_circuits}}}
    $\mathcal{E} \leftarrow $ \FuncSty{EquivalenceClassAssembly}($\elem$)\ \tcp*{see Alg.~M\ref{alg:RecEqClAss}}
    $\mathcal{A}\leftarrow \mathcal{A} \ \cup $ \FuncSty{AutocatalyticActivity}($\mathcal{E}$)\ \tcp*{see Sec. \emph{\nameref{sec:auto_cap}}}
    $\mathcal{A}_{Z} \leftarrow $ \FuncSty{CheckCentrality}($\mathcal{A}$)\ \tcp*{see Sec. M\ \textit{\nameref{sec:centAut}}}
    }
\end{algorithm}
    
An overview of all components is given in
Alg. \ref{alg:AlgorithmOverview}. The enumeration of equivalence classes of
fluffles is described in the main text in Alg.~M\ref{alg:RecEqClAss} and
\FuncSty{CheckCentrality} follows
Sec.~M\ \textit{\nameref{sec:centAut}}. The remaining 
components will be covered in the following sections.

\subsection{Network Preparation}
\label{sec:network_prep}

Highly interconnected metabolites that do not constitute a focal compound
of a reaction, i.e., co-factors like ATP, NADH, etc., and small molecules
such as CO\textsubscript{2} and H\textsubscript{2}O, do not contribute to
the generation of chemically meaningful autocatalytic cycles. They are
therefore removed to reduce complexity. We provide a manually curated list
of small molecules for this purpose. Otherwise, the K{\"o}nig graph of the
input CRN is not modified and follows standard definitions.

\subsection{Partitioning}\label{sec:Partitioning}

\begin{algorithm}[htb]
  \caption{PartitionNetwork}
  \label{alg:PartitionNetwork} 

  \SetKwInOut{Require}{Require}
  \SetKwInOut{Output}{Output}
  \SetKwProg{Def}{def}{:}{} 

  \Require{$\king(X,R)$}
  \Output {Partitioning tree $\tree$}  
   
  $\tree \leftarrow$ ($\king(X,R)$, $\emptyset$)\;
  \FuncSty{PartitionNetworkRecursion}($\king(X,R), \tree$)\;
  
  \vspace*{6pt}
  \Def{\FuncSty{PartitionNetworkRecursion}($\king(X,R)$, shared $\tree$)}{
        $\mathcal{R} \leftarrow$ \FuncSty{GenerateReactionNetwork}($\king(X,R)$)\;  
        $ShReD\leftarrow $ \FuncSty{ComputeActualShReDMatrix}($\mathcal{R}$)\;
        $P\leftarrow $ \FuncSty{ComputeExpectedShReDMatrix}($\mathcal{R}$)\;
        $G\coloneqq P-ShReD$\;
        $v\leftarrow$ \FuncSty{ComputeLeadingEigenvector}(G)\;
        $s\in \{-1,1\}^n,s_i\coloneqq \begin{cases}
          -1 &\text{if } v_i \leq 0\\
          1 & \text{else }
        \end{cases}$\;
        \If{$v = 0$}{
          $s\in \{-1,1\}^n,s_i\coloneqq \begin{cases}
          -1 &\text{if } v_i < 0\\
          1 & \text{else }
          \end{cases}$\;
        }
        \If{$\sum_i \sum_j Q_{ij}s_is_j\leq 0$}{
          \Return
        }
        $\king(X_1,R_1), \king(X_2,R_2) \leftarrow$ \FuncSty{SplitNetwork}($v,\king(X,R)$)\;
        \If{$\king(X_1,R_1)$ or $\king(X_2,R_2)$ is a DAG}{
            \Return
        }
        $V(\tree) \leftarrow V(\tree) \cup \king(X_1,R_1) \cup \king(X_2,R_2)$\;
        $E(\tree) \leftarrow E(\tree) \cup {(\king(X,R), \king(X_1,R_1))}$\;
        $E(\tree) \leftarrow E(\tree) \cup {(\king(X,R), \king(X_2,R_2))}$\;
        \FuncSty{PartitionNetwork}($\king(X_1,R_1), \tree$)\;
        \FuncSty{PartitionNetwork}($\king(X_2,R_2), \tree$)\;
        
   }
\end{algorithm}

The partition algorithm takes the König graph of a CRN (without small molecules) as input and computes a partition tree $\mathbb{T}$ whose nodes are labeled by subnetworks in König representation. The resulting partition tree $\mathbb{T}$ is then used to determine interfaces between modules.

To ensure that all elements of $\mathcal{E}$ can be generated via linear
combinations of elementary circuits, we treat strongly connected components
$S$ independently. For every $\king[S]$, a network $\mathcal{R}\coloneqq
(R, E)$ with the reactions as vertices and $E\coloneqq \{(r_1, r_2) |
\exists x \in X(S), s_{xr_1}^{+}>0, s_{xr_2}^->0\}$ is generated. The next
step partitions $\mathcal{R}$ based on a round trip distance metric, called
\emph{Shortest Retroactive Distance} (ShReD)
\cite{sridharan_identification_2011}. The difference between the expected
and actual \emph{ShReD} matrices, i.e. $G\coloneqq P-\text{ShReD}$, is
employed to solve an integer linear programming (ILP) problem: $\max
Q\coloneqq \sum_{i=1}^n\sum_{j=1}^n G_{ij}\cdot s_{i}\cdot s_j, \text{
  s.t. }s\in\{-1,1\}^{n}$. By construction, $G$ is symmetric. Thus all
eigenvalues are real, and for the sake of reduced runtime, we take advantage of
the fact that the leading eigenvector of $G$, i.e., the eigenvector to the
largest eigenvalue, approximates the solution vector for the given
optimization problem as proposed in \cite{newman_modularity_2006}. More
details can be found in \cite{sridharan_identification_2011}. Importantly,
each partitioning step yields two submodules
$\king(X_1,R_1),\king(X_2,R_2)$, where $R=R_1\cup R_2, R_1\cap
R_2=\emptyset$ and $X_i \coloneqq \{x\in X \ \vert \ \exists r\in R_i:
s_{xr}^->0 \text{ or } s_{xr}^+>0\}$. A partitioning tree $\mathbb{T}$ is
constructed such that $(\king(X,R), \king(X_1,R_1)), (\king(X,R),
\king(X_2,R_2))\in E(\mathbb{T})$. The upper steps are repeated recursively until
either $Q=0$ or one of the submodules contained is a DAG.

For more details on the functions \FuncSty{GenerateReactionNetwork()}, 
\FuncSty{Compute\-Actual\-ShReD\-Matrix()},
\FuncSty{Compute\-Expected\-ShReD\-Matrix()}, and \FuncSty{SplitNetwork()} we
refer to the original publication of the implemented partition algorithm
\cite{sridharan_identification_2011}.

\subsection{Enumeration of Elementary circuits}
\label{sec:enum_elem_circuits}

The enumeration of elementary circuits follows the partition tree
$\mathbb{T}$ of the last step from bottom to top. First, Johnson's
algorithm \cite{johnson_finding_1975} is applied to all leaf nodes, which
ensures the detection of all elementary circuits within biochemical
functional modules. Upon merging two modules, i.e., for interior nodes of
$\mathbb{T}$, we restrict Johnson's algorithm to metabolites in the
intersection that lie along directed paths from one child module into the
another, see Fig.~\ref{fig:EnumerateElementaryCircuits}. It should be
noted that while we expect that exhaustive enumeration of all circuits will
always be possible, enumeration for joined partitions may become infeasible
for larger networks, and is therefore size-limited in practice.

\begin{algorithm}[htb]
  \caption{OrientedNetwork}
  \label{alg:OrientedNetwork} 
  
    \SetKwInOut{Require}{Require}
    \SetKwInOut{Output}{Output}  
  
    \Require{Root, OutNetwork, InNetwork, L}
    \Output{G}
    $G \leftarrow$ Root\;
    \For{$u \in L$}{
          $V(G) \leftarrow V(G) \cup \{u_{in}, u_{out}\}$\;
          \For{$v\in V^{G}_{in}(u)$}{
             \uIf{$v\in V(InNetwork)$}{
               $E(G) \leftarrow E(G)\cup \{(v, u_{in})\}$\;
             } \uElseIf{$v\in V(OutNetwork)$}{
               $E(G) \leftarrow E(G) \cup \{(v, u_{out})\}$\;
             }
          }
          \For{$v \in V^{G}_{out}(u)$}{
             \uIf{$v\in V(InNetwork)$}{
                $E(G) \leftarrow E(G)\cup \{(u_{in}, v)\}$\;
             }\uElseIf{$v\in V(OutNetwork)$}{
                $E(G) \leftarrow E(G) \cup \{(u_{out}, v)\}$\;
             }
             $V(G) \leftarrow V(G)\setminus \{u\}$\;
          }
   }
   \Return{G}
\end{algorithm}

\begin{algorithm}[htb]
  \caption{EnumerateElementaryCircuits}
  \label{alg:EnumerateElementaryCircuits} 
  
      \SetKwInOut{Require}{Require}
      \SetKwInOut{Output}{Output}
    
      \Require{Partitioning tree $\tree$}
      \Output{Elementary Circuits $\elem$}

      \SetKwFunction{FuncEECR}{EnumElemCircuitsRec}      
      \SetKwProg{Def}{def}{:}{}

      $\elem \leftarrow \emptyset$\;
      $t \leftarrow \mathrm{root}(\tree)$\;
      \FuncSty{EnumElemCircuitsRec}($t, \elem$)\;
      \vspace*{10pt}
      \Def{ \FuncSty{EnumElemCircuitsRec}(tree node $t \in V(\tree)$, shared $\elem$)}{
        \eIf{t == Leaf}{
            $\elem \leftarrow \elem \ \cup $ \FuncSty{Johnsons}($t$, $\{\}$)
        }{
            $t_1, t_2 = \mathrm{children}(t)$\;
            \FuncSty{EnumElemCircuitsRec}($t_1$)\;           
            \FuncSty{EnumElemCircuitsRec}($t_2$)\;  
            $Y\leftarrow X(t_1)\cap X(t_2)$\;
        	$L \leftarrow \{\}$\;
        	\If{$Y\neq \emptyset$}{
            \For{$y\in Y$}{
                LIRO $\leftarrow$ \FuncSty{OrientedNetwork}($t$, $t_1$, $t_2$, L)\tcp*{see Alg. \ref{alg:OrientedNetwork}, proof Thm. \ref{thm:AllElemCircuits}}
                RILO $\leftarrow$ \FuncSty{OrientedNetwork}($t$, $t_2$, $t_1$, L)\; 
                $\elem'_1 \leftarrow $ \FuncSty{Johnsons}(LIRO, $\{y\}$)\;
                $\elem'_2 \leftarrow $ \FuncSty{Johnsons}(RILO, $\{y\}$)\;
                $\elem_1 \leftarrow$ \FuncSty{Backtranslation}($\elem'_1$)\tcp*{see proof Thm. \ref{thm:AllElemCircuits}}
                $\elem_2 \leftarrow$ \FuncSty{Backtranslation}($\elem'_2$)\;
                $\elem \leftarrow \elem \cup \elem_1 \cup \elem_2$\;
                $L\leftarrow L\cup \{y\}$\;
             }
        	}
         }
      }
\end{algorithm}

\begin{theorem}\label{thm:AllElemCircuits}
  Let $(X,R)$ be a CRN, $\elem(\king)$ be the set of all elementary 
  circuits of its K\"onig graph, and $\mathcal{C}$ the set of elementary circuits
  generated via Algorithm \ref{alg:EnumerateElementaryCircuits}. Then
  $\elem(\king) = \mathcal{C}$.
\end{theorem}
\begin{proof}
  We consider each strongly connected component independently since there
  are no elementary circuits connecting two strongly connected components
  by definition. By construction, $\tree$ is a strict binary tree, leading
  to a simple bottom enumeration scheme where a node is visited only after
  the full sub-trees of both children have been visited.  Note that from a
  purely algorithmic perspective, nodes may be visited in arbitrary order,
  as circuit sets for each node are independent. However, we require an
  order of closed sets of nodes for proof by induction. Leaf nodes serve
  as the base case, representing the minimal sub-networks of the
  partition. Here, the enumeration of elementary circuits is achieved by
  Johnson's Algorithm \cite{johnson_finding_1975, gupta_finding_2021},
  which has shown to be complete. In the inductive step, we only consider
  inner tree nodes; therefore, as the partition tree is strict, parent
  nodes $\king(X_\kappa, R_\kappa)$ (node) with non-empty children,
  $\king(X_1,R_1)$ (left) and $\king(X_2,R_2)$ (right).
  
  We first note that by construction, $R_\kappa=R_1\cup R_2$ and
  $X_\kappa=X_1\cup X_2$. While $R_1\cap R_2=\emptyset$, $X_1\cap X_2$ is
  not necessarily empty. Thus, for fusing two children, we consider only
  elementary circuits containing at least one intersecting metabolite. If
  $X_1\cap X_2=\emptyset$, there is nothing to do since there are no edges
  connecting $\king(X_1, R_1)$ and $\king(X_2, R_2)$. In any other case, we
  enumerate elementary circuits containing at least one compound of
  $X_1\cap X_2$. Several algorithms have been proposed to enumerate
  circuits containing a fixed node, such as modifications of Johnsons'
  Algorithm \cite{johnson_finding_1975}, with available pre-existing
  implementations \cite{SciPyProceedings_11}. Simply applying Johnson on
  each node $y_i \in Y = X_1\cap X_2$ yields a superset of desired
  circuits, as we also enumerate subsets of elementary circuits of
  $\king(X_1,R_1)$ and $\king(X_2,R_2)$, which have already been enumerated
  by induction. This is unproblematic from a purely mathematical
  standpoint, but would drastically increase runtime complexity.  To avoid
  this duplicate enumeration, we need to additionally enforce the inclusion
  of reactions from both $R_1$ and $R_2$ in each circuit. We say a circuit
  crosses child borders, implying that an intersecting metabolite is a
  product for a reaction from the left child and a reactant for a reaction
  from the right child or vice versa.  More formally, a border node $y_i
  \in Y \cap C$ of circuit $C$ is product of a unique reaction $r_i^{+} \in
  C$, connected by an ingoing edge $(r_i^{+},y_i) \in E(C)$, and educt of
  unique reaction $r_i^{-} \in C$, connected by an outgoing edge
  $(y_i,r_i^{-}) \in E(C)$. If $r_i^{+}, r_i^{-} \in R_1$ or $r_i^{+}, r_i^{-}
  \in R_2$ for all $y_i \in Y \cap X(C)$, the circuit $C$ includes no
  crossing and was therefore already enumerated by induction. Otherwise
  there is at least one $y_i \in Y \cap X(C)$ s.t. $r_i^{+} \in R_1$ and
  $r_i^{-} \in R_2$ or vice versa.  This leads to the following procedure
  (Fig.~\ref{fig:EnumerateElementaryCircuits}): For each vertex $y_i \in
  Y$ we consider two antidromic networks left-in-right-out (LIRO) and
  right-in-left-out (RILO), derived by removing all outgoing edges
  (respectively ingoing) in $R_1$ and all ingoing (respectively
  outgoing) in $R_2$ adjacent to $y_i$. Calling Johnson on $y_i$ on both
  networks now yields exactly the set of desired circuits containing
  $y_i$. Therefore, in the $i$-th iteration, all elementary circuits
  containing $y_i$ are enumerated crossing $y_i$ from the left to the right
  child or from the right to the left child, avoiding exactly all circuits
  not crossing at $y_i$. Again, completeness of the Johnsons algorithm
  guarantees finding all circuits. However, completeness also implies that
  circuits crossing metabolites on multiple points $Y'\subseteq Y \cap
  X(C)$ will be enumerated exactly $|Y'|$ times without further algorithmic
  restrictions. Given a fixed node order in $Y$, on $i$-th iteration any
  circuit with $y_i\in Y'$ was already enumerated iff there exists a
  $y_k\in Y': k<i$. To restrict re-enumeration, we split nodes $y_k\in Y':
  k<i$ into novel nodes $y_{k(in)}$ and $y_{k(out)}$ s.t. $y_{k(in)}$
  inherits edges in $\king(X_1,R_1)$ and $y_{k(out)}$ edges in
  $\king(X_2, R_2)$, thereby exactly disabling crossing, but not inclusion
  of previous metabolites. \FuncSty{Backtranslation} reverses this
  operation on the cycle node and edge sets. We note that elementary
  circuits containing both versions of an intersecting metabolite,
  $y_{i(in)}$ and $y_{i(out)}$, such as the blue circuit depicted in
  Fig.~\ref{fig:EnumerateElementaryCircuits} ($y_3$), are removed by
  backtranslation. It follows that all elementary circuits for the parent,
  $\king(X_\kappa,R_\kappa)$, are exactly enumerated without duplication.
  
  By induction, we conclude that all elementary circuits for
  $\king(X,R)$ are enumerated.
  \end{proof}
 
\begin{figure}[htb]
  \centering
  \includegraphics[width=\textwidth]{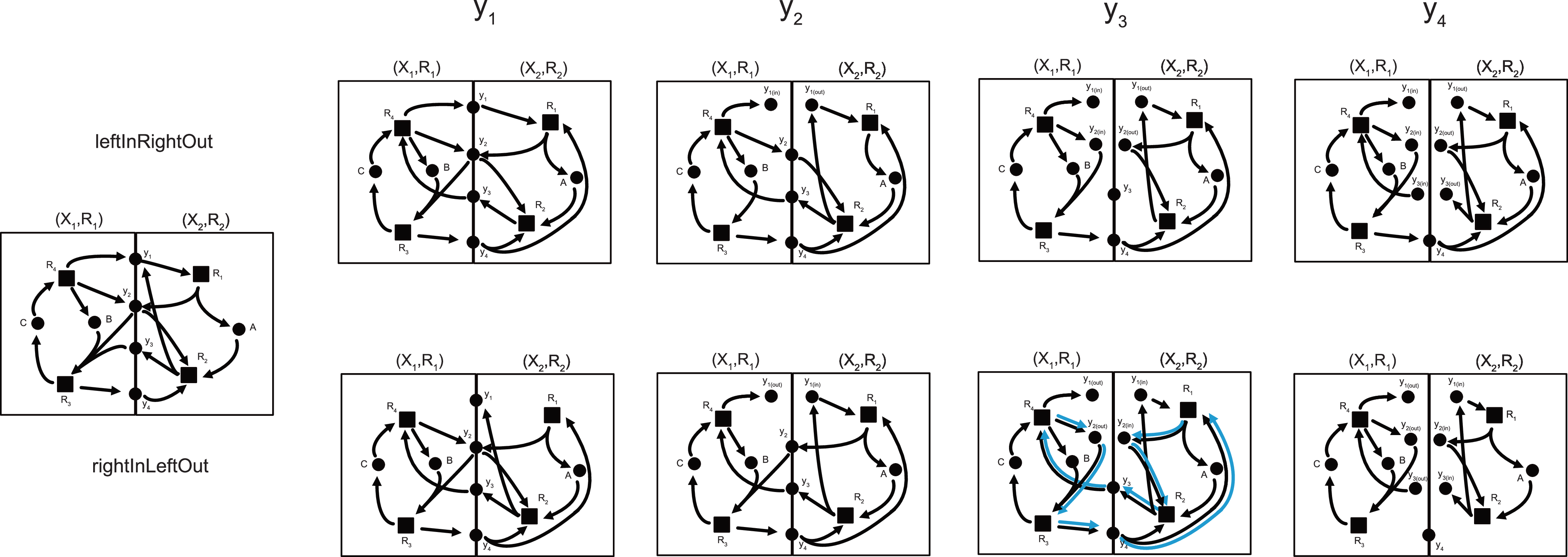}
  \caption{Depiction of an example for the fusion of two vertices of the
    partitioning tree $\tree$ with four intersecting metabolites,
    i.e. $Y\coloneqq X_1\cap X_2=\{y_1,y_2,y_3,y_4\}$. Iterative steps for
    intersecting metabolites are depicted in the four right panels. The
    upper and lower panel are illustrating the two antidromic oriented
    networks, left-in-right-out (LIRO) and right-in-left-out (RILO), for
    each iteration of intersecting metabolites. In detail, in the $i$-th
    iteration and LIRO orientation ($i$-th upper panel), $y_i$ exhibits
    incoming and outgoing edges only from the left and right child,
    respectively, while for RILO orientation, $y_i$ receives only incoming
    and outgoing edges from the right and the left child, respectively. In
    the $i+1$-th iteration for both orientations, LIRO and RILO, $y_i$ is
    split into $y_{i(in)}$ and $y_{i(out)}$. Subsequently, $y_{i(in)}$ is
    added to the left and $y_{i(out)}$ to the right network (LIRO) or
    vice-versa (RILO). In addition, edges pointing into $y_i$ are strictly
    retained in the appurtenant oriented network, i.e., $y_{i(in)}$ is only
    incident to edges with origin or target in the in-oriented network. In
    the lower third panel, an elementary circuit is depicted in blue, which
    contains both versions of the original vertex. Nevertheless, after
    re-translation into the original vertex sequence, it is not an elementary
    circuit. Thus, it is not considered to be checked for its
    autocatalytic capacity or for the assembly of larger cycles.}
  \label{fig:EnumerateElementaryCircuits}
\end{figure}

\subsection{Autocatalytic capacity}
\label{sec:auto_cap}

\begin{algorithm}[htb]
  \caption{AutocatalyticCapacity}
  \label{alg:DetermineAutocatalyticCapacity} 

    \SetKwInOut{Require}{Require}
    \SetKwInOut{Output}{Output}  

    \Require{$\mathcal{E}$, set of equivalence classes}
    \Output{$\mathcal{A}$, set of autocatalytic matrices}     
    
    $\mathcal{A} \leftarrow \emptyset$\;
    \For{$E_1(C) \in \mathcal{E}$}{
       $X(C), R(C) \leftarrow $ \FuncSty{SplitVertices}($E_1(C)$)\;
       $n \leftarrow |X(C)|$\;
       $A \leftarrow \SM[X(C), \kappa_C(R(C))]$ \tcp*{Compute Matrix from Graph}
      \eIf(\tcp*[f]{Case: Is Metzler Matrix}){$A ==\Metzler{A}$}{  
         $r_{\max} \leftarrow \max\{\text{Real}(\lambda) \mid \lambda \in \mathrm{spectrum}(A)\}$\;
         \If(\tcp*[f]{If Hurwitz unstable}){$r_{\max} > 0$}{ 
              $\mathcal{A} \leftarrow \mathcal{A} \ \cup \{A\}$\;
          }
       }(\tcp*[f]{Case: Is non-Metzler Matrix}){  
         \If{$\exists v\in \mathbb{R}_{>0}^{|n|}: Av>0$}{
            $\mathcal{A} \leftarrow \mathcal{A} \ \cup \{A\}$\;
         }
       }
	 }
\end{algorithm}

The autocatalytic capacity of a CS matrix can be determined for Metzler and non-Metzler matrices utilizing different methods. In the Metzler case, spectral properties are of great value, while for non-Metzler matrices, an optimization problem needs to be solved. Consider an $n\times n$ irreducible, Metzler matrix $A$. If $A$ possesses any eigenvalue with a positive real part, in particular the leading one, then $A$ is autocatalytic (see Lemma M\ref{cor:SingularMetzlerHurwitz}). For non-Metzler matrices, autocatalytic capacity is determined by the existence of a positive vector $v>0$, s.t. $Av>0$. This allows for a relatively straightforward implementation as it is shown in the pseudocode presented in Algorithm~\ref{alg:DetermineAutocatalyticCapacity}.


\end{document}